\documentclass[twocolumn]{svjour3}           %
\usepackage{amsmath}
\usepackage{amsfonts}
\usepackage{amssymb}
\usepackage{ifpdf}
\usepackage{wrapfig}
\usepackage{cite}
\usepackage{comment}
\usepackage{appendix}
\usepackage{bm}
\usepackage{subcaption}
\usepackage{url}
\usepackage[font=small]{caption}
\usepackage{harpoon}
\usepackage[table]{xcolor}

\newcommand{\subsubsubsection}[1]{\paragraph{#1}\mbox{}\\}
\def\arxiv{1}

\ifx\arxiv\undefined
\else
\def\makeheadbox{\relax}
\fi

\definecolor{lightlightgray}{rgb}{.9,.9,.9}
\newcolumntype{g}{>{\columncolor{lightgray}}c}

\newcommand{\vect}[1]{\mathbf{\overrightharp{#1}}}

\ifpdf
  \usepackage[pdftex]{epsfig}
  \usepackage[pdftex]{hyperref}
\else
    \usepackage[dvips]{epsfig}
    \newcommand{\href}[2]{#2}
\fi

\newcommand{\carpet}{{\boldsymbol S}}
\newcommand{\grout}{\texttt{gro\-u\-t}}
\newcommand{\initializer}{\texttt{ini\-t\-i\-a\-l\-i\-z\-e\-r}}
\newcommand{\indicating}{\texttt{ind\-ica\-tor}}

\newcommand{\starter}{\texttt{sta\-rt\--ga\-d\-get}}

\newcommand{\crawler}{\texttt{c\-ra\-w\-l\-e\-r}}

\newcommand{\stagebinding}{\texttt{st\-a\-g\-e-bin\-ding}}
\newcommand{\fractal}[2]{\ensuremath{#1^{#2}}}
\newcommand{\fracasm}[3]{\ensuremath{#1^{#2}_{#3}}}

\newcommand{\Z}{\mathbb{Z}}
\newcommand{\N}{\mathbb{N}}

\newcommand{\termasm}[1]{\mathcal{A}_{\Box}[\mathcal{#1}]}
\newcommand{\prodasm}[1]{\mathcal{A}[\mathcal{#1}]}
\newcommand{\calT}{\mathcal{T}}
\newcommand{\bX}{{\boldsymbol X}}

\begin{document}

\title{Self-Assembly of $4$-sided Fractals in the Two-handed Tile Assembly Model}

\author{
 Jacob Hendricks%
    \thanks{Department of Computer Science and Information Systems, University of Wisconsin - River Falls,
    \protect\url{jacob.hendricks@uwrf.edu}}
\and
 Joseph Opseth%
    \thanks{Department of Mathematics, University of Wisconsin - River Falls,
    \protect\url{joseph.opseth@my.uwrf.edu}}
}

\institute{}

\date{}

\maketitle

\begin{abstract}
We consider the self-assembly of fractals in one of the most well-studied models of tile based self-assembling systems known as the Two-handed Tile Assembly Model (2HAM). In particular, we focus our attention on a class of fractals called discrete self-similar fractals (a class of fractals that includes the discrete Sierpi\'{n}ski carpet). We present a 2HAM system that finitely self-assembles the discrete Sierpi\'{n}ski carpet with scale factor $1$. 
Moreover, the 2HAM system that we give lends itself to being generalized and we describe how this system can be modified to obtain a 2HAM system that finitely self-assembles one of any fractal from an infinite set of fractals which we call \emph{$4$-sided fractals}. 
The 2HAM systems we give in this paper are the first examples of systems that finitely self-assemble discrete self-similar fractals at scale factor $1$ in a purely growth model of self-assembly.
Finally, we show that there exists a \emph{$3$-sided fractal} (which is not a tree fractal) that cannot be finitely self-assembled by any 2HAM system.
\end{abstract}

\section{Introduction}

The study of fractals has both a mathematical and a practical basis, as patterns similar to these recursively self-similar patterns occur in nature in the form of circulatory systems and branch patterns.
Evidently many fractals found in nature are the result of a process where a simple set of rules dictating how individual basic components (such as individual molecules) interact to yield larger complexes with recursive self-similar structure. One approach to understanding this process is to model such a process with artificial self-assembling systems. 

One of the first and also one of the most studied mathematical models of self-assembling systems is Winfree's abstract Tile Assembly Model (aTAM)~\cite{Winf98} where individual autonomous components are represented as tiles with glues on their edges. The aTAM was intended to model DNA tile self-assembly, where tiles are implemented using DNA molecules.
In the context of DNA tile self-assembly, there have been two main reasons for considering the self-assembly of fractals. First, in~\cite{FujHarParWinMur07} and~\cite{RothTriangles}, DNA-based tiles are used to self-assemble the Sierpi\'{n}ski triangle, showing the potential for DNA tile self-assembly to be used for the controlled formation of complex nanoscale structures. Second, there are many proposed theoretical models (and generalizations of these models) of DNA tile self-assembly (see~\cite{Winf98, AGKS05g, jSignals, JonoskaKarpenkoSignals, OneTile, RNaseSODA2010, DotKarMas10, SFTSAFT, KS07} for some examples). While mathematical notions of simulation relations between systems in such models continue to further elucidate how these various models relate~\cite{Versus, j2HAMIU, IUNeedsCoop, IUSA, DirectedNotIU, WoodsMeunierSTOC}, many ``benchmark'' problems have also been introduced. These benchmarks include the efficient self-assembly of squares and/or general shapes~\cite{RNAPods, SummersTemp, RotWin00, SingleNegative}, the capacity to perform universal computation~\cite{Polyominoes, Polygons, jSignals, SingleNegative, RTAM, CookFuSch11, LSAT1}, and the self-assembly of fractals~\cite{jSADSSF, KL09, TreeFractals, jKautzShutters, RoPaWi04, LutzShutters12, STAM-fractals}. In addition to providing a way of benchmarking models of self-assembly, studying the self-assembly of fractals has the potential to lead to new techniques for the design self-assembling systems.

When considering the self-assembly of discrete self-similar fractals (dssf's) such as the Sierpi\'{n}ski triangle one can consider either ``strict'' self-assembly, wherein a shape is made by placing tiles only within the domain of the shape, or ``weak'' self-assembly where a pattern representing the shape forms as part of a complex of tiles that contains specially labeled tiles corresponding to points in the shape and possibly additional tiles not corresponding to points of the shape. 
Previous work (including~\cite{jSADSSF, KL09, TreeFractals, jKautzShutters, RoPaWi04, LutzShutters12}) has shown the difficulty of strict self-assembly of dssf's in the aTAM as no nontrivial dssf has been shown to self-assemble in the strict sense.
In fact, the Sierpi\'{n}ski triangle is known to be impossible to self-assemble in the aTAM~\cite{jSSADST}; though it is possible to design systems which ``approximate'' the strict self-assembly of fractals~\cite{jSSADST,LutzShutters12,jSADSSF}. Interestingly, it is unknown whether there exists a dssf which strictly self-assembles in the aTAM. This includes the Sierpi\'{n}ski carpet dssf.
In this paper, we consider 2HAM systems which ``finitely'' self-assemble dssf's. Finite self-assembly was defined in~\cite{Versus} to study 2HAM systems that self-assemble infinite shapes (e.g. dssf's). Intuitively, a shape $S$, finitely self-assembles in a tile assembly system if any finite producible assembly of the system can always continue to self-assemble into the shape $S$ and the shape of any finite producible assembly is a subshape of $S$. See~\cite{MHAM,Versus,HFractals} for results which use the definition of finite self-assembly.

While the aTAM models single tile attachment at a time (or step in the self-assembly process), a more generalized model and another of the most studied models of self-assembly called the 2-Handed Assembly Model \cite{AGKS05g} (2HAM, a.k.a. Hierarchical Assembly Model) allows pairs of large assemblies to bind together. 
Given the hierarchical nature of the self-assembly process modeled by the 2HAM, we consider employing this process to finitely self-assemble dssf's. In~\cite{MHAM} it is shown that the Sierpi\'{n}ski carpet finitely self-assembles in the 2HAM at temperature $2$, but with scale factor $3$. That is, instead of finitely self-assembling a structure with tiles corresponding to the points of the Sierpi\'{n}ski carpet, the structure that self-assembles contains a $3$ by $3$ block of tiles that corresponds to a single point of the Sierpi\'{n}ski carpet. Here we show that not only does the Sierpi\'{n}ski carpet finitely self-assemble with scale factor $1$, but an infinite class of fractals, which we call the $4$-sided fractals, finitely self-assembles at temperature $2$ in the 2HAM with scale factor $1$. Intuitively, $4$-sided fractals are fractals that have a generator (the set of points in the first stage of the fractal) such that the generator is connected and consists of a rectangle of points ``on the boundary'' of the generator as well as points ``inside'' this rectangle. In other words, a $4$-sided fractal is a fractal with a generator that contains all $4$ sides and one can define $0$, $1$, $2$, and $3$-sided fractals analogously. (Definitions are given in Section~\ref{sec:prelims}.) Moreover, we show that there exists a $3$-sided fractal that cannot be finitely self-assembled by any 2HAM system at any temperature. %

Theorem~\ref{thm:four-sided} implies that one of the most well-known dssf's (the Sierpi\'{n}ski carpet) finitely self-assembles in one of the simplest and most studied models of self-assembly, the 2HAM.  It should be noted that in~\cite{STAM-fractals} it is shown that \emph{any} dssf can finitely self-assemble in the Signal-passing Tile Assembly Model (STAM) where tiles can change state and even disassociate from an existing assembly, ``breaking'' an assembly into two disconnected assemblies. That is, given any dssf, there is a STAM system that finitely self-assembles this fractal. Additionally, in~\cite{STAM-fractals} it is shown that a large class of fractals finitely self-assembles in the STAM even with temperature restricted to $1$. In a model similar to the STAM, the Active Tile Assembly Model~\cite{JonoskaSignals1}, infinite, self-similar substitution tiling patterns which fill the plane have been shown to assemble~\cite{JonoskaSignals2}. This may be considered a testament to the power of active tiles. Here we show that it is still possible to finitely self-assemble an infinite class of fractals in the 2HAM even though tiles are not active and disassociation is not allowed. 

\section{Preliminaries}\label{sec:prelims}

Here we provide informal descriptions of the 2HAM. For more details see~\cite{AGKS05g,PatitzSurvey, Versus}. 
Definitions and notation in Section~\ref{sec:2ham-def} are based on definitions from~\cite{Versus, UniversalShapeReplicators, MHAM}. Similar definitions and notation for the 2HAM can also be found in~\cite{j2HAMIU, j2HAMSim, NegativeGluesShapes}. We restate the definitions in the context of this paper for the sake of completeness and convenience. Likewise, in Section~\ref{sec:fractal-defs}, we also give the definition of discrete self-similar fractals similar to the definitions found in~\cite{TreeFractals} and~\cite{STAM-fractals}.

\subsection{Informal description of the 2HAM}\label{sec:2ham-def}

Let $U_2 = \{ (0,1), (0,-1), (1,0), (-1,0) \}$ be the set of all unit vectors in $\Z^2$. A \emph{grid graph} is a graph $G = (V, E)$ such that $V\subseteq \Z^2$, and for any edge $\{\vect{a}, \vect{b}\}\in E$, $\vect{a} - \vect{b} \in U_2$.

\subsubsection{Tile types, tiles, and supertiles}

A \emph{tile type} is a unit square with $4$ well defined sides that each correspond to a vector in $U_2$ such that each side of the square has an associated \emph{glue}. A glue is defined by a \emph{label} and a \emph{strength}. A glue label is a string of symbols over some fixed alphabet, and a glue strength is a non-negative integer. Moreover, a tile type has an associated string of symbols over some fixed alphabet called a \emph{label}. A \emph{positioned tile} is a pair consisting of a tile type and a point in $\Z^2$ called a \emph{tile location}. A \emph{tile} is the set of all translations in $\Z^2$ of a positioned tile. We refer to the side of a tile type (or tile) corresponding to $(0,1)$, $(0,-1)$, $(1,0)$, or $(-1,0)$ as the north, south, east, or west edge of the tile type (or tile) respectively.

Let $T$ be a finite set of tile types. A \emph{positioned supertile over $T$} is a set of positioned tiles with tile types in $T$ such that the positioned tiles have distinct tile locations. For a positioned supertile $A$ over $T$, we let $|A|$ denote the cardinality of $A$. A \emph{supertile over $T$} is the set of all translations of a positioned supertile over $T$. For a positioned supertile $A$, note that cardinality is invariant under translation. Therefore, for a supertile $\alpha$ over $T$, we let $|\alpha|$ denote the cardinality of any positioned supertile in $\alpha$ and note that this is well-defined. When $T$ is clear from context, we will shorten the phrase ``supertile over $T$'' to simply ``supertile''. For two adjacent tiles $t_1$ and $t_2$ in a positioned supertile over $T$ and $s$ in $\N$, we say that $t_1$ and $t_2$ \emph{interact with strength $s$} if the glues on their abutting sides are equal\footnote{glue labels are equal and glue strengths are equal} and these glues have non-zero strengths equal to $s$. 

Let $A$ be a positioned supertile over $T$. The \emph{binding graph of $A$} is the weighted undirected grid graph $G = (V,E)$ such that 1) $V$ is the set of all tile locations of tiles in $A$, and 2) $E$ is the set of all unordered pairs of vertices $v_1$ and $v_2$ in $V\times V$ with weight $w\in \N$ such that the two tiles in $A$ with tile locations equal to $v_1$ and $v_2$ interact with strength $w$. Note that a binding graph is a grid graph. For a non-negative integer $\tau$, $A$ is \emph{$\tau$-stable} if for every cut $C$ of the binding graph of $A$, the sum of the weights of the edges in the cut-set of $C$ is greater than or equal to $\tau$. A supertile $\alpha$ over $T$ is $\tau$-stable if it contains a positioned supertile over $T$ that is $\tau$-stable. Note that if $A$ is $\tau$-stable, then any translation of $A$ is $\tau$-stable. Therefore, the notion of $\tau$-stable for supertiles is well-defined. 

Let $A$, $B$, and $C$ be positioned supertiles over $T$ such that $A$ and $B$ are $\tau$-stable. We say that $A$ and $B$ are \emph{$\tau$-combinable into $C$} if $C = A\cup B$ and $C$ is $\tau$-stable. Moreover, let $\alpha$, $\beta$, and $\gamma$ be supertiles over $T$ such that $\alpha$ and $\beta$ are $\tau$-stable. $\alpha$ and $\beta$ are \emph{$\tau$-combinable into $\gamma$} if there exists $A$ in $\alpha$, $B$ in $\beta$, and $C$ in $\gamma$ such that $A$ and $B$ are $\tau$-combinable into $C$. Note that if $\alpha$ and $\beta$ are $\tau$-combinable into $\gamma$, then $\gamma$ is $\tau$-stable. We also define the \emph{subassembly} relation between two supertiles as follows. For supertiles $\alpha$ and $\beta$, $\alpha$ is a subassembly of $\beta$ provided that there exist positioned supertiles $A$ in $\alpha$ and $B$ in $\beta$ such that $A\subseteq B$.

\subsubsection{Tile assembly systems and assembly sequences}

A \emph{tile assembly system} (TAS) in the 2HAM is defined to be an ordered pair $\mathcal{T} = (T, \tau)$ such that $T$ is a finite set of tile types, and $\tau$ is a positive integer which we call the \emph{temperature} of $\mathcal{T}$. Let $\mathcal{T}=(T,\tau)$ be a TAS. A \emph{state} $S$ is a multiset of $\tau$-stable supertiles over $T$ such that the multiplicity of any supertile in $S$ is in $\N\cup \{\infty \}$. Let $S_0$ and $S_1$ be states. $S_0$ \emph{transitions} to $S_1$ at temperature $\tau$ if 1) there exists a supertile $\gamma$ such that $S_1 = S_0 \cup \{\gamma\}$, and 2) there exists $\alpha$ and $\beta$ in $S_0$ such that $\alpha$ and $\beta$ are $\tau$-combinable into $\gamma$. 

Let $k$ be in $\N\cup \{\infty\}$. A \emph{state sequence} of $\mathcal{T}$ is a sequence of states $\vec{S} = \langle S_i \rangle_{i = 0}^k$ such that for all $i$, $S_i$ transitions to $S_{i+1}$. A state sequence is called \emph{nascent} if $S_0$ is the multiset consisting of infinitely copies of tiles, one tile for each tile type in $T$. For a producible supertile $\alpha$, an \emph{assembly sequence} for $\alpha$ is a sequence of supertiles $\vec{\alpha} = \langle \alpha_i \rangle_{i = 0}^k$ such that there exists a state sequence $\vec{S} = \langle S_i \rangle_{i = 0}^k$ such that for all $i < k$, $\alpha_i \in S_i$ and there exists a supertile $\beta_i \in S_i$ such that $\alpha_i$ and $\beta_i$ are $\tau$-combinable into $\alpha_{i+1}$. Such an assembly sequence is called \emph{nascent} if $\vec{S}$ is nascent. The \emph{result} of an assembly sequence $\vec{\alpha} = \langle \alpha_i \rangle_{i = 0}^k$ is the unique supertile $\rho$ such that there exists $R\in \rho$ and $A_i \in \alpha_i$ such that $R=\cup_{0\leq i<k} A_i$, and for each $i$ such that $0\leq i < k$, $\alpha_i$ is a subassembly of $\rho$.  

\subsubsection{Producible supertiles and shapes}

Given a TAS $\mathcal{T}=(T,\tau)$, a supertile is \emph{producible} if it is the result of a nascent assembly sequence. A producible supertile $\alpha$ is \emph{terminal} if for any producible supertile $\beta$ there does not exist a $\tau$-stable supertile $\gamma$ such that $\alpha$ and $\beta$ are $\tau$-combinable into $\gamma$. We refer to the set of producible supertiles for $\mathcal{T}$ as $\prodasm{T}$ and the set of terminal supertiles for $\mathcal{T}$ as $\termasm{T}$. 

We refer to a set of points in $\Z^2$ as a \emph{shape}. For a shape $X$, a supertile $\alpha$ \emph{has shape $X$} if there exists a positioned supertile $A$ in $\alpha$ such that the set of tile locations of positioned tiles in $A$ is equal to $X$. Given a TAS $\mathcal{T}=(T,\tau)$, for an infinite shape $X\subseteq \Z^2$, we say that $\mathcal{T}$ \emph{finitely self-assembles} $X$ if for every finite producible supertile $\alpha$ of $\mathcal{T}$, $\alpha$ has the shape of a subset of points in $X$, and there exist an assembly sequence $\vec{\alpha} = \langle \alpha_i \rangle_{i = 0}^{\infty}$ such that $\alpha_0 = \alpha$ and the result of $\vec{\alpha}$ has shape $X$. In this paper we consider finite self-assembly of dssf's. 

\subsection{Discrete Self-Similar Fractals}\label{sec:fractal-defs}

In order to state the main theorem, we need to provide a few definitions. The definition of a discrete self-similar fractals and some of the notation used here also appears in~\cite{jSADSSF, TreeFractals, STAM-fractals}. First we introduce some notation.

Given $V\subseteq \Z^2$, the \emph{full grid graph} of $V$ is the undirected graph $G_V^f = (V, E)$, such that for all $\vec{x}, \vec{y} \in V$, $\{\vec{x}, \vec{y}\}\in E$ iff $||\vec{x} - \vec{y}|| = 1$.

Let $\N_g$ denote the subset $\{ 0, \dots, g-1\}$ of $\N$, and let $\N^2_g = \N_g \times \N_g$. For $g\in \N$ and $G\subseteq \N_g^2$, let $l_G$, $r_G$, $b_G$, and $t_G$ denote the integers: $l_G = \min_{(x,y)\in G} x$, $r_G = \max_{(x,y)\in G} x$, $b_G = \min_{(x,y)\in G} y$, and $t_G = \max_{(x,y)\in G} y$.
Moreover, let $w_G = r_G-l_G + 1$ and $h_G = t_G-b_G + 1$ denote the \emph{width} and \emph{height} of $G$ respectively.
Finally, let $L_G = \{(l_G, y) \mid b_G \leq y \leq t_G \}$, $R_G = \{(r_G, y) \mid b_G \leq y \leq t_G \}$, $T_G = \{(x, t_G) \mid l_G \leq x \leq r_G \}$, and $B_G = \{(x, b_G) \mid l_G \leq x \leq r_G \}$. In other words, $L_G$, $R_G$, $T_G$, and $B_G$ are the sets of points belonging to left, right, top, and bottom line segments of a ``bounding box'' of $G$.
Finally, if $A$ and $B$ are subsets of $\N^2$ and $(x,y)\in \N^2$, then $A+(x,y)B = \{(x_a,y_a) + (x\cdot x_b, y\cdot y_b) \mid (x_a, y_a)\in A\text{ and } (x_b,y_b)\in B\}$.
First we give the definition of a discrete self-similar fractal. 

\begin{definition}\label{def:dssf}
Let $\bX \subset \N^2$. We say that $\bX$ is a \emph{discrete self-similar fractal} (or \emph{dssf} for short), if there is a set $\{ (0,0) \} \subset G \subset \N^{2}_{g}$ with at least one point in every row and column, such that 
\begin{enumerate}
\item the full grid-graph of $G$ is connected, 
\item $w_G > 1$ and $h_G > 1$,
\item $G \subsetneq \N_{w_G}\times \N_{h_G}$, and
\item ${\boldsymbol X} = \cup^{\infty}_{i=1} X_i$, where \fractal{X}{i}, the \emph{$i^{th}$ stage} of ${\boldsymbol X}$, is defined by $\fractal{X}{1} = G$ and $\fractal{X}{i+1} = \fractal{X}{i} +(w_G^i, h_G^i)G$. 
\end{enumerate} 
Moreover, we say that $G$ is the \emph{generator} of ${\boldsymbol X}$. 
\end{definition}

A connected discrete self-similar fractal is one in which every component is connected in every stage, i.e. there is only one connected component in the grid graph formed by the points of the shape.

\begin{definition}\label{def:n-sided}
Let $n\in \{0,1,2,3,4\}$, $1 < g \in \N$ and $\bX \subset \N^2$. 
We say that $\bX$ is a \emph{$n$-sided fractal} if $\bX$ is a discrete self-similar fractal with generator $G$ such that: 
\begin{enumerate}
\item the full grid graph of $G$ is connected,
\item $S\cap G = S$ for at least $n$ distinct sets $S$ in \\ $\{L_G,R_G, T_G,B_G\}$.
\end{enumerate}

\end{definition}

Intuitively, the second condition in Definition~\ref{def:n-sided} is saying that the fractal generator contains all points of at least $n$ of the left, right, top, and bottom line segments of a ``bounding box'' of $G$. In particular, the generator of a $4$-sided fractal contains all of the points along the left, right, top, and bottom ``sides'' of the fractal generator.   
Finally, for a fractal $\bX$ with generator $G$, an enumeration of the points in a generator $G = \{\vec{v}_i\}_{i=1}^{|G|}$, and $j\in \N$, the stages of $\bX$ are $\fractal{X}{1} = G$ and $\fractal{X}{j+1} = \fractal{X}{j} + (w_G^j, h_G^j)G$. For $i\in \N$ such that $1\leq i \leq |G|$, we call the points of the $j+1$ stage given by $X_j + (w_G^j, h_G^j)\vec{v}_i$ the \emph{$j^{th}$ stage at position $i$}. For dssf $\bX$ and $i\in\N$ such that $i\geq 1$, we let $\fractal{X}{i}$ denote the $i^{th}$ stage of $\bX$.

\section{Self-assembly of Four Sided Fractals}\label{sec:four-sided-fractals}

In this section we show how to finitely self-assemble the class of $4$-sided discrete self-similar fractals in the 2HAM with scale factor of $1$ (i.e. no scaling is required). The most well-known example of a $4$-sided fractal is the Sierpi\'{n}ski carpet.

\begin{theorem}\label{thm:four-sided}
Let $\bX$ be a $4$-sided fractal. Then, there exists a 2HAM TAS $\calT_{\bX} = (T, 2)$ that finitely self-assembles $\bX$. Moreover, if $G$ is the generator for $\bX$ and $|G| = g$, $|T|$ is $O(g^3)$.
\end{theorem}

We build intuition for a construction showing Theorem~\ref{thm:four-sided} by showing that the Sierpi\'{n}ski carpet finitely self-assembles in the 2HAM at scale factor $1$. We then describe the modifications needed to extend the construction for the Sierpi\'{n}ski carpet to all $4$-sided fractals. 

\subsection{The Sierpi\'{n}ski carpet construction overview}\label{sec:carpet}

The Sierpi\'{n}ski carpet dssf is the dssf with generator $G = \{(0,2),$ $(1,2),$ $(2,2),$ $(0,1),$ $(1,2),$ $(0,0),$ $(1,0),$ $(2,0)\}.$ Figure~\ref{fig:carpet1} depicts this generator, while Figures~\ref{fig:carpet2} and~\ref{fig:carpet3} depict the second and third stages of the dssf respectively. We denote the carpet dssf by $\carpet$ and for $i\in \N$, we denote the $i^{th}$ stage of $\carpet$ as \fractal{S}{i}. We enumerate the points of \fractal{S}{1} as depicted in Figure~\ref{fig:carpet1} and use this enumeration to reference the positions of some substage within a subsequent stage of the carpet. 

\begin{figure}[htp]
    \centering
    \begin{subfigure}[b]{0.14\textwidth}
    		\centering
        \includegraphics[width=.9in]{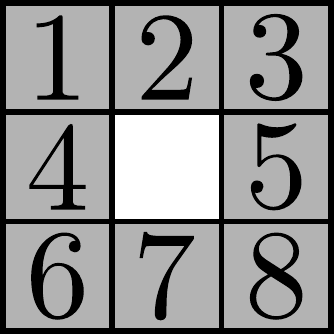}
        \caption{Stage $1$}
        \label{fig:carpet1}
    \end{subfigure}
    \quad\quad\quad\quad
    \begin{subfigure}[b]{0.14\textwidth}
		\centering        
        \includegraphics[width=1.1in]{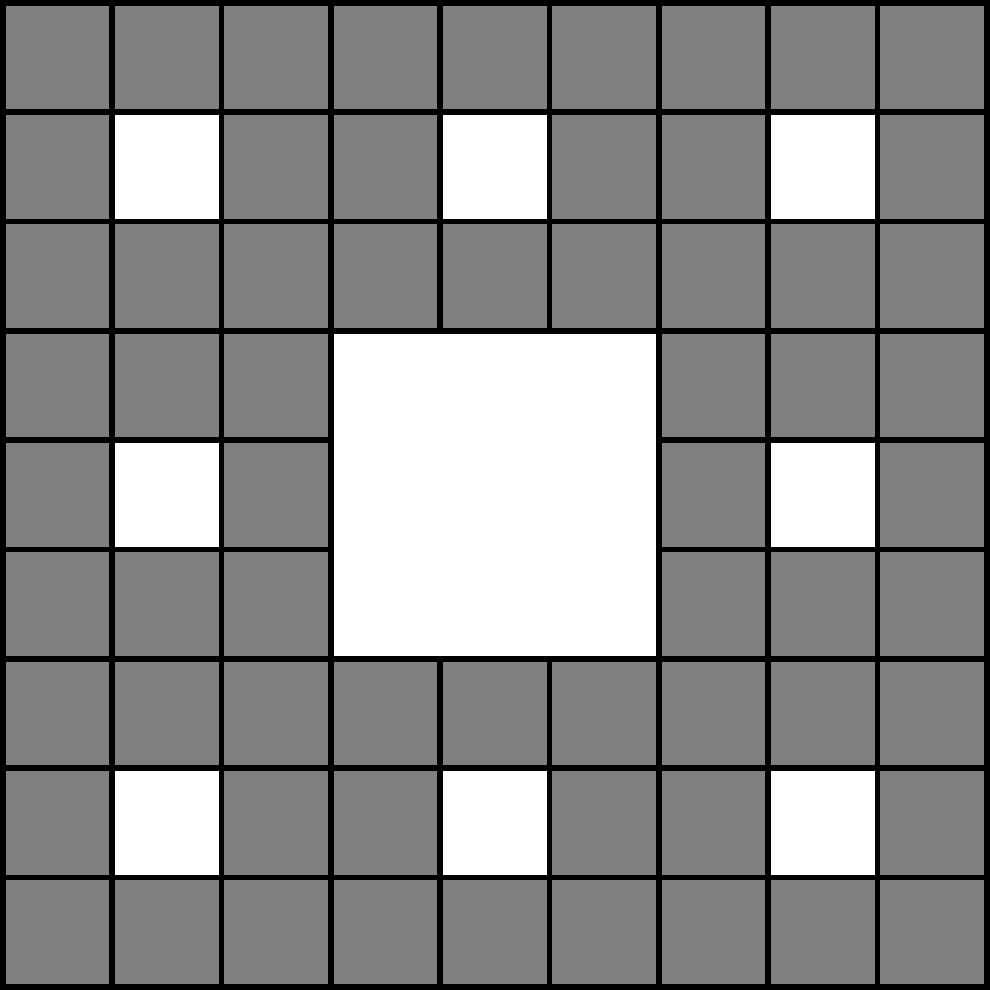}
        \caption{Stage $2$}
        \label{fig:carpet2}
    \end{subfigure}
      \begin{subfigure}[b]{0.27\textwidth}
		\centering        
        \includegraphics[width=1.7in]{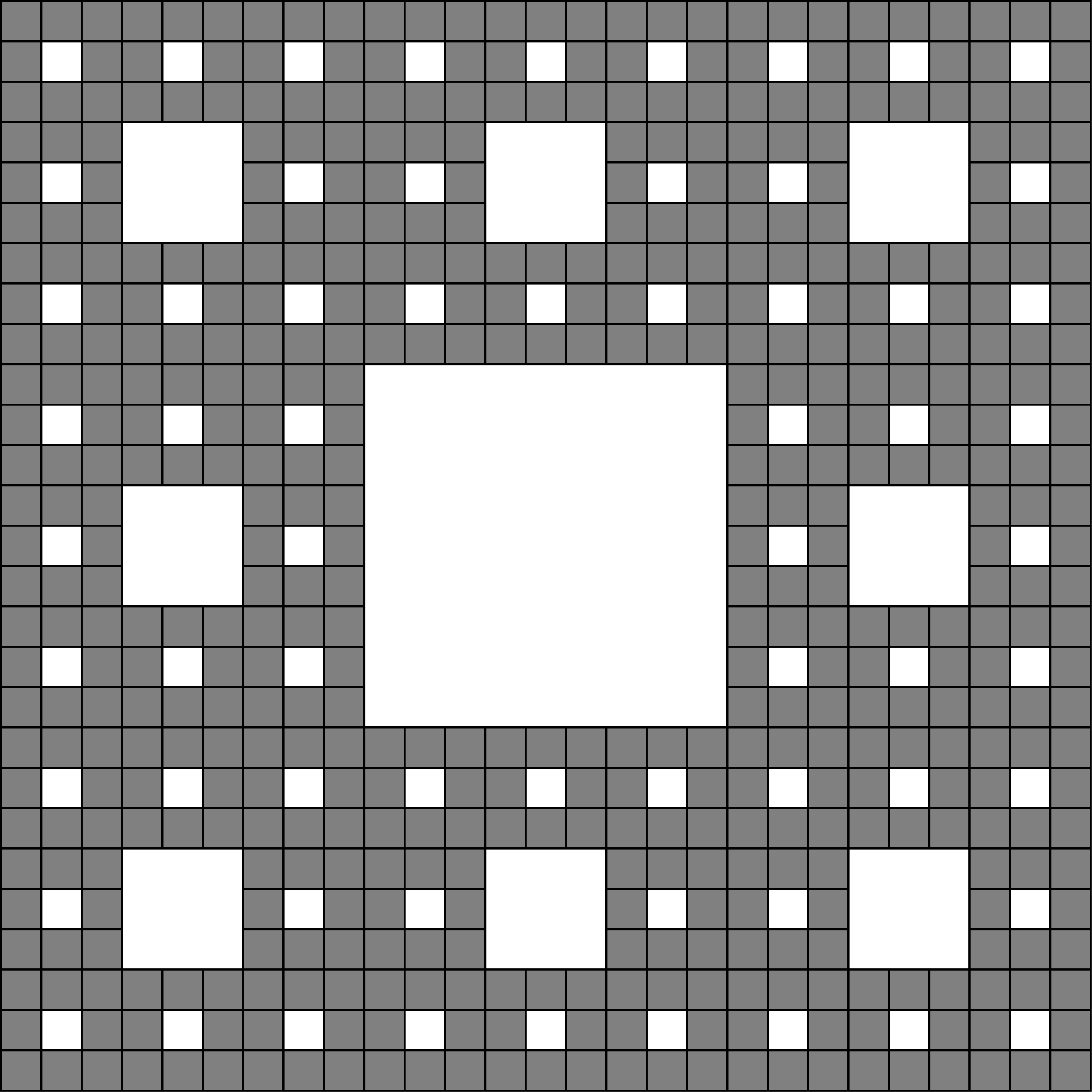}
        \caption{Stage $3$}
        \label{fig:carpet3}
    \end{subfigure}
    
    \caption{Three stages of the Sierpi\'{n}ski carpet}\label{fig:carpet-stages}
\end{figure}

\begin{figure*}[!htp]
    \centering
        \includegraphics[width=6in]{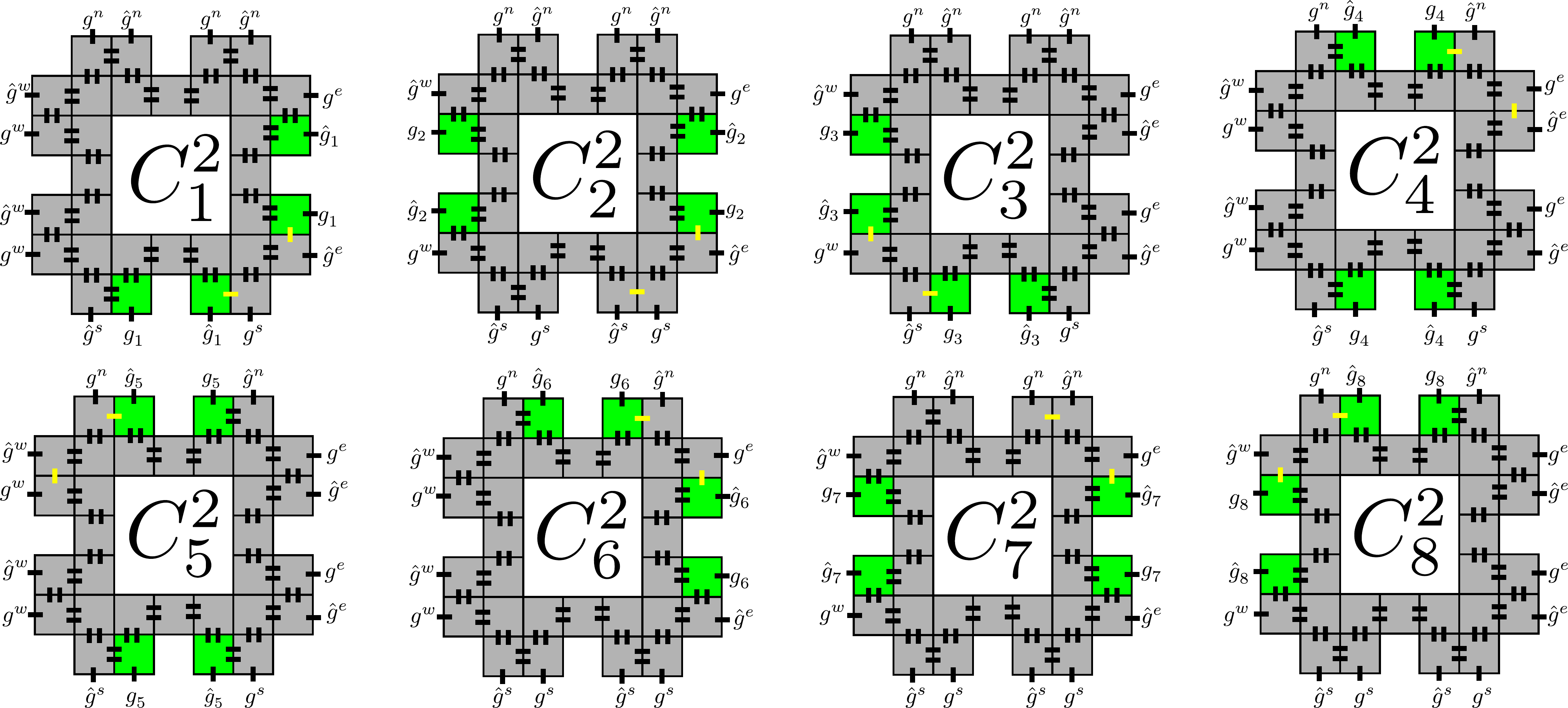}
    \caption{The tiles that self-assemble a stage $2$ supertile \fracasm{C}{2}{i}. The unlabelled strength $1$ and $2$ black and yellow glues shown on edges of two adjacent tiles in each of the $8$ supertiles are defined to have matching type. Moreover, these glues do not match any other glues of other tile types in $T$. In other words, tiles types have been ``hard-coded'' so that each of these $8$ supertiles self-assembles.} \label{fig:carpet2-glues}
\end{figure*}

We now describe the tile set, $T$, that is used to finitely self-assemble $\carpet$ in the 2HAM at temperature $\tau = 2$ at scale factor $1$. 

\subsubsection{The Sierpi\'{n}ski carpet tile set}\label{sec:carpet-details}

To define the tile set $T$, we begin by distinguishing between two classes of tile types called \grout\ tile types and \initializer\ tile types. We say \grout\ 
(respectively \initializer ) tiles or supertiles when referring to tiles or supertiles consisting only of tiles with \grout\ (respectively \initializer ) tile types.  At a high-level, \initializer\ tiles self-assemble into supertiles corresponding to $\carpet$ at stage $2$, and \grout\ tile types self-assemble into supertiles which facilitate the self-assembly of each consecutive stage of $\carpet$ starting from the stage $2$ supertiles self-assembled by \initializer\ tile types. We first describe \initializer\ tile types.

\subsubsubsection{Self-assembly of stage $2$ by \initializer\ tile types}\label{sec:carpet2}

The \initializer\ tiles self-assemble to form $8$ different supertiles, the domains of which are contained in a portion of \fractal{S}{2}. See Figure~\ref{fig:carpet2-glues} for a depiction of these $8$ supertiles. We denote these $8$ supertiles by \fracasm{C}{2}{i} for $1\leq i \leq 8$.  For each $i$, we define $32$ unique tile types of $T$ that self-assemble the supertile \fracasm{C}{2}{i} corresponding to a portion of \fractal{S}{2} that will be in the $i^{th}$ position of a supertile corresponding to a portion of \fractal{S}{3} (this portion is depicted in Figure~\ref{fig:Cj3}).  The main idea is that tiles that self-assemble \fracasm{C}{2}{i} have been ``hard-coded'' (i.e. for any glue on the edge of some tile, there exists a single matching glue on another tile) to ensure that for each $i$, \fracasm{C}{2}{i} self-assembles. Moreover, tile types are defined so that all tiles of \fracasm{C}{2}{i} self-assemble before \fracasm{C}{2}{i} can be contained in a strictly larger supertile. 
In other words, referring to Figure~\ref{fig:carpet2-glues-with-grout}, the gray and green tiles self-assemble supertiles consisting \fracasm{C}{2}{i} before any of the the aqua tiles can attach. To see this, note the presence of the yellow glues in the supertiles shown in Figure~\ref{fig:carpet2-glues}. These yellow glues restrict the assembly sequences for each supertile at temperature $2$. In particular, the final step in the assembly sequence of $C_1^2$ is the binding event between a supertile of size $3$ and a supertile of size $29$ via two yellow glues. Therefore, $C_1^2$ is completely self-assembled exactly when glues $g_1$ and $g^s$ are exposed by edges of tiles of $C_1^2$, and only after these glues are present can a supertile (called a \starter\ and described in more detail in Section~\ref{sec:carpet-grout}) shown in Figure~\ref{fig:carpet-grout1-init} bind, leading to a supertile strictly containing \fracasm{C}{2}{1} as a subassembly.

Referring to Figure~\ref{fig:carpet2-glues}, note that for each $i$, \fracasm{C}{2}{i} supertiles may expose glues of type $g^d$ or $\hat{g}^d$ for $d$ either $n$, $s$, $e$, or $w$, as well as possibly $g_k$ or $\hat{g}_k$ for $1\leq k \leq 8$. These glues allow \grout\ supertiles to cooperatively bind and the glues labeled $g_k$ and $\hat{g}_k$ indicate where special \grout\ supertiles should bind, hence they are called \indicating\ glues. Tiles containing an edge with an \indicating\ glue are depicted in green in Figure~\ref{fig:carpet2-glues}. 

The self-assembly of supertiles corresponding to stage $3$ of the Sierpi\'{n}ski carpet will require \grout\ tile types. These tile types are described in the next section. We first describe how \grout\ tile types facilitate the self-assembly of supertiles corresponding to stage $3$ of the carpet and then describe how these same \grout\ tile types facilitate the self-assembly of supertiles corresponding to any stage, $s$ say, by binding to supertiles corresponding to stage $s-1$.

\subsubsubsection{\grout\ tile types and stage $3$ carpet assembly}\label{sec:carpet-grout}

Figures~\ref{fig:inits126}-\ref{fig:carpet-grout2-west} describe \grout\ supertiles that bind to \fracasm{C}{2}{1} or \fracasm{C}{2}{2}. For a depiction of the \grout\ supertiles that bind to \fracasm{C}{2}{i} for $3\leq i \leq 8$, 
\ifx\arxiv\undefined
see the supplementary material for this paper.
\else
see Section~\ref{sec:appendix-carpet-tiles}.
\fi
We describe the \grout\ supertiles that attach to \fracasm{C}{2}{1} and \fracasm{C}{2}{2}, and note that the \grout\ supertiles that attach to \fracasm{C}{2}{i} for $3\leq i \leq 8$ are similar. 

\begin{figure*}[htp]
    \centering
\begin{subfigure}{0.25\textwidth}
            \centering
            \includegraphics[width=\textwidth]{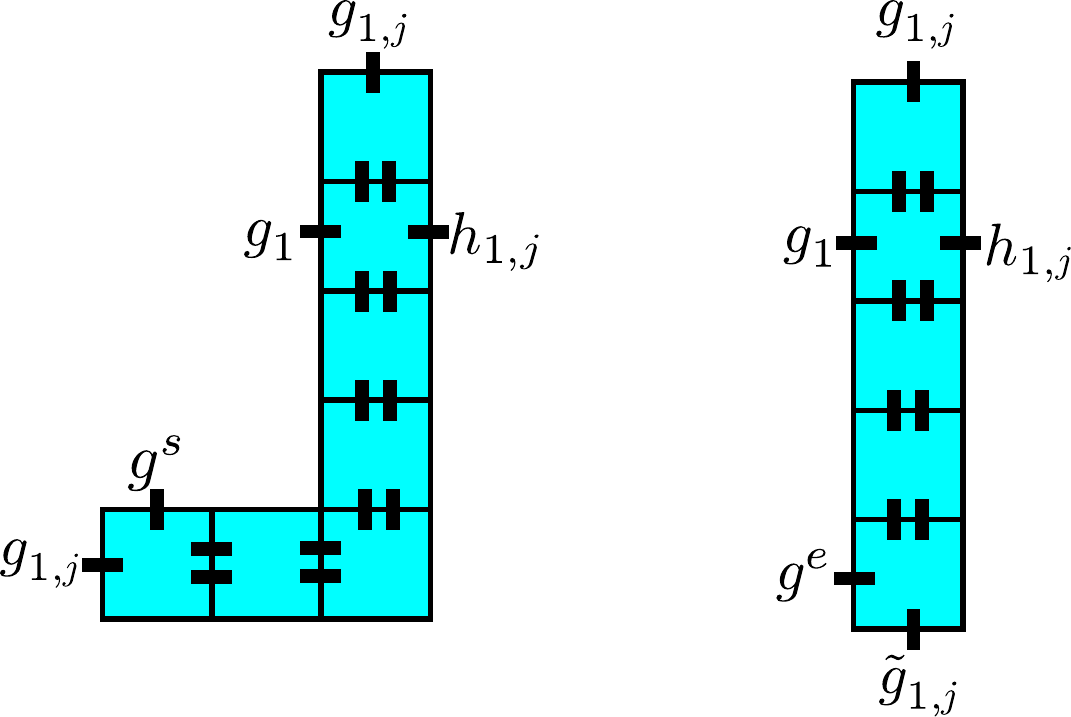}
            \caption{}\label{fig:carpet-grout1-init}
          \end{subfigure}
                \qquad\qquad
\begin{subfigure}{0.25\textwidth}
            \centering
            \includegraphics[width=\textwidth]{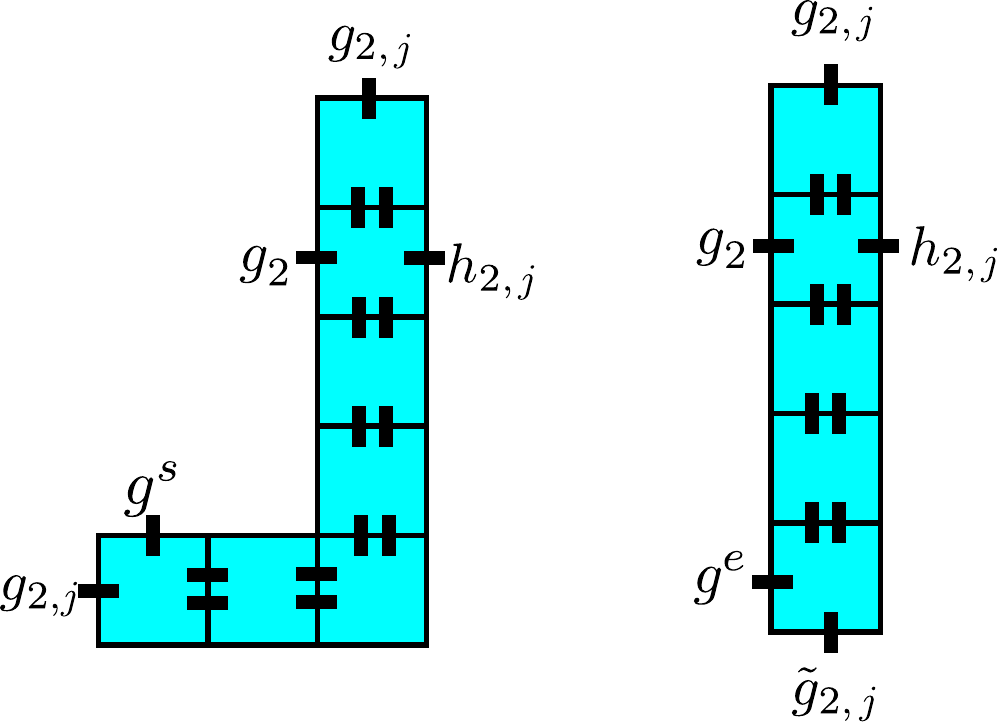}
    \caption{} \label{fig:carpet-grout2-init}
          \end{subfigure}    
          
    \caption{\starter\ supertiles. Tiles depicted in this figure have \grout\ class $j$ for some $j$ between $1$ and $8$ (inclusive). (a) Left: The supertile that starts the growth of \grout\ for \fracasm{C}{2}{1}. Right: The supertile that starts the growth of \grout\ for \fracasm{C}{s}{1}\ for $s>2$. Note that for each $s\geq 2$, only one of these supertiles can bind to tiles of \fracasm{C}{s}{1}.  Moreover, the supertile depicted on the left can bind to some \fracasm{C}{s}{1}\ iff $s=2$, and the supertile depicted on the right can bind to some \fracasm{C}{s}{1}\ iff $s>2$. (b) The supertiles that start the growth of \grout\ for $C^s_2$ for $s\geq 2$.}\label{fig:inits126}
\end{figure*}

\begin{figure}[htp]
    \centering
        \includegraphics[width=3in]{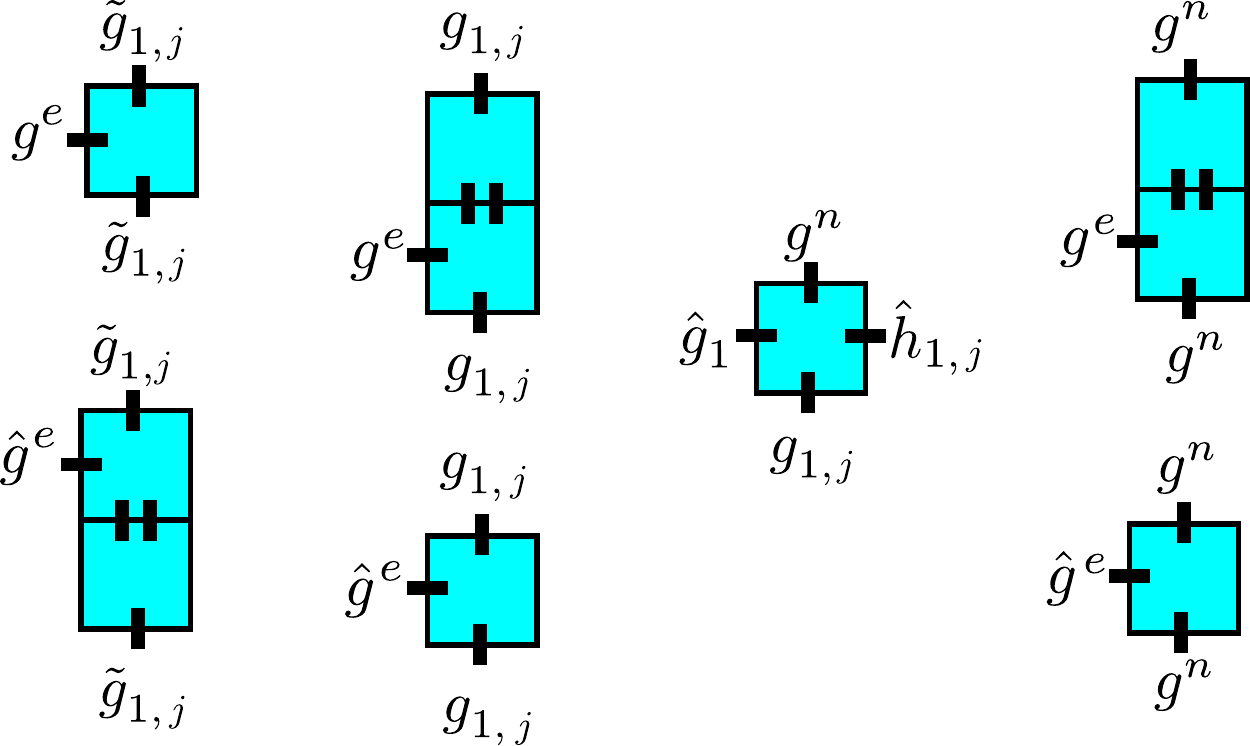}
    \caption{A depiction of \grout\ tiles that bind to the easternmost tiles of a \fracasm{C}{s}{1}\ supertile. Labels for unlabelled glues are ``hard-coded'' to enforce the self-assembly of each supertile shown here.} \label{fig:carpet-grout1-west}
\end{figure}

\begin{figure}[htp]
    \centering
        \includegraphics[width=3.2in]{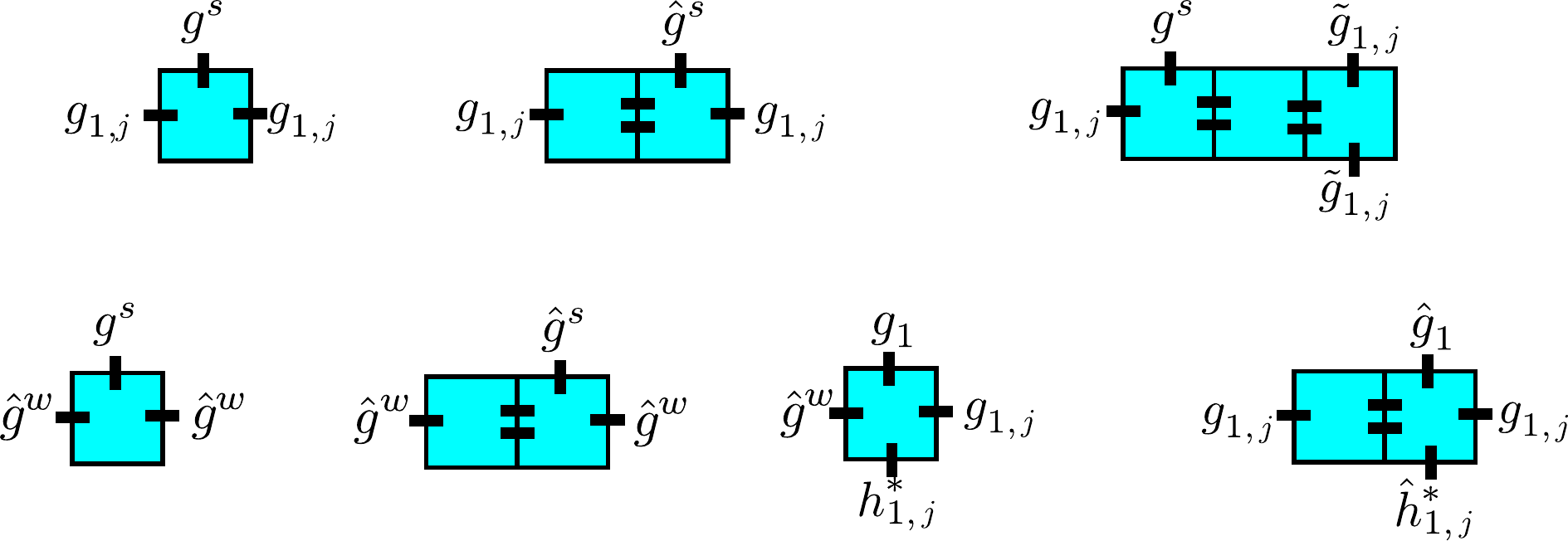}
    \caption{These \grout\ supertiles will self-assemble a row of tiles that bind to the southernmost tiles of \fracasm{C}{s}{1}\ for some stage $s\geq 1$. Labels for unlabelled glues are ``hard-coded'' to enforce the self-assembly of each supertile shown here.} \label{fig:carpet-grout1-south}
\end{figure}

\begin{figure}[htp]
    \centering
        \includegraphics[width=3.2in]{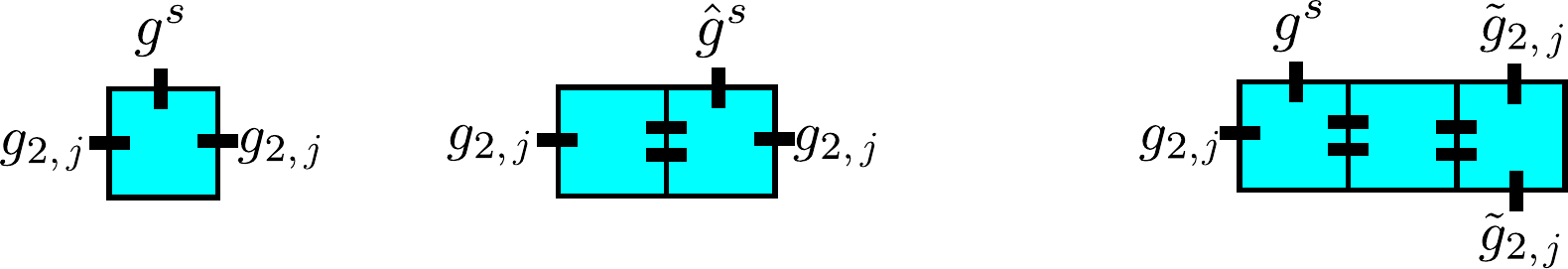}
    \caption{These \grout\ supertiles will self-assemble a row of tiles that bind to the southernmost tiles of \fracasm{C}{s}{2}\ for some stage $s\geq 1$.} \label{fig:carpet-grout2-south}
\end{figure}

\begin{figure}[htp]
    \centering
        \includegraphics[width=3in]{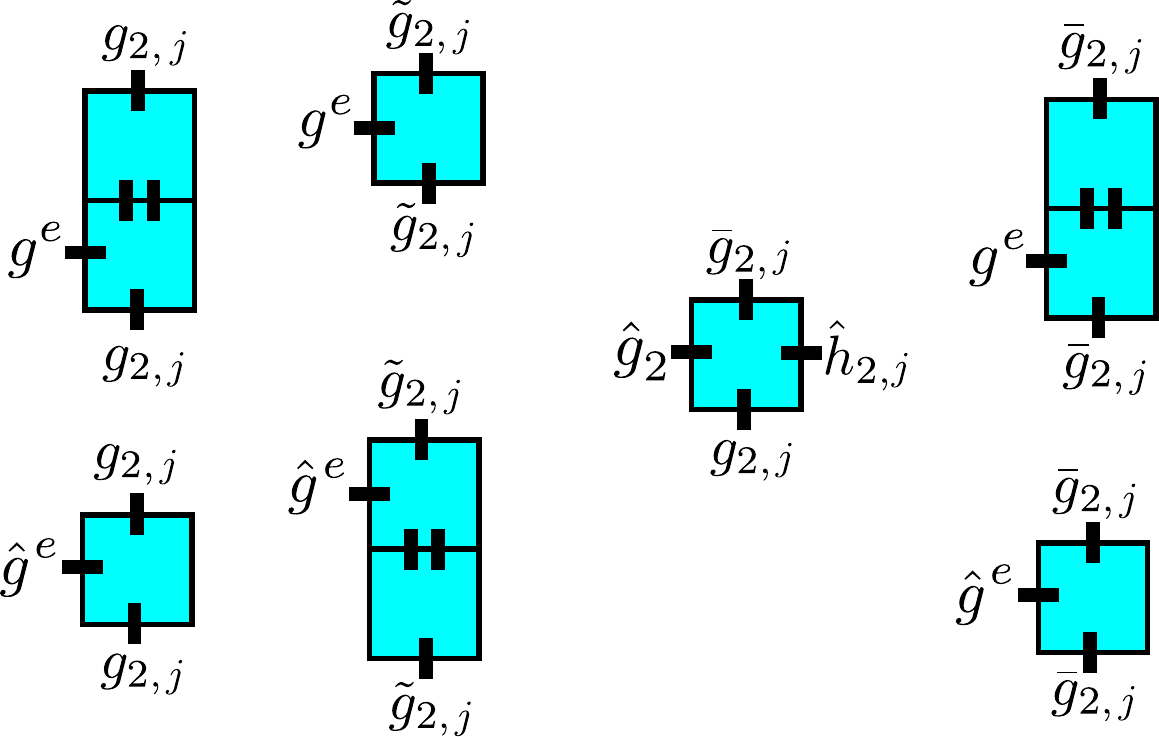}
    \caption{These \grout\ supertiles bind to the easternmost tiles of \fracasm{C}{s}{2}. Note the presence of the glue $\bar{g}_{2,i}$. This glue will either be $g^n$ or $g_i$ depending on $i$.} \label{fig:carpet-grout2-east}
\end{figure}

\begin{figure}[htp]
    \centering
        \includegraphics[width=3.2in]{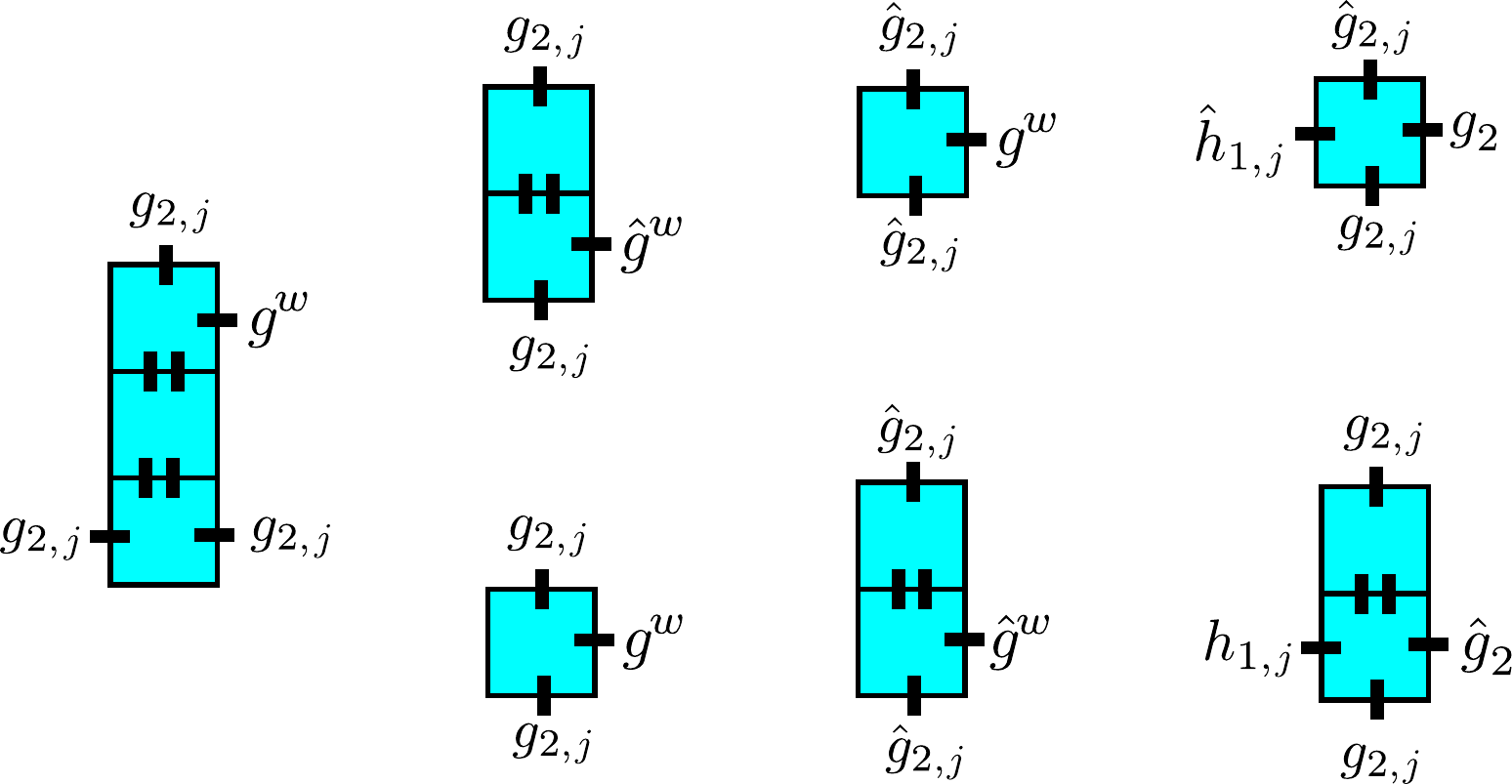}
    \caption{These tiles and supertiles are analogous to those in Figure~\ref{fig:carpet-grout2-east} only they bind to the westernmost tiles of some \fracasm{C}{s}{2}\ for $s\geq 2$. Note the presence of the glue $\hat{g}_{2,i}$. This glue will either be $\hat{g}^n$ or $\hat{g}_i$ depending on $i$.} \label{fig:carpet-grout2-west}
\end{figure}

For each $i$, there are $8$ different classes of \grout\ tile types which we enumerate with $1$ through $8$ that can bind to supertile $C^2_i$. In other words, for each supertile Figures~\ref{fig:inits126}-\ref{fig:carpet-grout2-west}, tile types for \grout\ tiles are defined so that eight different versions of each of \grout\ supertiles, corresponding to eight \grout\ \emph{classes}, self-assemble. In each figure, $j\in \N$ is such that $1\leq j \leq 8$, and tiles of supertiles belong to \grout\ class $j$. Depending on the value of $j$, for $k \in \N$ such that $1\leq k \leq 8$, the glues $h_{k,j}$, $\hat{h}_{k,j}$, $h_{1,j}^*$, and $\hat{h}_{1,j}^*$ are defined to either have strength $1$ or $0$. Table~\ref{tbl:binding-glues} describes glue strengths for these glues for each $j$. In addition, for $p \in \{2, 4, 5, 7\}$, glues with labels $\hat{g}_{p,j}$ and $\bar{g}_{p,j}$ are defined in Table~\ref{tbl:indicator-glues}. 

The \grout\ tiles are hard-coded to self-assemble \grout\ supertiles such that only \grout\ tiles belonging to the same class can bind. Moreover, two distinct \grout\ supertiles have matching glues iff the tiles of these supertiles have types belonging to the same \grout\ class. That is, for each $i$, \grout\ supertiles with tiles of any one, and only one, of the $8$ \grout\ classes can bind to some \fracasm{C}{2}{i}. For example, the \grout\ supertiles that bind to some \fracasm{C}{2}{i} before any other \grout\ supertiles are called \starter\ supertiles. See Figure~\ref{fig:inits126} for examples of \starter\ supertiles.  

\begin{table}[!htp]
\centering
  \begin{tabular}{ | g | c | }
    \hline\rowcolor{lightgray}
    $j$ & glues with strength $0$ \\ \hline               
    $1$ & $h_{5,j}$, $\hat{h}_{7,j}$  \\ \hline 
    $2$ & $h_{5,j}$, $\hat{h}_{7,j}$ \\ \hline
    $3$ & $\hat{h}_{4,j}$, $\hat{h}_{6,j}$ \\ \hline
    $4$ & $h_{2,j}$, $\hat{h}_{3,j}$ \\ \hline
    $5$ & $h_{1,j}$, $\hat{h}_{1,j}^*$ \\ \hline
    $6$ & $h_{2,j}$, $\hat{h}_{3,j}$ \\ \hline
    $7$ & $h_{2,j}$, $\hat{h}_{3,j}$ \\ \hline
    $8$ & $h_{1,j}$, $\hat{h}_{1,j}^*$ \\ \hline 
  \end{tabular}
\caption{For $j\in \N$ such that $1\leq j \leq 8$, this table lists those glues defined to have strength $0$. For all $k \in \N$ such that $1\leq k \leq 8$, $h_{k,j}$, $\hat{h}_{k,j}$, $h_{1,j}^*$, and $\hat{h}_{1,j}^*$ not listed in a row for a fixed value $j$ are defined to have strength $1$.}\label{tbl:binding-glues}
\end{table}

\begin{table}[!htp]
\centering
  \begin{tabular}{ | g | c | c | c | c | c | c | c | c | }
    \hline\rowcolor{lightgray}
    $j$ & $\hat{g}_{2,j}$ & $\bar{g}_{2,j}$ & $\hat{g}_{4,j}$ & $\bar{g}_{4,j}$ & $\hat{g}_{5,j}$ & $\bar{g}_{5,j}$ & $\hat{g}_{7,j}$ & $\bar{g}_{7,j}$ \\ \hline
    $1$ & $\hat{g}^n$ & $g^n$ & $\hat{g}^w$ & $g^w$ & $\hat{g}_1$ & $g_1$ & $g_1$ & $\hat{g}_1$ \\\hline
    $2$ & $\hat{g}^n$ & $g^n$ & $\hat{g}_2$ & $g_2$ & $g_2$ & $\hat{g}_2$ & $g^s$ & $\hat{g}^s$ \\\hline 
    $3$ & $\hat{g}^n$ & $g^n$ & $\hat{g}_3$ & $g_3$ & $g^e$ & $\hat{g}^e$ & $g_3$ & $\hat{g}_3$ \\\hline 
    $4$ & $\hat{g}_4$ & $g_4$ & $\hat{g}^w$ & $g^w$ & $g^e$ & $\hat{g}^e$ & $g_4$ & $\hat{g}_4$ \\\hline 
    $5$ & $\hat{g}_5$ & $g_5$ & $\hat{g}^w$ & $g^w$ & $g^e$ & $\hat{g}^e$ & $g_5$ & $\hat{g}_5$ \\\hline 
    $6$ & $\hat{g}_6$ & $g_6$ & $\hat{g}^w$ & $g^w$ & $g_6$ & $\hat{g}_6$ & $g^s$ & $\hat{g}^s$ \\\hline 
    $7$ & $\hat{g}^n$ & $g^n$ & $\hat{g}_7$ & $g_7$ & $g_7$ & $\hat{g}_7$ & $g^s$ & $\hat{g}^s$ \\\hline 
    $8$ & $\hat{g}_8$ & $g_8$ & $\hat{g}_8$ & $g_8$ & $g^e$ & $\hat{g}^e$ & $g^s$ & $\hat{g}^s$ \\\hline               
  \end{tabular}
\caption{For $j\in \N$ such that $1\leq j \leq 8$, this table gives glue definitions. For example, when $j=1$, $\hat{g}_{2,j} = \hat{g}^n$. All glues in this table are also defined to have strength $1$.}\label{tbl:indicator-glues}
\end{table}

\begin{figure*}[!htp]
    \centering
        \includegraphics[width=6.5in]{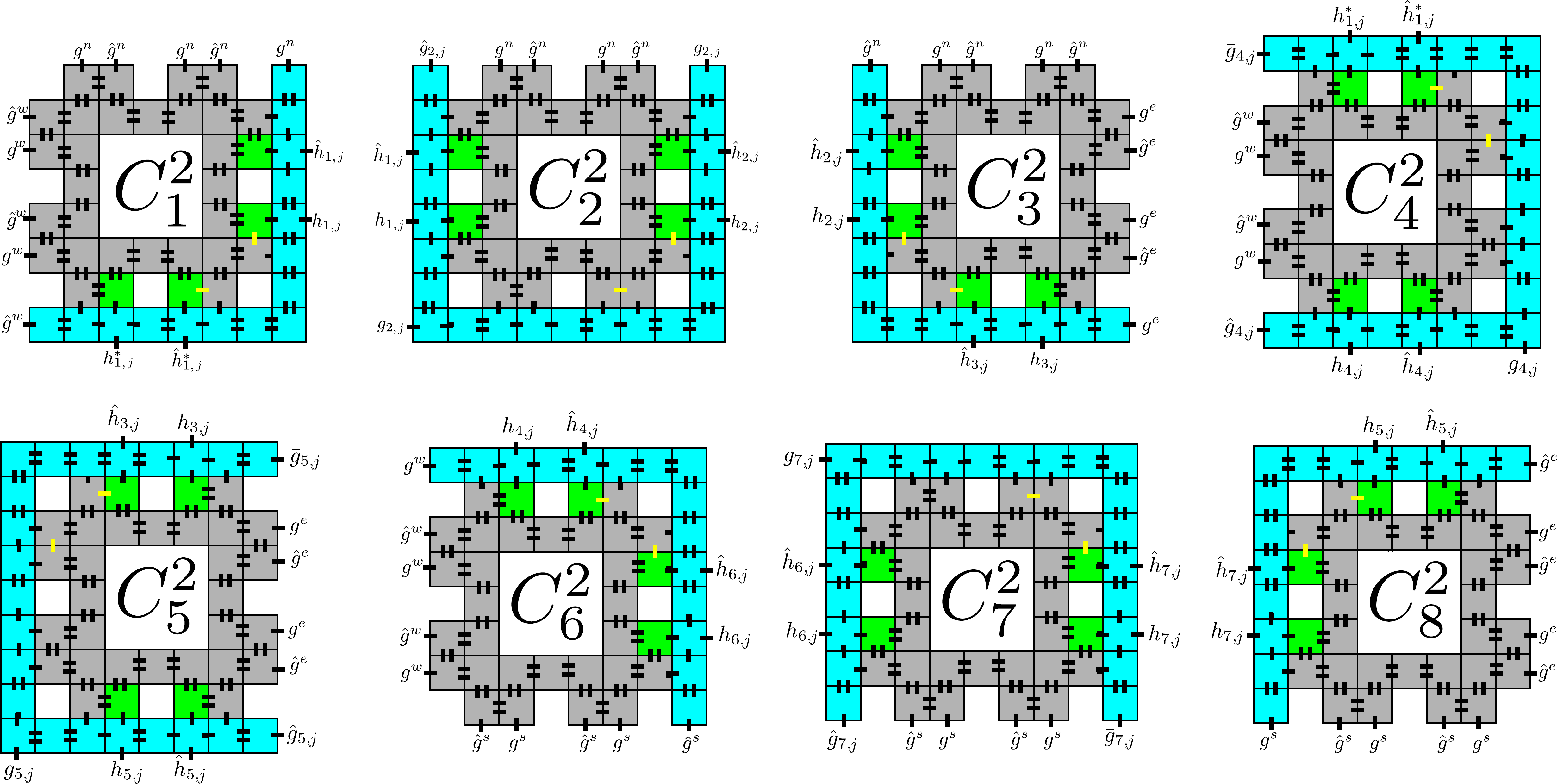}
    \caption{Supertiles \fracasm{C}{2}{(i, j)} for $1\leq i \leq 8$ and some $j$ such that $1\leq j \leq 8$. Depending on $j$, certain glues will have strength of $0$ as described in Table~\ref{tbl:binding-glues} though they are shown here as strength-$1$ glues.} \label{fig:carpet2-glues-with-grout}
\end{figure*}

\begin{figure}[htp]
    \centering
        \includegraphics[width=3in]{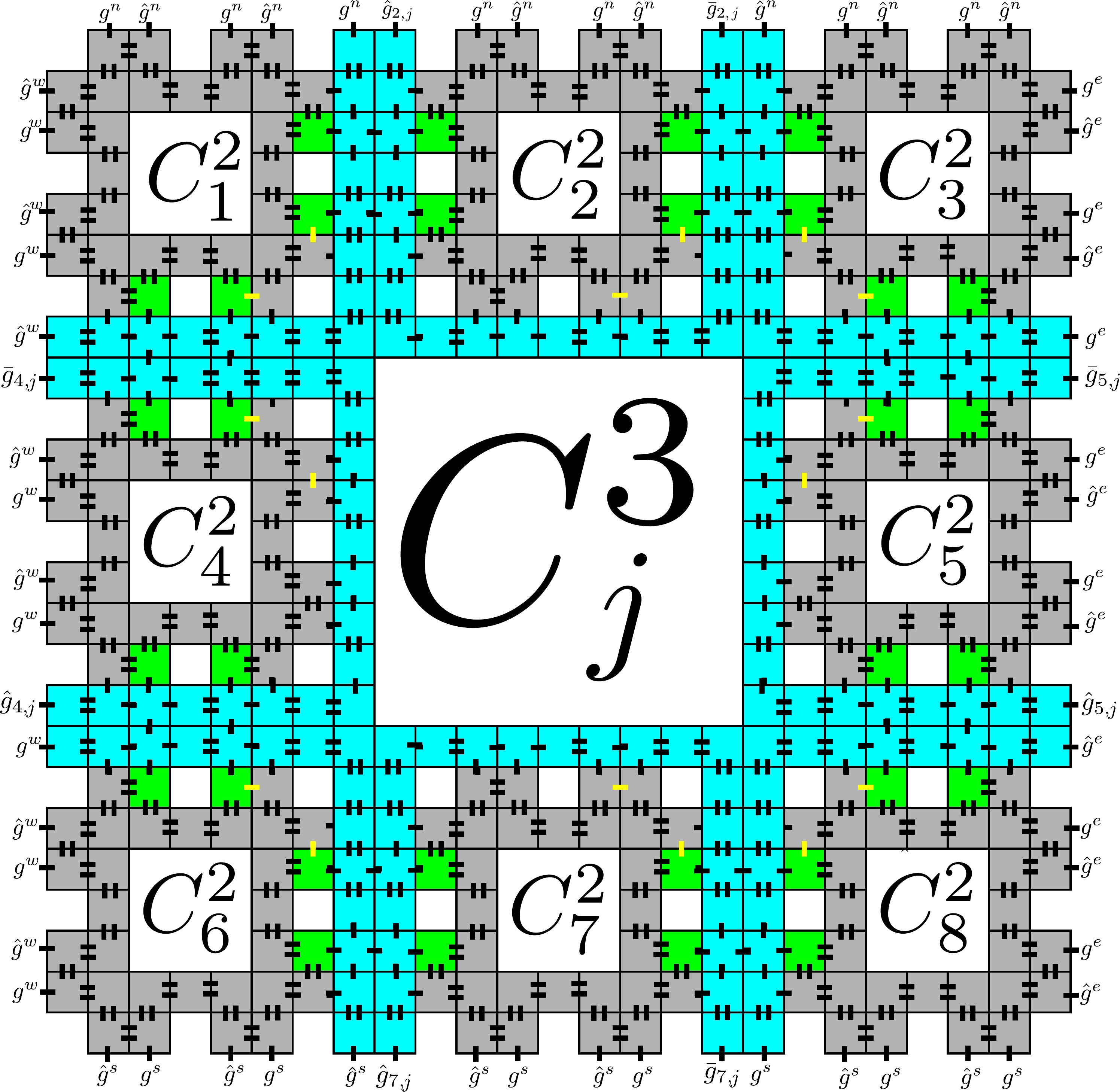}
    \caption{A depiction of \fracasm{C}{3}{j}. Note that for $p\in \{2,4,5,7\}$ the glues $\hat{g}_{p,j}$ and $\bar{g}_{p,j}$ shown here are defined in Table~\ref{tbl:indicator-glues}.} \label{fig:Cj3}
\end{figure}

For $i$ between $1$ and $8$ (inclusive), after supertiles \fracasm{C}{2}{i} self-assemble, \grout\ tiles attach to these supertiles to form supertiles which expose glues that allow them to bind to each other to self-assemble a supertile corresponding to stage $3$ of the Sierpi\'{n}ski carpet. Figure~\ref{fig:carpet2-glues-with-grout} shows each supertile \fracasm{C}{2}{i} for $1\leq i \leq 8$ along with \grout\ supertiles with \grout\ class $j$ attached. Figure~\ref{fig:Cj3} gives a depiction of the portion of $S^3$ that self-assembles; \grout\ supertiles in this figure are depicted in aqua.

Starting from some supertile \fracasm{C}{2}{i}, initial growth of \grout\ tiles begins when a \starter\ cooperatively binds to some \fracasm{C}{2}{i}\ via pairs of glues exposed by each supertile \fracasm{C}{2}{i}. Figure~\ref{fig:carpet-grout1-init} depicts such a supertile that binds to a \fracasm{C}{2}{1}\ supertile when the glues $g_1$ and $g^s$ cooperatively bind to the matching glues of \fracasm{C}{2}{1}. One can observe that the glues of \grout\ supertiles have been defined so that binding of \grout\ supertiles to \fracasm{C}{2}{i}\ for $1\leq i \leq 8$ always begins with the attachment of a \starter\ supertile.

Glues of \grout\ tiles have also been defined so that after a \starter\ binds to \fracasm{C}{2}{i} for some $i$,  \grout\ supertiles cooperatively bind one at a time and partially surround the supertile \fracasm{C}{2}{i} as in Figure~\ref{fig:carpet2-glues-with-grout}. We refer to the \grout\ supertiles other than \starter\ supertiles that cooperatively bind to \fracasm{C}{2}{i} as \crawler\ supertiles. Figures~\ref{fig:carpet-grout1-west} and~\ref{fig:carpet-grout1-south}. depict \crawler\ supertiles that bind to \fracasm{C}{2}{1}, and Figures~\ref{fig:carpet-grout2-south},~\ref{fig:carpet-grout2-east}, and~\ref{fig:carpet-grout2-west} depict \crawler\ supertiles that bind to \fracasm{C}{2}{2}. 

A \grout\ tile that binds to an \indicating\ glue (for $1\le k \le 8$, glues with label $g_k$ or $\hat{g}_k$ in Figure~\ref{fig:carpet2-glues}) of a south edge of a tile belonging to \fracasm{C}{2}{i} (respectively north, east, or west) will have a glue on its south (respectively north, east, or west) edge. The strength of such a glue is either $0$ or $1$ as given in Table~\ref{tbl:binding-glues}. The type of glue and whether or not a \grout\ tile exposes such a glue depends on the class of the \grout\ supertiles that attach to some \fracasm{C}{2}{i}. We call these glues exposed on an edge of a \grout\ tile \stagebinding\ glues. In Figures~\ref{fig:inits126} through~\ref{fig:carpet-grout2-west} and~\ref{fig:carpet2-glues-with-grout}, \stagebinding\ glues are $h_{1,j}^*$, $\hat{h}_{1,j}^*$, or $h_{k,j}$, $\hat{h}_{k,j}$ for $1\leq k \leq 8$. Strength-$1$ \stagebinding\ glues exposed by \grout\ supertiles bound to \fracasm{C}{2}{i}\ supertiles bind to allow for the self-assembly of a supertile that corresponds to the third stage of the carpet. 

Now let \fracasm{C}{2}{(i,j)}\ denote any supertile consisting only of tiles of \fracasm{C}{2}{i} and \grout\ tiles of class $j$. Figure~\ref{fig:carpet2-glues-with-grout} depicts such supertiles. The supertiles depicted in Figure~\ref{fig:carpet2-glues-with-grout} are such that no other \grout\ supertiles can bind to a given \fracasm{C}{2}{i} and have been depicted this way to show all of the glues exposed after \grout\ supertiles bind to each \fracasm{C}{2}{i}. We note that \grout\ tile types have been defined such that for $i,j,i'$ and $j'$ between $1$ and $8$ (inclusive), supertiles \fracasm{C}{2}{(i,j)}\ and \fracasm{C}{2}{(i',j')}\ can bind only after exposing sufficient \stagebinding\ glues. Moreover, such supertiles can bind iff $j=j'$. That is the \grout\ tiles of \fracasm{C}{2}{(i,j)}\ and \fracasm{C}{2}{(i',j')}\ belong to the same class.

For a fixed \grout\ class $j$ between $1$ and $8$, the $8$ supertiles \fracasm{C}{2}{(i,j)} (where $i$ ranges from $1$ to $8$) with sufficient \grout\ supertiles attached bind to self-assemble a supertile, which we denote by \fracasm{C}{3}{j}, corresponding to stage $3$ of the carpet. Figure~\ref{fig:Cj3} depicts such a supertile \fracasm{C}{3}{j}. Just as $i$ corresponds to the position that \fracasm{C}{2}{i}\ is located in \fracasm{C}{3}{j}, the \grout\ class $j$ determines the position that \fracasm{C}{3}{j}\ will be located as a substage of a supertile corresponding to stage $4$ of the carpet. Moreover, with glues strengths given Table~\ref{tbl:binding-glues}, we note that \grout\ tiles have been defined so that such $C^2_{(i,j)}$ supertiles bind before the ``next iteration'' of \grout\ tiles can attach. In other words, $C^2_{(i,j)}$ supertiles bind for all $i$ between $1$ and $8$ before a \starter\ can bind to the resulting supertile $C^3_{j}$. For example, when $j=1$, \stagebinding\ glues are defined such that $h_{5,j}$ and $\hat{h}_{7,j}$ have strength $0$. Therefore, any assembly sequence of \fracasm{C}{3}{1} ends with \fracasm{C}{2}{(8,1)} binding to a supertile consisting of \fracasm{C}{2}{(k,1)} for $1\leq k \leq 7$. Hence, only after \fracasm{C}{2}{(8,1)} binds can a \starter\ bind to the resulting supertile. The cases where $j$ is such that $2\leq j \leq 8$ are similar.

Then, for $i'$ such that $1\leq i' \leq 8$, the glues that might allow (depending on $i$ and $i'$) some supertile \fracasm{C}{2}{(i,j)}\ to bind to another supertile \fracasm{C}{2}{(i',j)}\ are \stagebinding\ glues separated by a distance of $2 = 3^{2-1} - 1$.\footnote{We are including glues with strength $0$ here.} This distance is ensured by the locations of the \indicating\ glues. As we will see, \stagebinding\ glues will be reused as each consecutive stage of the carpet self-assembles. The distance between \stagebinding\ glues will prevent supertiles corresponding to different fractal stages from binding. 

Finally, the class of \grout\ tiles that bind to some \fracasm{C}{2}{i} determines the presence and locations of \indicating\ glues exposed by edges belonging to tiles of some \fracasm{C}{3}{j}. These \indicating\ glues belonging to \grout\ tiles are defined according to Table~\ref{tbl:indicator-glues}. The locations of \indicating\ glues exposed by \fracasm{C}{3}{j} are analogous to the locations of these glues exposed by \fracasm{C}{2}{j} as shown in Figure~\ref{fig:carpet2-glues}, only the \indicating\ glues of \fracasm{C}{3}{j} are at distance $8 = 3^{3-1} - 1$ apart. For example, referring to Figure~\ref{fig:Cj3}, when $j=1$, we note the presence of four \indicating\ glues (two belonging to easternmost tiles and two belonging to southernmost tiles according to Table~\ref{tbl:indicator-glues}) exposed by \fracasm{C}{3}{1} that are distance $8$ apart. Note the similarity between the locations of \indicating\ glues in \fracasm{C}{3}{1} and in \fracasm{C}{2}{1}. \grout\ tile types have been defined so that the same similarity is drawn between \fracasm{C}{3}{j} and \fracasm{C}{2}{j} for $j$ between $1$ and $8$ (inclusive). 

\begin{figure}[htp]	

  \begin{center}
        \includegraphics[width=3in]{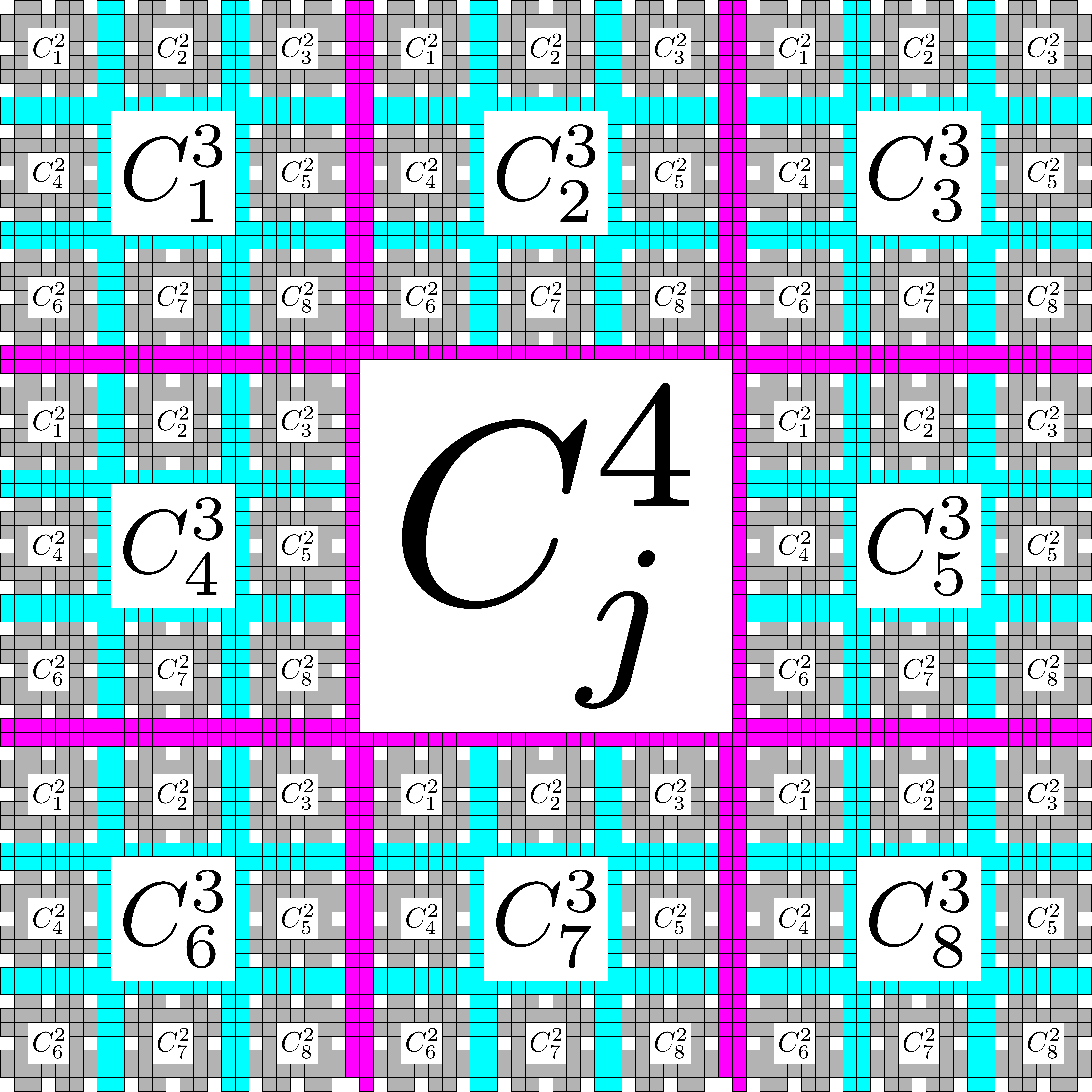}
  \end{center}

    \caption{A depiction of the portion of \fractal{S}{4}\ that is self-assembled by supertiles denoted by \fracasm{C}{3}{i}\ for $i$ between $1$ and $8$ (inclusive) and some class $j$ for $j$ between $1$ and $8$ of \grout\ tiles.}\label{fig:carpet4-overview}
 
\end{figure}

\subsubsubsection{Self-assembly of stage $s$ carpet for $s\geq 2$}

Repurposing $i$, we now let \fracasm{C}{3}{j}\ be denoted by \fracasm{C}{3}{i}.  Now, for each $i$ and $j$ with $1\leq i,j \leq 8$, the $8$ different classes of \grout\ tile types can attach to each \fracasm{C}{3}{i}\ supertile to give supertiles \fracasm{C}{3}{(i,j)}. The class \grout\ class determines where the supertiles  \fracasm{C}{3}{(i,j)}\ attach to self-assemble a supertile, \fracasm{C}{4}{j}, corresponding to a portion of \fractal{S}{4}. Such a \fracasm{C}{4}{j} is depicted in Figure~\ref{fig:carpet4-overview}. 
Moreover, the glues that allow some supertile \fracasm{C}{3}{(i,j)} to bind to another supertile \fracasm{C}{3}{(i',j)}, for some $i'$ say, are strength $1$ or $0$ glues, according to Table~\ref{tbl:binding-glues}, separated by a distance of $8$ apart. Note that the definitions of glues in Table~\ref{tbl:binding-glues} ensure that a \fracasm{C}{4}{j} supertile contains a supertile \fracasm{C}{3}{(i,j)} for each $1\leq i \leq 8$ before a \starter\ supertile can attach to such a \fracasm{C}{4}{j}. 

It is important to note that two \stagebinding\ glues may be exposed on some strict subassembly of \fracasm{C}{3}{(i,j)}, and therefore for some $i$ and $i'$, two subassemblies of \fracasm{C}{3}{(i,j)} and \fracasm{C}{3}{(i',j)} may bind to form a subassembly of \fracasm{C}{4}{j} where some \fracasm{C}{3}{(i,j)} has only partially assembled. This can lead to cases of nondeterminism like the case depicted in Figure~\ref{fig:carpet-corner-example}. We define glues belonging to \grout\ tiles so that this does not prevent tiles from binding in locations corresponding to points of stage $2$ at positions $i$ and $i'$ from completing assembly as a subassembly of \fracasm{C}{4}{j}. One such glue is shown in Figure~\ref{fig:carpet-corner-example} with label $g_{2,j}$. We also note that these glues do not permit tiles to bind in locations outside of locations in of tiles in positioned supertiles of \fracasm{C}{4}{j}. It is important to note that before such cases of nondeterminism can occur, all \stagebinding\ glues of \fracasm{C}{4}{j} must be bound. Glues such as $g_{2,j}$ also ensure correct assembly of higher stage analogs of \fracasm{C}{4}{j} where analogous nondeterminism can also occur in the self-assembly of $S^s$ for any higher stage $s>3$. 
\begin{figure*}[!h]
    \centering
        \includegraphics[width=4.5in]{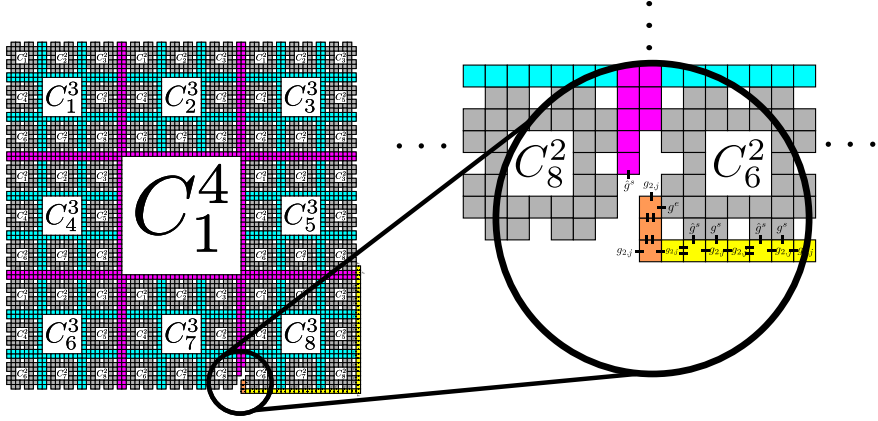}
    \caption{An example where \grout\ tiles have ``turned a corner too early''. The \grout\ tiles are shown in aqua, fuchsia, yellow, and orange. Note that \fracasm{C}{3}{8} and \fracasm{C}{3}{7} only have partial grout, though both have \grout\ supertiles with \stagebinding\ glues as is required for \fracasm{C}{4}{1} to be stable. In this case, when a \grout\ supertile shown in orange binds, a $g_{2,j}$ is exposed that will eventually allow for \grout\ tiles to continue to bind to the southernmost tiles of \fracasm{C}{4}{1}, but only after a sufficient number of \grout\ supertiles bind to \fracasm{C}{3}{7}. } \label{fig:carpet-corner-example}
\end{figure*}

Recursively repeating this process, we see that for any $i,j,s\in \N$ such that $1 \leq i,j \leq 8$ and $s > 2$, supertiles \fracasm{C}{s-1}{i}\ corresponding to a portion of \fractal{S}{s-1}\ (again, we are leaving room for \grout\ tiles) self-assemble, and supertiles \fracasm{C}{s-1}{(i,j)}\ corresponding to \fracasm{C}{s-1}{i}\ with the attachment of \grout\ tiles all belonging to the $j^{th}$ class of \grout\ tile types self-assemble. Moreover, the supertiles \fracasm{C}{s-1}{(i,j)} with sufficient \grout\ supertiles attached expose \stagebinding\ glues that are at a distance of $3^{s-2} - 1$ apart (including glues with strength $0$) that allow for the stable binding of these supertiles to form a supertile \fracasm{C}{s}{j} corresponding to \fractal{S}{s}. For $i' \in \N$ such that $1\leq i' \leq 8$, since the distance between the $2$ glues that allow for two supertiles \fracasm{C}{s-1}{(i,j)}\ and \fracasm{C}{s-1}{(i',j)}\ to bind is $3^{s-2} - 1$, one can observe that for $p,q \in \N$ such that $p,q \geq2$, \fracasm{C}{p}{(i,j)}\ can bind to some \fracasm{C}{q}{(i',j')}\ for some $i'$ and $j'$ iff $p=q$ and $j=j'$. Moreover, by  definition of the \grout\ tile types, specific edges of tiles of \fracasm{C}{s}{j}\ will expose \indicating\ glues which are analogous to the indicating glues of \fracasm{C}{s-1}{j}, only at distance $3^{s-1} - 1$ apart. 

\subsubsubsection{Correctness for the Sierpi\'{n}ski carpet construction}

To prove that the tile set, $T$, gives a 2HAM TAS $\mathcal{T} = (T,2)$ that finitely self-assembles $\carpet$, we note that by construction, for any finite producible supertile $\alpha$ of $\mathcal{T}$ and for any $s\in \N$, there exists positive integers $k$, and $j$, and an assembly sequence $\vec{\alpha} = \langle \alpha_i \rangle_{i=0}^{k}$ such that $\alpha_0 = \alpha$ and $\alpha_k$ is a $\fracasm{C}{s}{j}$ supertile. Therefore, any finite producible supertile $\alpha$ of $\mathcal{T}$ has the shape of a subset of points in $\carpet$. Moreover, for any finite producible supertile $\alpha$ of $\mathcal{T}$, there exists an assembly sequence which starts with $\alpha$ and results in a supertile that has shape of $\carpet$. Therefore, we see that $\calT$ finitely self-assembles $\carpet$.

\subsection{Self-assembly of $4$-sided fractals}\label{sec:4sided}

The construction that shows that any $4$-sided fractal finitely self-assembles in the 2HAM at scale factor $1$ (Theorem~\ref{thm:four-sided}) is a generalization of the construction given in Section~\ref{sec:carpet}. Let $G$ be the generator for a $4$-sided fractal and recall the notation of $L_G$, $R_G$, $B_G$, and $T_G$ defined in Section~\ref{sec:fractal-defs}. We will describe a tile set $T$ such that $\bX$ finitely self-assembles in the 2HAM system $\mathcal{T} = (T,2)$. As an example, consider the generator in Figure~\ref{fig:4sided1}. Stage $2$ of this fractal is depicted in Figure~\ref{fig:4sided2}. 

\begin{figure}[htp]
    \centering
    \begin{subfigure}[b]{0.4\textwidth}
    		\centering
        \includegraphics[width=1.7in]{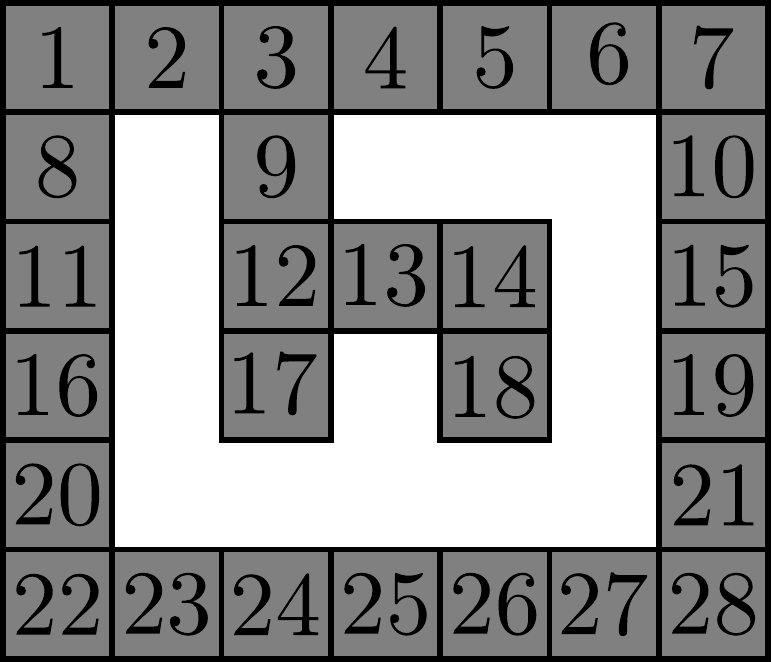}
        \caption{Stage $1$}
        \label{fig:4sided1}
    \end{subfigure}
    \begin{subfigure}[b]{0.4\textwidth}
		\centering        
        \includegraphics[width=2.2in]{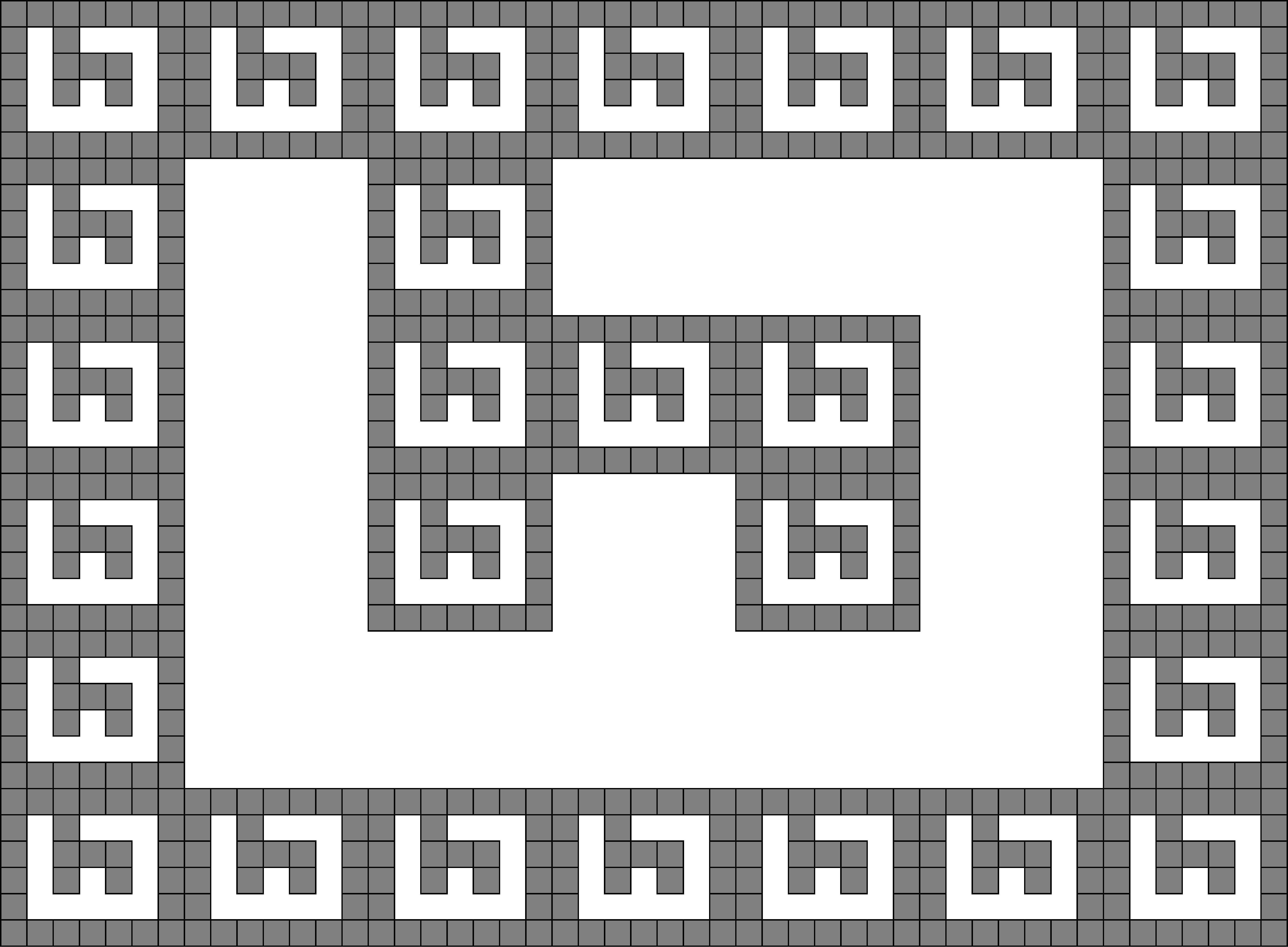}
        \caption{Stage $2$}
        \label{fig:4sided2}
    \end{subfigure}
    
    \caption{Two stages of a $4$-side fractal.}\label{fig:4sided-stages}
\end{figure}

Lemma~\ref{lem:4sided} will be helpful for showing Theorem~\ref{thm:four-sided}. This lemma states that if $\bX$ is a fractal with a generator $G$ such that $G$ only contains points along its perimeter, then $\bX$ finitely self-assembles in the 2HAM at temperature $2$.

\begin{lemma}\label{lem:4sided}
Let $\bX$ be a $4$-sided fractal with generator $G$ such that $G\setminus (L_G\cup L_G\cup T_G\cup B_G) = \emptyset$. Then, there exists a 2HAM TAS $\calT_{\bX} = (T, 2)$ that finitely self-assembles $\bX$.
\end{lemma}

\begin{proof}[Sketch]
For $s\in \N$, let \fractal{X}{s} denote the $s^{th}$ stage of $\bX$, and let $r = |G|$.
We note that the construction given in Section~\ref{sec:carpet} generalizes in a straightforward way to give a tile set $T$ satisfying Lemma~\ref{lem:4sided}. For example, given the generator in Figure~\ref{fig:4sidedLemma1}, the modifications to the construction given in Section~\ref{sec:carpet} are as follows. Once again, we consider two types of tiles in $T$ which we call \initializer\ tiles and \grout\ tiles.

\subsubsection{The \initializer\ tile types for Lemma~\ref{lem:4sided}.}
Let \fractal{X'}{2} denote the set of points in \fractal{X}{2} that are not on the perimeter of \fractal{X}{2}. Figure~\ref{fig:4sidedLemma2} depicts the points of an example \fractal{X'}{2}. 
\initializer\ tiles of $T$ now hard-code $r$ different versions of \fractal{X'}{2}. For $i$ between $1$ and $r$ (inclusive), we call these hard-coded supertiles \fracasm{\Gamma}{2}{i}. We note that as there is a Hamiltonian path in the full grid-graph of \fractal{X'}{2}, the glues of the \initializer\ tiles can be specified so that \fracasm{\Gamma}{2}{i} completely assembles prior to being a subassembly of any other producible supertile. 

In addition to hard-coding the shape of \fractal{X'}{2}, \initializer\ tiles are specified so that once \fracasm{\Gamma}{2}{i} has completely self-assembled: 
\begin{enumerate}
\item the north edges of northernmost tiles expose a $g^n$ or $\hat{g}^n$ such that the westernmost tile and every other tile from west to east exposes $g^n$ and the remaining northernmost tiles expose a $\hat{g}^n$, 
\item the east edges of easternmost tiles expose a $g^e$ or $\hat{g}^e$ such that the northernmost tile and every other tile from north to south exposes $g^e$ and the remaining easternmost tiles expose a $\hat{g}^e$, 
\item the south edges of southernmost tiles expose a $g^s$ or $\hat{g}^s$ such that the easternmost tile and every other tile from east to west exposes $g^s$ and the remaining southernmost tiles expose a $\hat{g}^s$, and finally,
\item the west edges of westernmost tiles expose a $g^w$ or $\hat{g}^w$ such that the southernmost tile and every other tile from south to north exposes $g^w$ and the remaining westernmost tiles expose a $\hat{g}^w$.  
\end{enumerate}
Edges of tiles in \fracasm{\Gamma}{2}{i} in ``key locations'' expose special glues $\hat{g}_i$ and $g_i$ which we call \indicating\ glues. At these key locations, $g_i$ is exposed instead of a $g^n$, $g^s$, $g^e$, or $g^w$ and $\hat{g}_i$ is exposed instead of a $\hat{g}^n$, $\hat{g}^s$, $\hat{g}^e$, or $\hat{g}^w$. These key locations of the tiles in \fracasm{\Gamma}{2}{i} that expose these glues are shown as red squares in Figure~\ref{fig:4sidedLemma2}. In general, these key locations will be the second to westernmost (resp. northernmost) and second to easternmost (resp. southernmost) tile locations of the northernmost (resp. easternmost) and southernmost (resp. westernmost) tile locations.  Whether or not \fracasm{\Gamma}{2}{i} exposes \indicating\ glues at these key locations depend on $i$. In particular, if the $i^{th}$ location in $G$ is adjacent to some other point that is north (resp. south, east, or west) of it, then, \fracasm{\Gamma}{2}{i} will expose \indicating\ glues on the north (resp. south, east, or west) edges of tiles in northernmost (resp. southernmost, easternmost, or westernmost) key locations.  \indicating\ glues in these key locations serve the same purpose to the \indicating\ glues described in Section~\ref{sec:carpet-details}.

\begin{figure}[htp]
    \centering
    \begin{subfigure}[b]{0.4\textwidth}
    		\centering
        \includegraphics[width=1.7in]{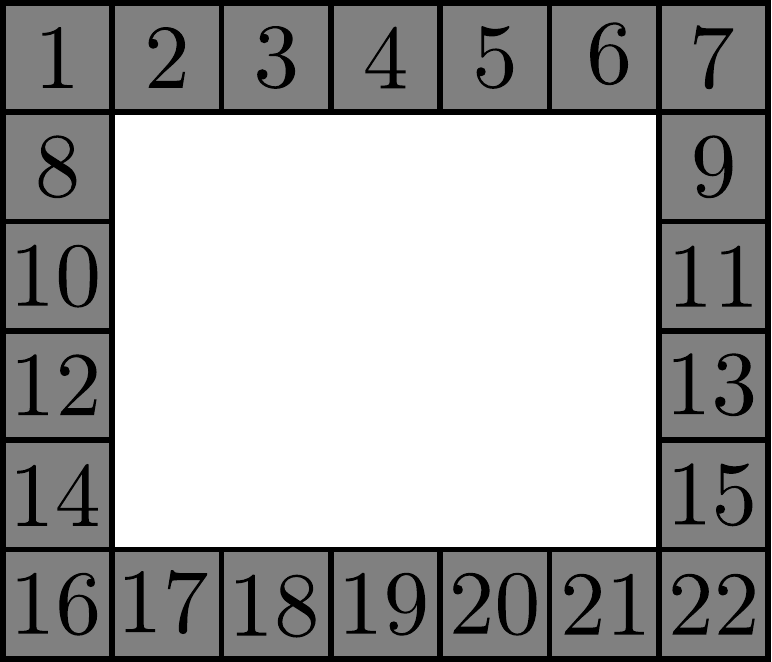}
        \caption{}
        \label{fig:4sidedLemma1}
    \end{subfigure}
    \begin{subfigure}[b]{0.4\textwidth}
		\centering        
        \includegraphics[width=2.2in]{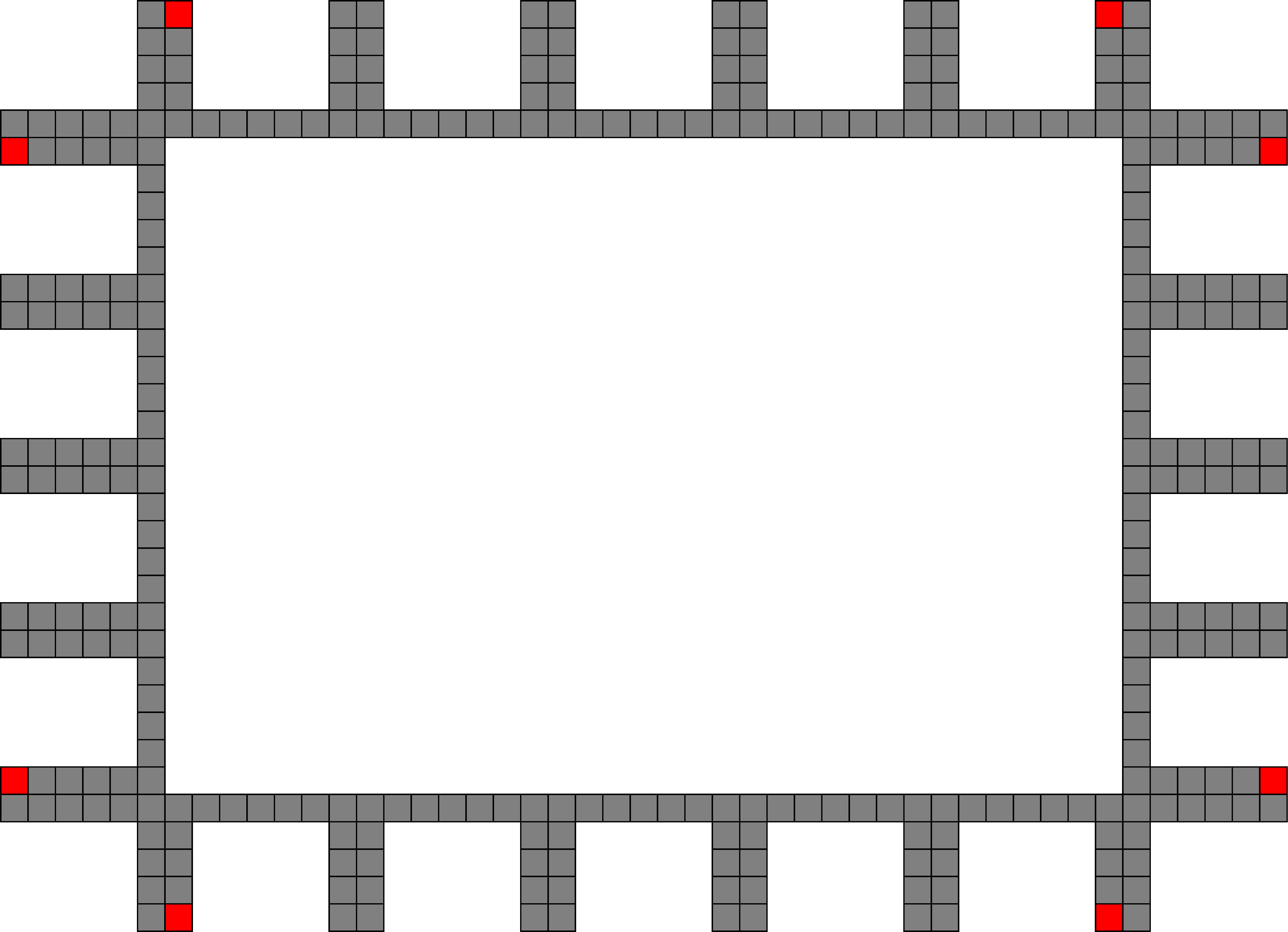}
        \caption{}
        \label{fig:4sidedLemma2}
    \end{subfigure}
    
    \caption{(a) An example generator for the $4$-sided fractals considered in Lemma~\ref{lem:4sided}. (b) A depiction of \fractal{X'}{2}. Red squares indicate possible locations of \indicating\ glues $\hat{g}_i$ and $g_i$.}\label{fig:4sidedLemma}
\end{figure}

\subsubsection{The \grout\ tile types for Lemma~\ref{lem:4sided}.}

With the ``base case'' hard-coded to give \fracasm{\Gamma}{2}{i}, we are now ready to describe \grout\ tiles. \grout\ tiles will be almost identical to the \grout\ tiles described in Section~\ref{sec:carpet} with the exception that now the \grout\ tiles must hard-code analogous though elongated versions of \grout\ supertiles from Section~\ref{sec:carpet}. For example, elongated version of \starter\ supertiles that initiate the binding of \grout\ tiles to \fracasm{\Gamma}{2}{1} is shown on the left in Figure~\ref{fig:elongated-init}. \grout\ tiles of $T$ are hard-coded to form similar ``elongated'' versions of \grout\ supertiles to those described in Section~\ref{sec:carpet}. The only difference being that now these supertiles must span a distance of $w_G$ between easternmost or westernmost tiles of \fracasm{\Gamma}{2}{i} and must span a distance of $h_G$ between northernmost or southernmost tiles of \fracasm{\Gamma}{2}{i} in order to cooperatively bind. 

\begin{figure}[htp]
    \centering
        \includegraphics[width=1.7in]{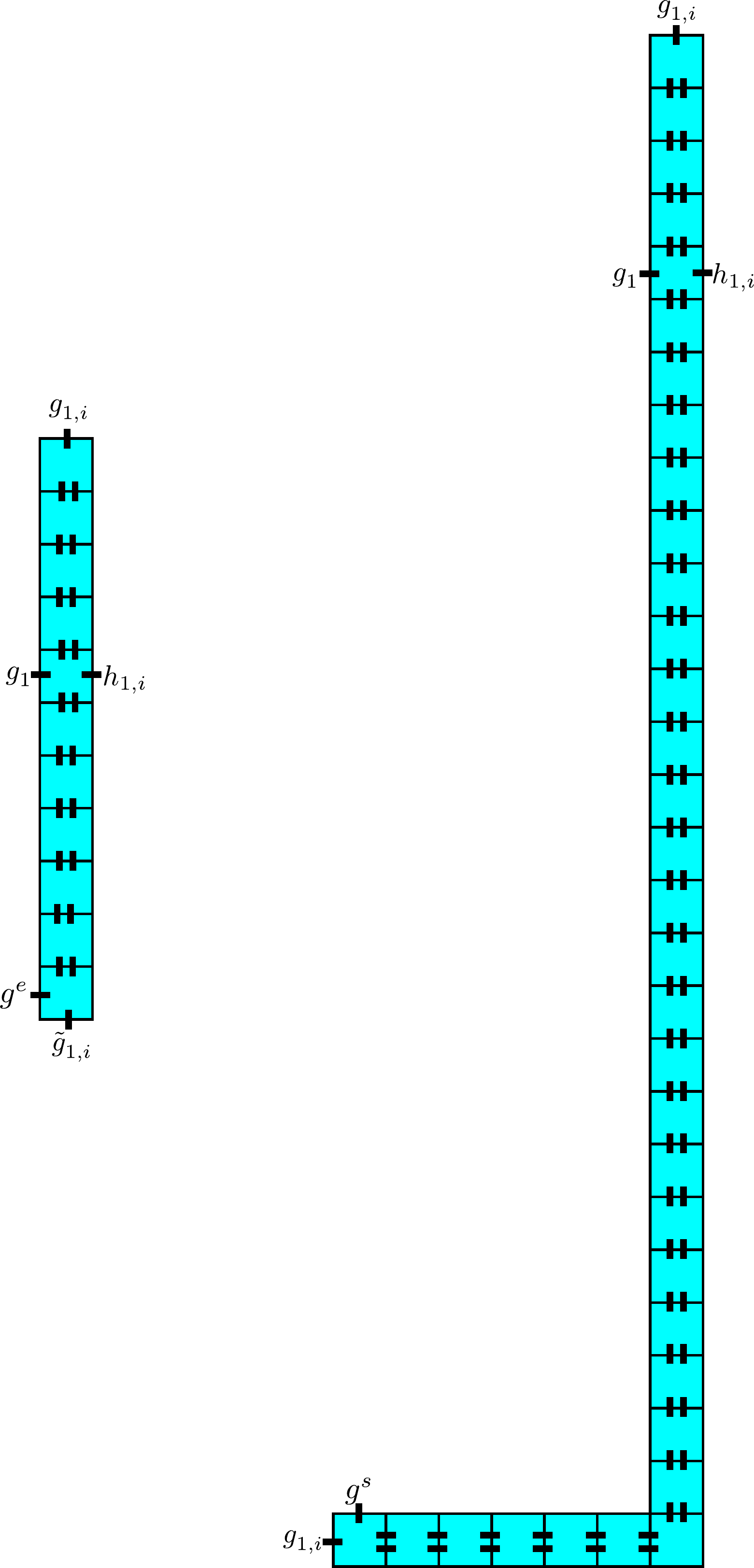}
    \caption{``Elongated'' versions of the supertiles that initiates the attachment of \grout\ tiles to a supertile \fracasm{\Gamma}{s}{1} (left), where $s \geq 3$, or \fracasm{\Gamma}{2}{1} (right). These are elongated versions of the \starter\ supertiles shown in Figure~\ref{fig:carpet-grout1-init}.}\label{fig:elongated-init}
\end{figure}

Now, \grout\ tiles fall into $r$ different classes where each class corresponds to a position in $G$. For some class $j$ between $1$ and $r$ (inclusive), \grout\ tiles of class $j$ bind to \fracasm{\Gamma}{2}{i} for each $i$ such that $1 \leq i \leq r$. Then, \grout\ tiles bind to the \indicating\ glues of edges of tiles of \fracasm{\Gamma}{2}{i} in the key locations described above, the resulting supertiles, which we call \fracasm{\Gamma}{2}{(i,j)}, further expose \stagebinding\ glues on edges of tiles adjacent to tiles in key locations such that the presence of these glues enables the supertiles \fracasm{\Gamma}{2}{(i,j)} to bind and form a supertile that corresponds to the subsequent stage \fractal{X}{3}. Moreover, once all \fracasm{\Gamma}{2}{(i,j)} supertiles bind, a \starter\ supertile (like the one depicted on the left in Figure~\ref{fig:elongated-init}) can then initiate the binding of more \grout\ tiles. Furthermore, by defining certain \stagebinding\ glues to have strength $0$, analogous to Table~\ref{tbl:binding-glues}, we can enforce that such a supertile that initiates the binding of \grout\ tiles (\starter\ supertiles) can bind only after all \fracasm{\Gamma}{2}{(i,j)} supertiles are subassemblies of the same supertile. We call this latter supertile, that corresponds to \fractal{X}{3}, \fracasm{\Gamma}{3}{j}. For a stage $s> 3$, the self-assembly of supertiles, \fracasm{\Gamma}{s}{j}, which correspond to \fractal{X}{s} is similar to the self-assembly of supertiles \fracasm{C}{s}{j} for the Sierpi\'{n}ski carpet given in Section~\ref{sec:carpet-details}. Finally, glue definition similar to Table~\ref{tbl:indicator-glues} can be given for \grout\ tiles so that appropriate \indicating\ glues are exposed by tiles belonging to \fracasm{\Gamma}{3}{j} to ensure that \fracasm{\Gamma}{3}{j} exposes \indicating\ glues so that the next iteration of \grout\ supertiles to bind expose \stagebinding\ glues in specific locations. These specific locations are chosen so that for $s\geq 2$, the distance between the \indicating\ glues of some \fracasm{\Gamma}{s}{j} is a strictly increasing function of $s$, which ensures that two such supertiles can bind iff they correspond to the same stage of the fractal $\bX$.

Similar to the Sierpi\'{n}ski carpet construction, we can see that the \initializer\ tiles self-assemble supertiles that correspond to \fractal{X}{2} and that \grout\ tiles can attach to supertiles that correspond to \fractal{X}{s} for some stage $s\geq 2$ to form supertiles that bind to yield a supertile corresponding to $X^{s+1}$.  Therefore, with tiles $T$, the 2HAM system $\mathcal{T} = (T, 2)$ finitely self-assembles $\bX$. Therefore, Lemma~\ref{lem:4sided} holds. Now we are ready to prove Theorem~\ref{thm:four-sided}. 
\end{proof}

\subsection{Proof of Theorem~\ref{thm:four-sided} (Sketch)}

Let $\bX$ be a $4$-sided dssf with generator $G$ and let $r = |G|$. In this section, we give a sketch of the proof of Theorem~\ref{thm:four-sided} by describing how to modify the tile set give in the proof of Lemma~\ref{lem:4sided} to obtain a tile set $T$ such that the 2HAM TAS $\mathcal{T} = (T,2)$ finitely self-assembles $\bX$. Figure~\ref{fig:4sided1} gives an example of a generator $G$ where we enumerate the points of $G$ from left to right, from top to bottom. Now let $G_{int} = G \setminus (L_G\cup R_G\cup T_G \cup B_G)$ (i.e. the points of $G$ that are not on the perimeter of $G$), and let $G_{bdry}$ be $G\setminus G_{int}$.

\begin{figure}[htp]
 \centering
    \begin{subfigure}[b]{0.4\textwidth}
    		\centering
        \includegraphics[width=0.5in]{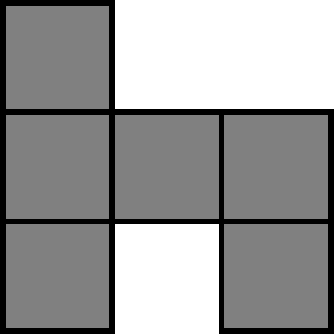}
	\caption{}
	\label{fig:4sided-junk}
    \end{subfigure}
    \qquad
    \begin{subfigure}[b]{0.4\textwidth}
		\centering        
        	\includegraphics[width=2.4in]{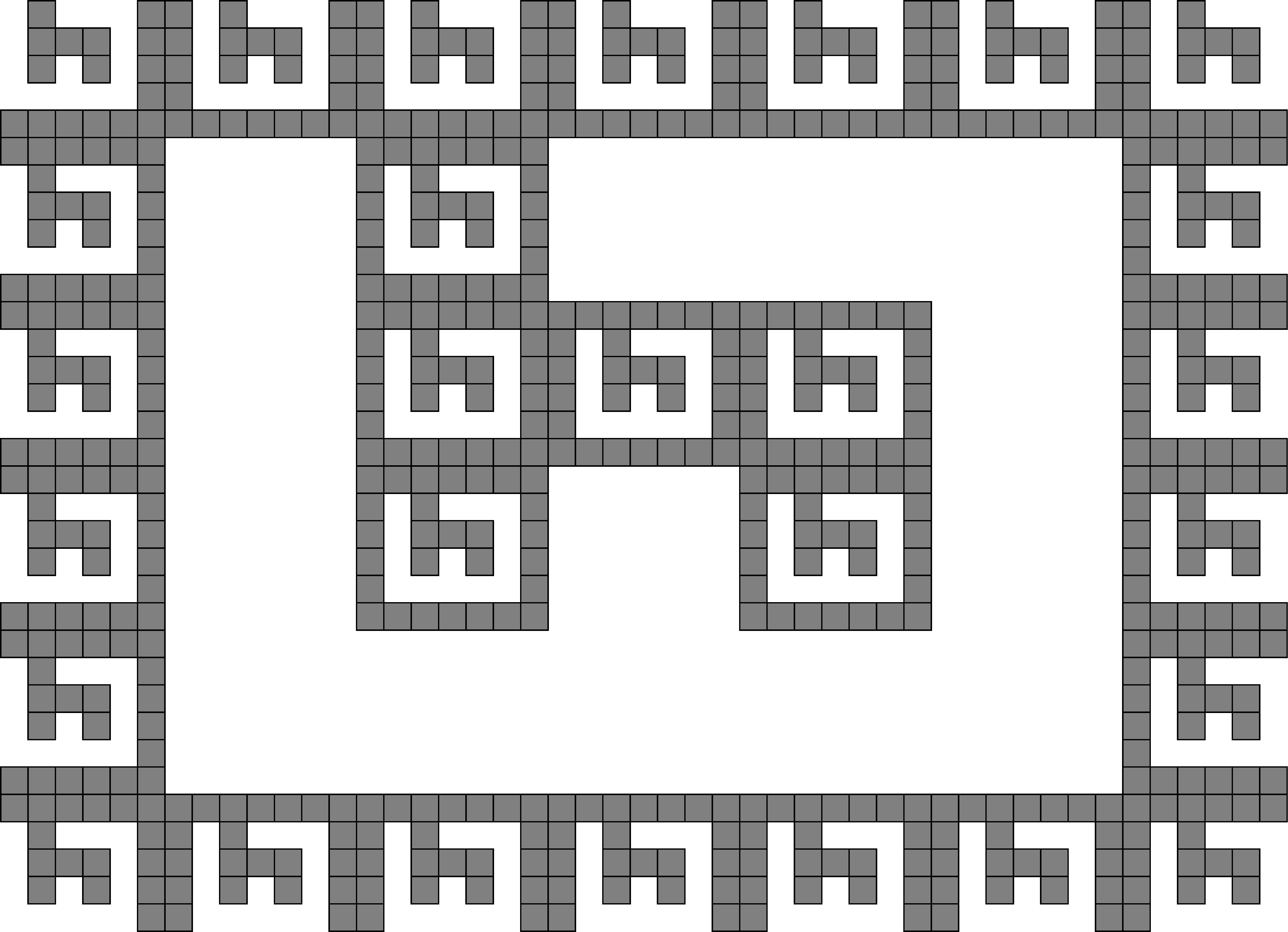}
	\caption{}
        \label{fig:4sided-junk2}
    \end{subfigure}
    
    \caption{(a) A depiction of $\mathcal{G}^-_1$ for the generator in Figure~\ref{fig:4sided1}. (b) A depiction of $\mathcal{G}^-$ for the generator in Figure~\ref{fig:4sided1}.}\label{fig:4sided-stages}

\end{figure}

\begin{figure}[htp]
 \centering
    	\includegraphics[width=2.4in]{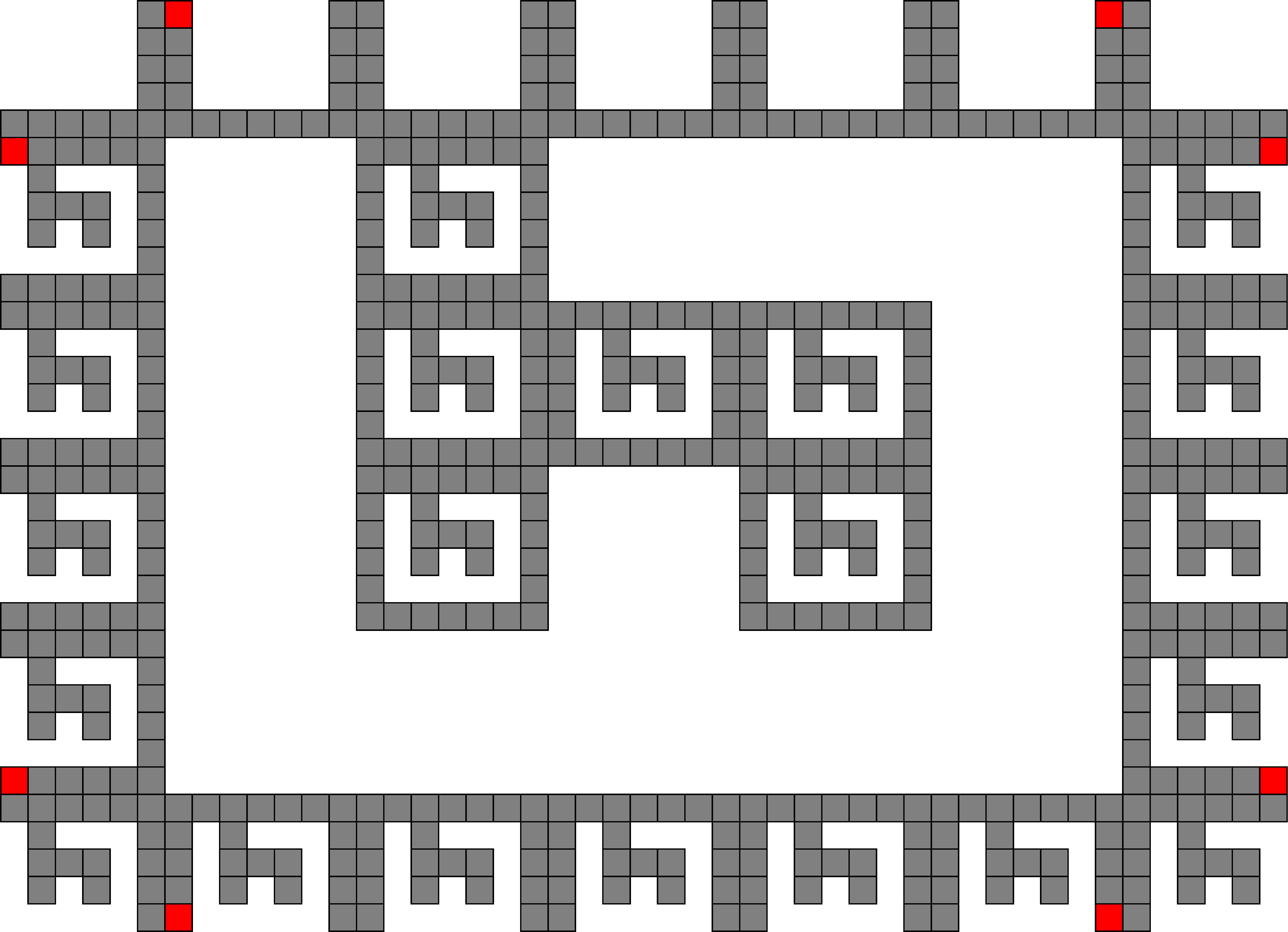}
    \caption{A depiction of \fracasm{\Gamma}{2}{i}. This is the portion of the second stage of the fractal with generator in Figure~\ref{fig:4sided1} that is hard-coded to self-assemble. It is analogous to the second stages that assemble shown in Figure~\ref{fig:4sidedLemma2} for the construction for Lemma~\ref{lem:4sided}. }\label{fig:4sided2-special}

\end{figure}

\begin{figure}[htp]
	\centering        
	\includegraphics[width=3in]{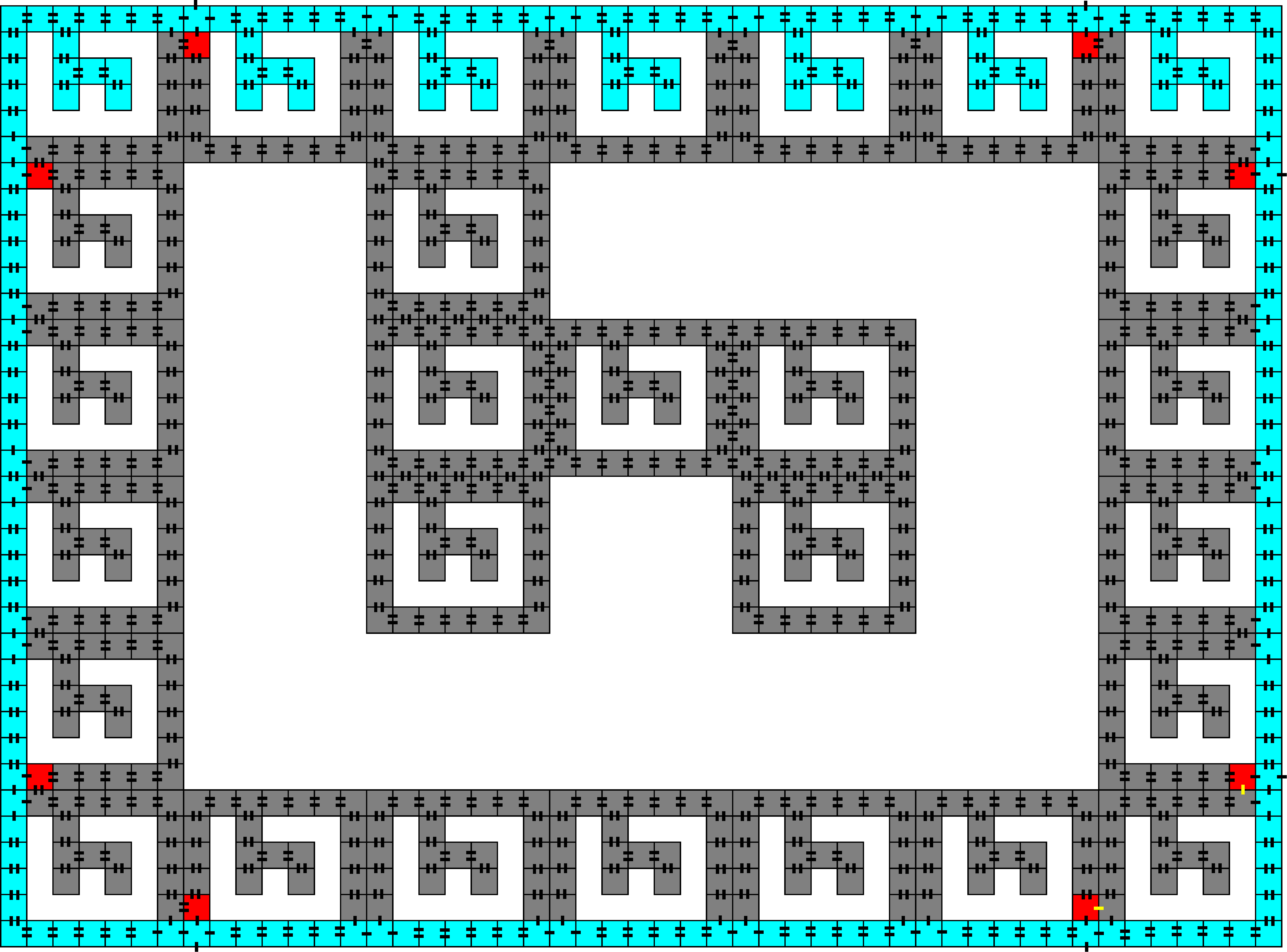}
    \caption{ A depiction of \fracasm{\Gamma}{2}{(12,j)} for some $j\in \N$ corresponding to the $j$ class of \grout. Note the glues that are exposed on tiles adjacent to tiles with \indicating\ glues (red tiles). In this case, as position $12$ in the generator $G$ is in $G_{int}$, \grout\ supertiles bind to on all four sides of \fracasm{\Gamma}{2}{12}. \grout\ supertiles that bind to \indicating\ glues expose \stagebinding\ glues which allow \fracasm{\Gamma}{2}{(12,j)} to bind in position $12$ during the self-assembly of a \fracasm{\Gamma}{3}{j} supertile.}
	\label{fig:4sided2-12-grout}
\end{figure}

\begin{figure*}[htp]
	\centering        
	\includegraphics[width=5in]{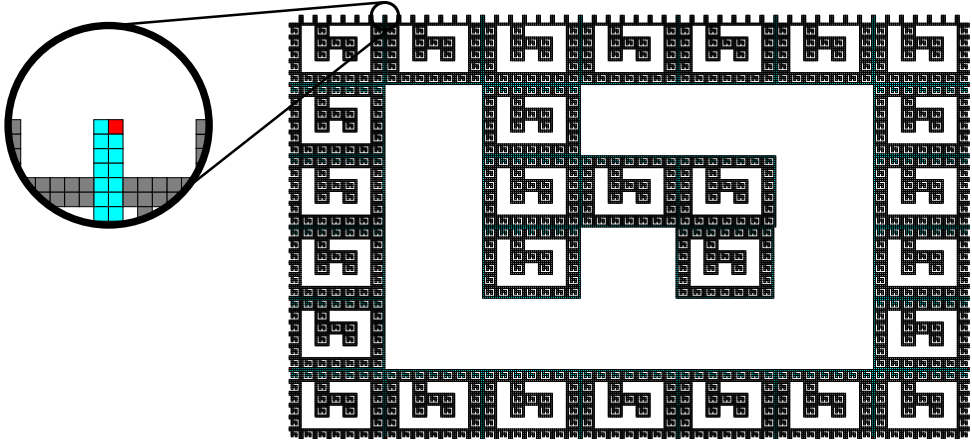}
    \caption{A schematic picture of \fracasm{\Gamma}{3}{j}. Note the red tile locations where tiles with \indicating\ glues (red tiles) will be present. Also note that \grout\ supertiles that bind to the northernmost tiles of the supertile depicted here hard-code the placement of tiles in locations corresponding to \fractal{X}{1}.}
	\label{fig:4sided3-details}
\end{figure*}

By Lemma~\ref{lem:4sided} there is a 2HAM system $\mathcal{T}'$ which finitely self-assembles the dssf with generator $G_{bdry}$. Let $T'$ be the tile set for $\mathcal{T}'$ as described in the construction for Lemma~\ref{lem:4sided}. We will show how to modify the tile set $T'$ to obtain $T$. 

\subsubsection{Self-assembly of stage $2$ for $4$-sided fractals}\label{sec:4sides-stage2}

Let $\mathcal{G}_1$ denote the full grid-graph of $G$ and let $\mathcal{G}^-_1$ denote the full grid-graph of $G_{int}$. Note that it is not necessary for $\mathcal{G}_1^-$ to be connected. Also note that $\mathcal{G}^-_1$ may be empty if $G = L_G\cup R_G\cup T_G\cup B_G$ as in the case for the Sierpi\'{n}ski carpet dssf. An example of $\mathcal{G}^-_1$ for the generator shown in Figure~\ref{fig:4sided1} is shown in Figure~\ref{fig:4sided-junk} where vertices correspond to squares and there is assumed to be an edge between two vertices iff these squares abut. Now let $\mathcal{G}$ denote the full grid-graph of \fractal{X}{2}. Let $\mathcal{G}^-$ be the (not necessarily connected) graph obtained by removing the northernmost, southernmost, easternmost, and westernmost points from $\mathcal{G}$. For the generator given in Figure~\ref{fig:4sided1}, $\mathcal{G}^-$ is shown in Figure~\ref{fig:4sided-junk2}. Finally, let $\mathcal{G}_c$ be the connected component of $\mathcal{G}^-$ that is not equal to a connected component of $\mathcal{G}^-_1$ up to translation. See Figure~\ref{fig:4sided2-special} for an example of $\mathcal{G}_c$  for the generator shown in Figure~\ref{fig:4sided1}.

Then, the \initializer\ tiles of $T$ are hard-coded to self-assemble $r$ different versions of $\mathcal{G}_c$ which we call \fracasm{\Gamma}{2}{i} for $1\leq i \leq r$. Similar to the \initializer\ tiles described in the proof of Lemma~\ref{lem:4sided}, each \fracasm{\Gamma}{2}{i} contains tiles in key locations (defined as in Lemma~\ref{lem:4sided}) that expose \indicating\ glues that depend on the value of $i$. These \initializer\ tiles can be thought of as being equivalent to the \initializer\ tiles of $T'$, appropriately modified with additional glues and additional tiles that hard-code the stage $1$ subassemblies of \initializer\ supertiles whose positions in the \fracasm{\Gamma}{2}{i} correspond to the points of $G_{int}$. In the example in Figure~\ref{fig:4sided2-special}, these additional tiles self-assemble at locations $9$, $12$, $13$, $14$, $17$, and $18$ within stage-$1$ subassemblies at locations $8$ through $28$, as well as self-assemble entire stage-$1$ subassemblies at locations $9$, $12$, $13$, $14$, $17$, and $18$. Figure~\ref{fig:4sided2-special} depicts the locations of tiles of \fracasm{\Gamma}{2}{i} for the generator in Figure~\ref{fig:4sided1}, where red tiles may contain edges with \indicating\ glues. 

\subsubsection{Tile types for \grout\ tiles.}

The \grout\ tile types of $T$ consist of tile types that are equivalent to the \grout\ tile types of $T'$ with additional glues along with additional tile types that hard-code the appropriate stage $1$ growth that complete any subassembles that represent \fractal{X}{1}. Figure~\ref{fig:4sided2-12-grout} gives an example of \fracasm{\Gamma}{2}{12} with complete grout. In this particular example, \grout\ tiles have been hard-coded to place tiles in locations corresponding to \fractal{X}{1} as the \grout\ tiles bind to the northernmost tiles of \fracasm{\Gamma}{2}{12}. \grout\ tiles are added for each $i$ between $1$ and $r$ (inclusive) and as in Figure~\ref{fig:4sided2-12-grout}, \grout\ tiles may bind to some \fracasm{\Gamma}{2}{i} where $i$ corresponds to a point in $G_{int}$. In this case, \grout\ tiles can be defined to completely surround \fracasm{\Gamma}{2}{i} (or \fracasm{\Gamma}{s}{i} for $s > 2$) and expose appropriate \stagebinding\ glues at key locations. \stagebinding\ glues ensure that for all $i$ and $j$ both between $1$ and $8$ (inclusive), once a sufficient number of \grout\ tiles bind to each \fracasm{\Gamma}{2}{i}, the resulting supertiles, which we again call \fracasm{\Gamma}{2}{(i,j)} (or \fracasm{\Gamma}{s}{(i,j)} for $s > 2$) can bind to yield a supertile corresponding to \fractal{X}{3} (or \fractal{X}{s} for $s>2$). We call this latter supertile \fracasm{\Gamma}{3}{j} (or \fracasm{\Gamma}{s+1}{j} for $s>2$). Figure~\ref{fig:4sided3-details} depicts \fracasm{\Gamma}{3}{j}. 

As $\mathcal{T}$ is based on $\mathcal{T}'$, the assembly sequences of each system share similarities that are important to note.  For a stage $s\in \N$, and $j$ such that $1\leq j \leq 8$, let \fracasm{\Gamma'}{s}{j} be the supertile producible in $\mathcal{T}'$ corresponding to $X'^s$. Note that as the tile types in $T$ are based on tile types in $T'$, in an assembly sequence for \fracasm{\Gamma}{s}{j}, the tiles in \fracasm{\Gamma}{s}{j} with locations (up to some fixed positioning of the supertile) corresponding to points of $G_{bdry}$ (at any stage) must bind in an order corresponding to some assembly sequence of \fracasm{\Gamma'}{s}{j}. In other words, the portion of the fractal $\bX$ equal to $\bX'$ must self-assemble following an assembly sequence in $\mathcal{T}$ analogous to an assembly sequence in $\mathcal{T}'$. The analogous assembly sequence can be obtained by ignoring any supertile combinations that involve a supertile corresponding to points of $G_{int}$ at any stage. Therefore, $\bX'$ finitely self-assembles in $\mathcal{T}$. The additional \initializer\ and \grout\ tiles are defined to ``fill in'' tile locations in $\bX$ that are not in $\bX'$ by nondeterministically binding, following one of many possible assembly sequences. 

Finally, the \initializer\ tiles assemble a supertile that corresponds to \fractal{X}{2}, and \grout\ supertiles tiles can attach to supertiles that correspond to \fractal{X}{s} for some stage $s\geq 2$ to form supertiles that bind to yield a supertile corresponding to \fractal{X}{s+1}. Therefore, with tiles $T$, the 2HAM system $\mathcal{T} = (T, 2)$ finitely self-assembles $\bX$. Therefore, Theorem~\ref{thm:four-sided} holds.

\section{A $3$-sided Fractal That Does Not Finitely Self-assemble}\label{sec:impossibility-overview}

In this section we prove that there exist $3$-sided fractals that do not finitely self-assemble in the 2HAM. 

\begin{theorem}\label{thm:three-sided}
  There exists a $3$-sided fractal $\bX$ for which there is no 2HAM TAS $\calT_{\bX} = (T, \tau)$ that finitely self-assembles $\bX$ for any temperature $\tau \in \N$.
\end{theorem}

To prove Theorem~\ref{thm:three-sided}, we consider the fractal with generator
$G=\{(0,4),(1,4),$ $(2,4),$ $(3,4),$ $(0,3),$ $(2,3),$ $(0,2),$ $(2,2),$ $(0,1),$ $(0,0),$ $(1,0),$ $(2,0),$ $(3,0)\}$. Stages 1 and 2 of this fractal are shown in Figure~\ref{fig:3sided-impossible-stages}.  We refer to this fractal as $\bX$. For a stage $s\in \N$, we refer to the $i^{th}$ position of \fractal{X}{s} as \fracasm{X}{s}{i} where $1\leq i\leq 13$ (Figure~\ref{fig:3sided-impossible1}).  We call a supertile with shape \fractal{X}{s} \fractal{\gamma}{s}, and when such a supertile is a subassembly of some \fractal{\gamma}{s+1} and corresponds to points location $i$, we denote such assemblies by \fracasm{\gamma}{s}{i}.

\begin{figure}[htp]
  \centering
  \begin{subfigure}[b]{0.14\textwidth}
    \centering
    \includegraphics[width=6pc]{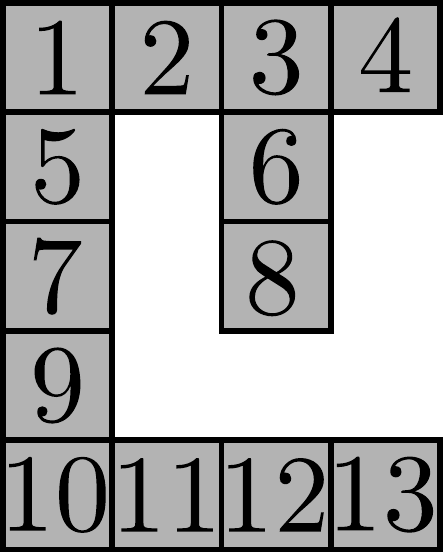}
    \caption{Stage $1$}\label{fig:3sided-impossible1}
  \end{subfigure}
  \quad\quad\quad
  \begin{subfigure}[b]{0.14\textwidth}
    \centering
    \includegraphics[width=6pc]{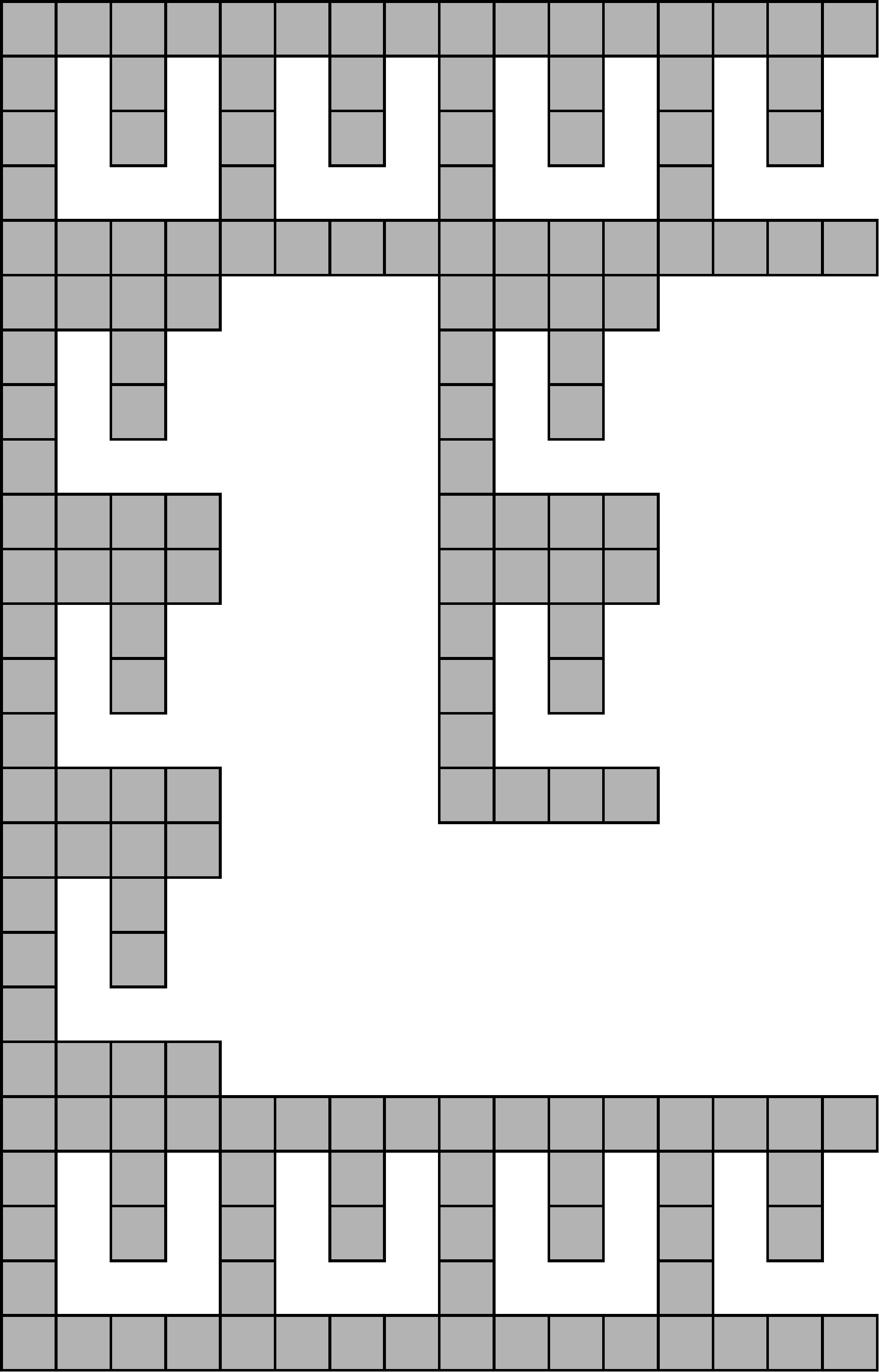}
    \caption{Stage $2$}\label{fig:3sided-impossible2}
  \end{subfigure}
  \caption{\fractal{X}{1} and \fractal{X}{2}}\label{fig:3sided-impossible-stages}
\end{figure}

\begin{figure}[htp]
  \centering
  \begin{subfigure}[b]{0.14\textwidth}
    \centering
    \includegraphics[width=6pc]{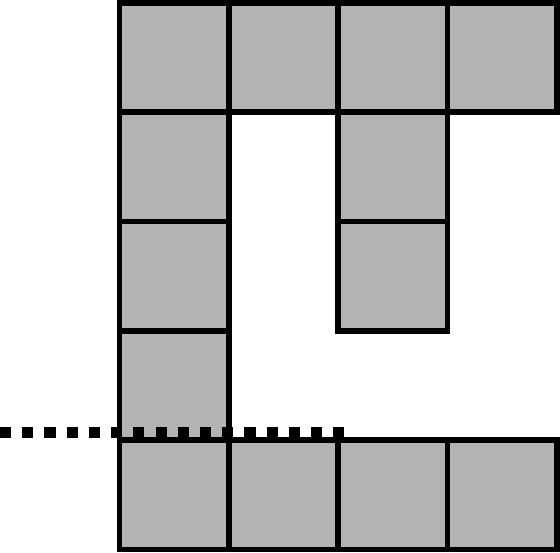}
    \caption{}\label{fig:3sided-impossible1-cut}
  \end{subfigure}
  \quad\quad\quad
  \begin{subfigure}[b]{0.14\textwidth}
    \centering
    \includegraphics[width=6pc]{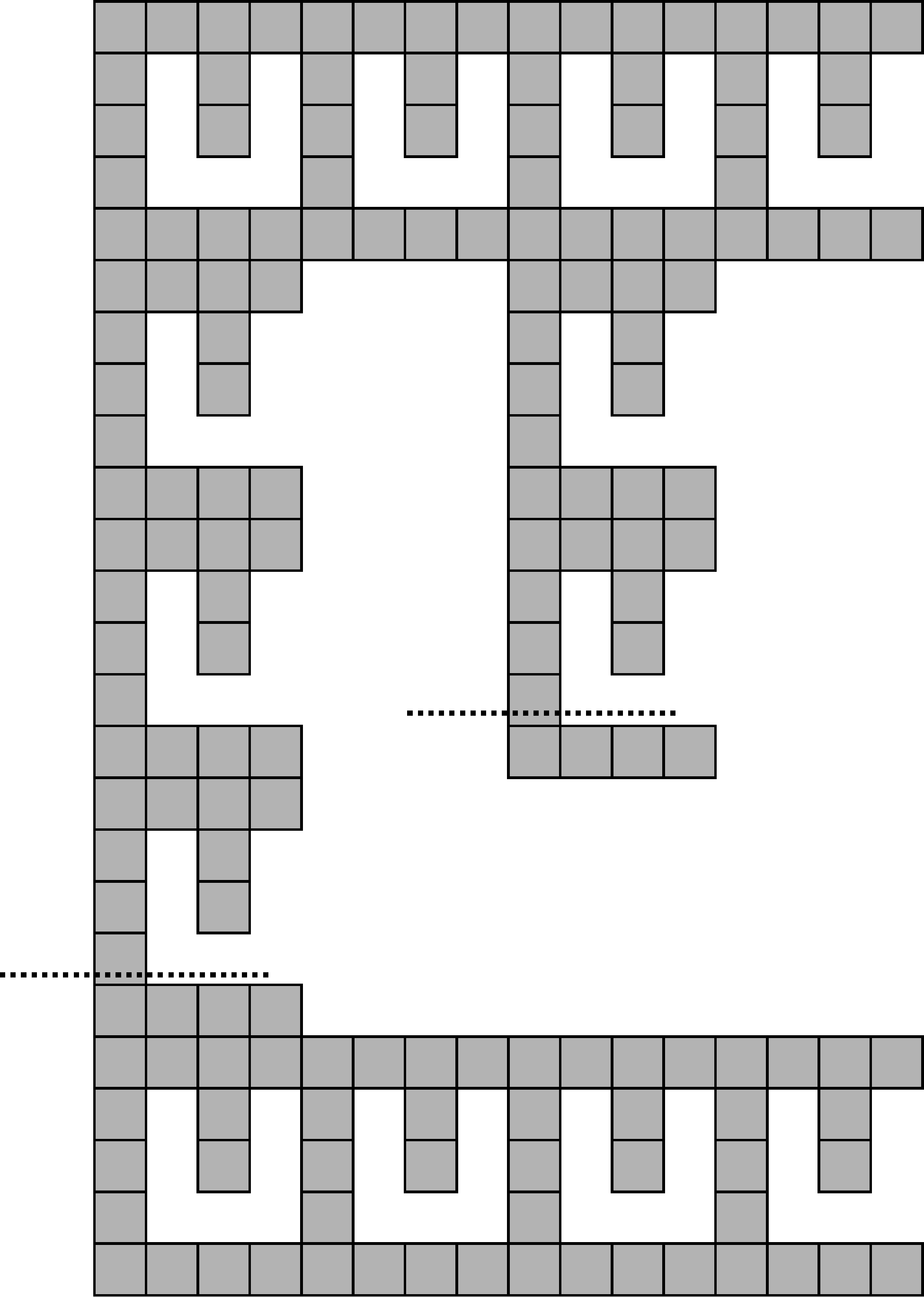}
    \caption{}\label{fig:3sided-impossible2-cut}
  \end{subfigure}
  \caption{Strength $\tau$ cuts in \fractal{\gamma}{1} and \fractal{\gamma}{2}}\label{fig:3sided-impossible-cuts}
\end{figure}

\begin{figure}[htp]
  \centering
  \includegraphics[scale=.07]{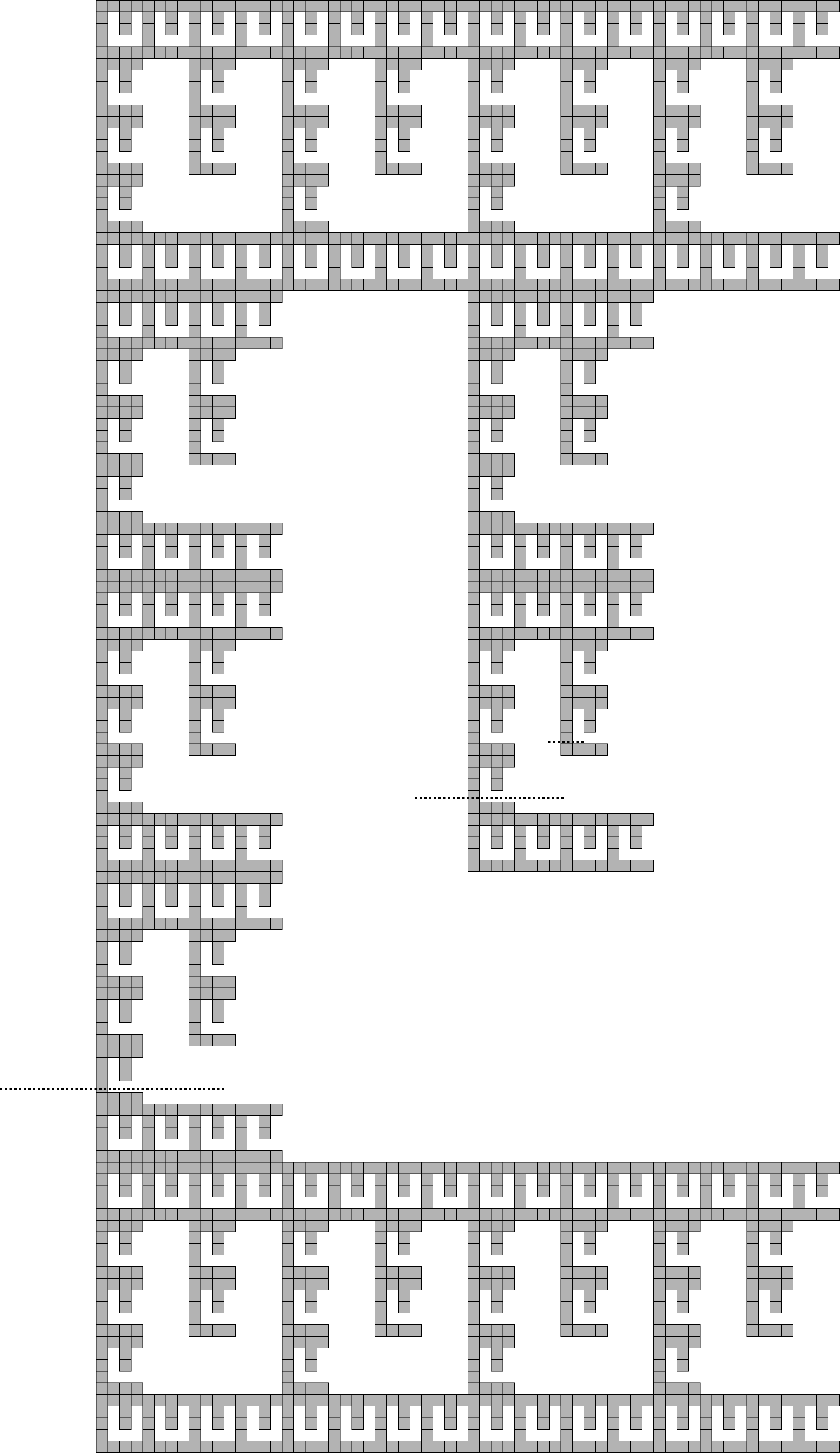}
  \caption{There are at least $s$ strength $\tau$ cuts within each supertile $\gamma^s$.  Here $\gamma^3$ with 3 strength $\tau$ cuts is shown.  The subassembly to the south of the rightmost cut is referred to as $\beta_1$, the subassembly to the south of the next rightmost cut as $\beta_2$, and the subassembly to the south of the leftmost cut as $\beta_3$.}\label{fig:3sided-impossible-cuts-2}
\end{figure}

\begin{figure}[htp]
  \centering
  \includegraphics[scale=.025]{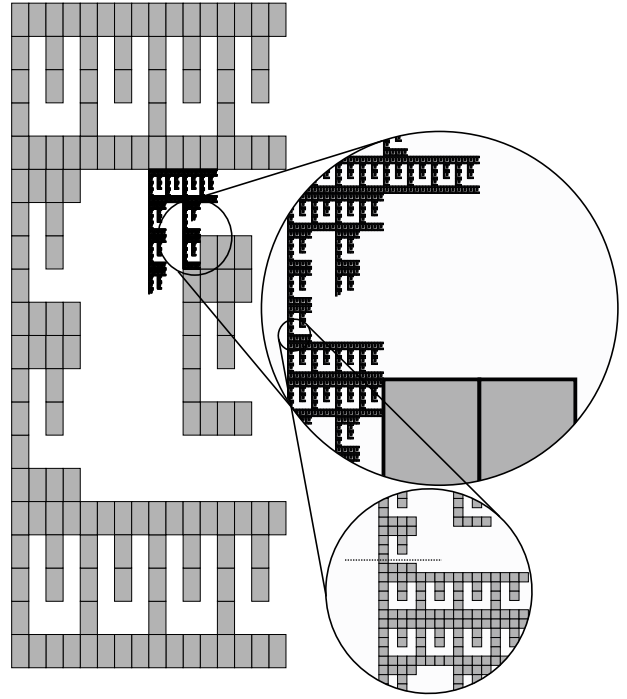}
  \caption{An example of erroneous binding within \fractal{\gamma}{5}. Because of the large number of tiles some of the \fractal{\gamma}{3} subassemblies are shown as rectangles. In this example, a $\tau$ strength cut is shown in the bottom right circle. The subassembly of \fractal{\gamma}{5} containing the tile to the north of this cut is $\alpha_2$ and the subassembly containing the tile to the south of this cut is $\beta_{s'}$.}\label{fig:3sided-impossible5-2}
\end{figure}

For the sake of contradiction, assume that $\calT_{\bX} = (T, \tau)$ is a 2HAM TAS such that $\bX$ finitely self-assembles in $\calT_{\bX}$. Consider any 2HAM TAS $\calT_{\bX} = (T, \tau)$. We show that $\calT_{\bX}$ does not finitely self-assemble $\bX$ by showing that there is a producible supertile $\alpha\in\mathcal{A}_{\Box}[\calT_{\bX}]$ that does not have the shape of of any subset of $\bX$. 

Then, for any $s\in \N$, and for every supertile $\alpha$ such that $\alpha$ contains a $\gamma^s$ subassembly, there is a stage 1 subassembly $\fractal{\gamma}{1}$ of  $\fracasm{\gamma}{s}{9}$ such that this stage 1 subassembly contains a strength $\tau$ cut between \fracasm{\gamma}{1}{9} and \fracasm{\gamma}{1}{10} that separates some \fracasm{\gamma}{s}{10}, \fracasm{\gamma}{s}{11}, \fracasm{\gamma}{s}{12}, and \fracasm{\gamma}{s}{13} subassemblies, along with a sequence of subassemblies \fracasm{\gamma}{i}{10},  \fracasm{\gamma}{i}{11}, \fracasm{\gamma}{i}{12}, and \fracasm{\gamma}{i}{13}, $i<s$, from the rest of \fractal{\gamma}{s}.  For an example of such a cut, see the bottom left cuts shown in Figure~\ref{fig:3sided-impossible2-cut} for $s=2$ and in Figure~\ref{fig:3sided-impossible-cuts-2} for $s=3$.  Then note that for any $s>2$, \fracasm{\gamma}{s}{8} has a \fractal{\gamma}{s-1} subassembly which contains a similar strength $\tau$ cut between two tiles \fracasm{\gamma}{1}{9} and \fracasm{\gamma}{1}{10} in the \fractal{\gamma}{1} subassembly directly above \fracasm{\gamma}{s-1}{10}. 

Then \fracasm{\gamma}{s}{8} has a subassembly \fractal{\gamma}{1} which contains a single strength $\tau$ cut between \fracasm{\gamma}{1}{9} and \fracasm{\gamma}{1}{10} (shown as the cut on the right in Figure~\ref{fig:3sided-impossible2-cut}).  We also note that when $s=1$ there is one strength $\tau$ cut between \fracasm{\gamma}{s}{9} and \fracasm{\gamma}{s}{10}. Therefore every supertile $\alpha$ such that there exists $A\in \alpha$ with $\fractal{X}{s}\subseteq A$ contains a sequence of $s$ strength $\tau$ cuts between positions 9 and 10 of $s$ distinct stage 1 subassemblies.  An example of this for $s=3$ is shown in Figure~\ref{fig:3sided-impossible-cuts-2}.

Let $g$ be the number of tiles in $T$.  Consider a producible supertile $\alpha$ such that there exists $A\in \alpha$ with $\fractal{X}{g+2}\subseteq A$.  Within $\alpha$ there is a \fractal{\gamma}{g+2} subassembly with some \fracasm{\gamma}{g+1}{6} as a subassembly.  As we have shown, this \fracasm{\gamma}{g+1}{6} contains a sequence of $g+1$ strength $\tau$ cuts, each consisting of a single glue. By the pigeonhole principle, there are at least two such cuts that consist of the same single $\tau$ strength glue.  Let the subassembly to the south of the cut within \fractal{\gamma}{1} be called $\beta_1$, the subassembly to the south of the cut within \fracasm{\gamma}{2}{9} be called $\beta_2$, etc., with the subassembly to the south of the cut within \fracasm{\gamma}{g+1}{9} called $\beta_{g+1}$ (see Figure~\ref{fig:3sided-impossible-cuts-2} for an example of $\beta_1$, $\beta_2$, and $\beta_3$).  
Consider two cuts directly above $\beta_s$ and $\beta_{s'}$ with $s'>s$ that contain the same glue.  Let $\alpha_2$ be $\alpha$ with subassemlies of $\beta_s$, $\beta_{s+1}$, $\dots$, $\beta_{s'}$ removed.  We will show that $\alpha_2$ and $\beta_{s'}$ are producible assemblies.  Additionally, we notice that between \fracasm{X}{g+2}{8} and \fracasm{X}{g+2}{12} there is enough room to fit an entire stage \fractal{X}{g+1}, and since $s'\leq g+1$, erroneous binding of $\alpha_2$ and $\beta_{s'}$ cannot be prevented. Figure~\ref{fig:3sided-impossible5-2} depicts an example of such erroneous binding within a $\gamma^5$ supertile.  Hence $\alpha_2$ and $\beta_{s'}$ are $\tau$-combinable into some supertile $\chi\in\mathcal{A}[\calT_{\bX}]$. Then, note that for all $A\in \chi$, the set of all tile locations of tiles in $A$ is not contained in $\subseteq\bX$. Therefore, $\calT_{\bX}$ does not finitely self-assembly $\bX$.

To complete the proof, we now show that the subassemblies $\alpha_2$ and $\beta_{s'}$ are producible. If one of $\alpha_2$ or $\beta_{s'}$ is not producible, then the binding graph of that one must contain a cut with strength less than $\tau$. However, since every $\beta_i$, $1\leq i\leq g+1$, is connected to $\alpha$ by a singe strength-$\tau$ glue between two single tiles, if the the binding graph of $\alpha_2$ or $\beta_{s'}$ contains a cut with strength less than $\tau$, then $\alpha$ would contain the same cut with strength less than $\tau$. This contradicts the assumption that $\alpha$ is producible.  Hence $\alpha_2,\beta_{s'}\in\mathcal{A}[\calT_{\bX}]$. Thus, Theorem~\ref{thm:three-sided} holds.

\section{Conclusion}\label{sec:conclusion}

Theorem~\ref{thm:four-sided} shows that any $4$-sided dssf finitely self-assembles in the 2HAM at temperature $2$ and with scale factor $1$. Theorem~\ref{thm:three-sided} shows that there exists a $3$-sided fractal that does not finitely self-assemble in any 2HAM system at any temperature. 
For a $4$-sided fractal generator $G$, the presence of a Hamiltonian cycle in the full grid graph of the points on the perimeter of $G$ proved helpful in our construction. Similar techniques to those described in Section~\ref{sec:four-sided-fractals} might be used to show that a much more general class of fractals finitely self-assembles in the 2HAM at temperature $2$ with scale factor $1$. In particular, a fractal belonging to this class can be described as having a 
generator such that 1) the full grid-graph of the generator contains a  Hamiltonian cycle through each point in the generator and 2) the northernmost, southernmost, easternmost, and westernmost points of the generator each contain $2$ distinct points. An example of such a fractal is shown in Figure~\ref{fig:conclusion-example}
\begin{figure}[htp]
  \begin{center}
        \includegraphics[width=.7in]{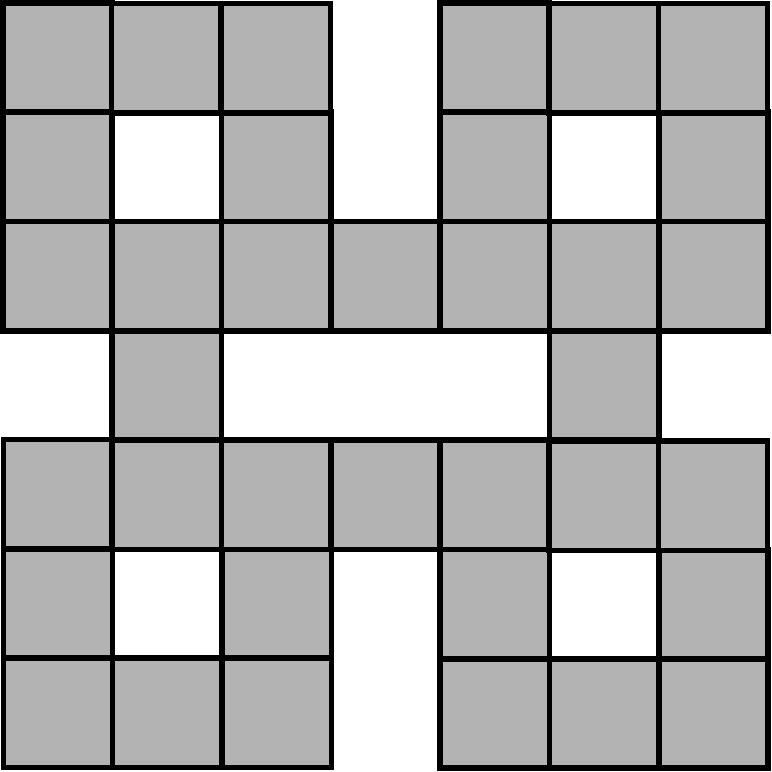}
  \end{center}
    \caption{Do fractals with generators like the one depicted in this figure finitely self-assemble in the 2HAM?}\label{fig:conclusion-example}
\end{figure}

\vspace{-20pt}

\section{Acknowledgements}\label{sec:acknowledgements}

The authors would like to thank the anonymous reviewers for their time and effort in helping to improve this paper.

\bibliographystyle{abbrv} 
\bibliography{tam,experimental_refs}

\ifx\arxiv\undefined

\else
\begin{appendices}
	\newpage
\onecolumn
\section{Tiles for Sierpinski's Carpet Construction}\label{sec:appendix-carpet-tiles}

We describe the supertiles that consist of \grout\ tiles for the Sierpinski's carpet construction. Tile types are defined so that eight different versions of each of the supertiles in each figure self-assemble, corresponding to the eight \grout\ classes. In each figure, $j\in \N$ is such that $1\leq j \leq 8$, and tiles of supertiles belong to \grout\ class $j$. Depending on the value of $j$, for $k \in \N$ such that $1\leq k \leq 8$, the glues $h_{k,j}$, $\hat{h}_{k,j}$, $h_{1,j}^*$, and $\hat{h}_{1,j}^*$ are defined to either have strength $1$ or $0$. Table~\ref{tbl:binding-glues} describes glue strengths for these glues for each $j$.  
In addition, depending on the value of $j$, for $p \in \{2, 4, 5, 7\}$, glues with labels $\hat{g}_{p,j}$ and $\bar{g}_{p,j}$ are defined in Table~\ref{tbl:indicator-glues}.

\ifx\arxiv\undefined

\definecolor{lightlightgray}{rgb}{.9,.9,.9}
\newcolumntype{g}{>{\columncolor{lightgray}}c}

\begin{table}[!htp]
\centering
  \begin{tabular}{ | g | c | }
    \hline\rowcolor{lightgray}
    $j$ & glues with strength $0$ \\ \hline               
    $1$ & $h_{5,j}$, $\hat{h}_{7,j}$  \\ \hline 
    $2$ & $h_{5,j}$, $\hat{h}_{7,j}$ \\ \hline
    $3$ & $\hat{h}_{4,j}$, $\hat{h}_{6,j}$ \\ \hline
    $4$ & $h_{2,j}$, $\hat{h}_{3,j}$ \\ \hline
    $5$ & $h_{1,j}$, $\hat{h}_{1,j}^*$ \\ \hline
    $6$ & $h_{2,j}$, $\hat{h}_{3,j}$ \\ \hline
    $7$ & $h_{2,j}$, $\hat{h}_{3,j}$ \\ \hline
    $8$ & $h_{1,j}$, $\hat{h}_{1,j}^*$ \\ \hline 
  \end{tabular}
\caption{For $j\in \N$ such that $1\leq j \leq 8$, this table lists those glues defined to have strength $0$. For all $k \in \N$ such that $1\leq k \leq 8$, $h_{k,j}$, $\hat{h}_{k,j}$, $h_{1,j}^*$, and $\hat{h}_{1,j}^*$ not listed in a row for a fixed value $j$ are defined to have strength $1$.}\label{tbl:binding-glues}
\end{table}

\begin{table}[!htp]
\centering
  \begin{tabular}{ | g | c | c | c | c | c | c | c | c | }
    \hline\rowcolor{lightgray}
    $j$ & $\hat{g}_{2,j}$ & $\bar{g}_{2,j}$ & $\hat{g}_{4,j}$ & $\bar{g}_{4,j}$ & $\hat{g}_{5,j}$ & $\bar{g}_{5,j}$ & $\hat{g}_{7,j}$ & $\bar{g}_{7,j}$ \\ \hline
    $1$ & $\hat{g}^n$ & $g^n$ & $\hat{g}^w$ & $g^w$ & $\hat{g}_1$ & $g_1$ & $g_1$ & $\hat{g}_1$ \\\hline
    $2$ & $\hat{g}^n$ & $g^n$ & $\hat{g}_2$ & $g_2$ & $g_2$ & $\hat{g}_2$ & $g^s$ & $\hat{g}^s$ \\\hline 
    $3$ & $\hat{g}^n$ & $g^n$ & $\hat{g}_3$ & $g_3$ & $g^e$ & $\hat{g}^e$ & $g_3$ & $\hat{g}_3$ \\\hline 
    $4$ & $\hat{g}_4$ & $g_4$ & $\hat{g}^w$ & $g^w$ & $g^e$ & $\hat{g}^e$ & $g_4$ & $\hat{g}_4$ \\\hline 
    $5$ & $\hat{g}_5$ & $g_5$ & $\hat{g}^w$ & $g^w$ & $g^e$ & $\hat{g}^e$ & $g_5$ & $\hat{g}_5$ \\\hline 
    $6$ & $\hat{g}_6$ & $g_6$ & $\hat{g}^w$ & $g^w$ & $g_6$ & $\hat{g}_6$ & $g^s$ & $\hat{g}^s$ \\\hline 
    $7$ & $\hat{g}^n$ & $g^n$ & $\hat{g}_7$ & $g_7$ & $g_7$ & $\hat{g}_7$ & $g^s$ & $\hat{g}^s$ \\\hline 
    $8$ & $\hat{g}_8$ & $g_8$ & $\hat{g}_8$ & $g_8$ & $g^e$ & $\hat{g}^e$ & $g^s$ & $\hat{g}^s$ \\\hline               
  \end{tabular}
\caption{For $j\in \N$ such that $1\leq j \leq 8$, this table gives glue definitions. For example, when $j=1$, $\hat{g}_{2,j} = \hat{g}^n$. All glues in this table are also defined to have strength $1$.}\label{tbl:indicator-glues}
\end{table}
\else

\fi

\subsection{\starter\ tile types}

Figures~\ref{fig:inits345} and~\ref{fig:inits678} depict \starter\ tile types.

\begin{figure}[htp]
    \centering
    
\begin{subfigure}{0.25\textwidth}
            \centering
            \includegraphics[width=\textwidth]{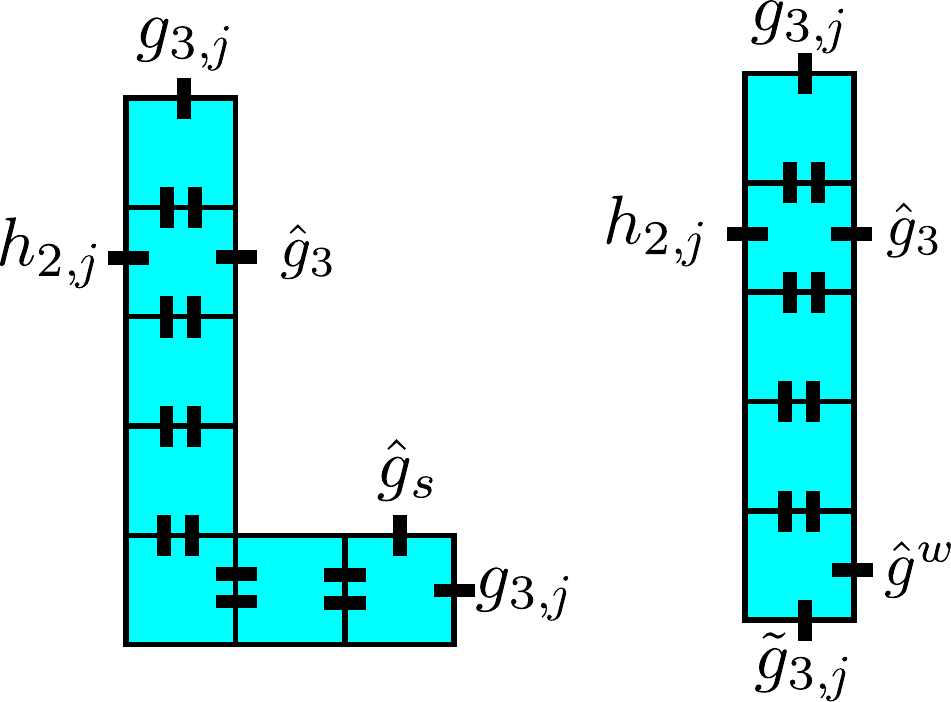}
            \caption{The supertiles that start the growth of \grout\ for $C^s_3$ for $s\geq 2$.}\label{fig:3-init}
          \end{subfigure}
                \qquad\qquad
\begin{subfigure}{0.25\textwidth}
            \centering
            \includegraphics[width=\textwidth]{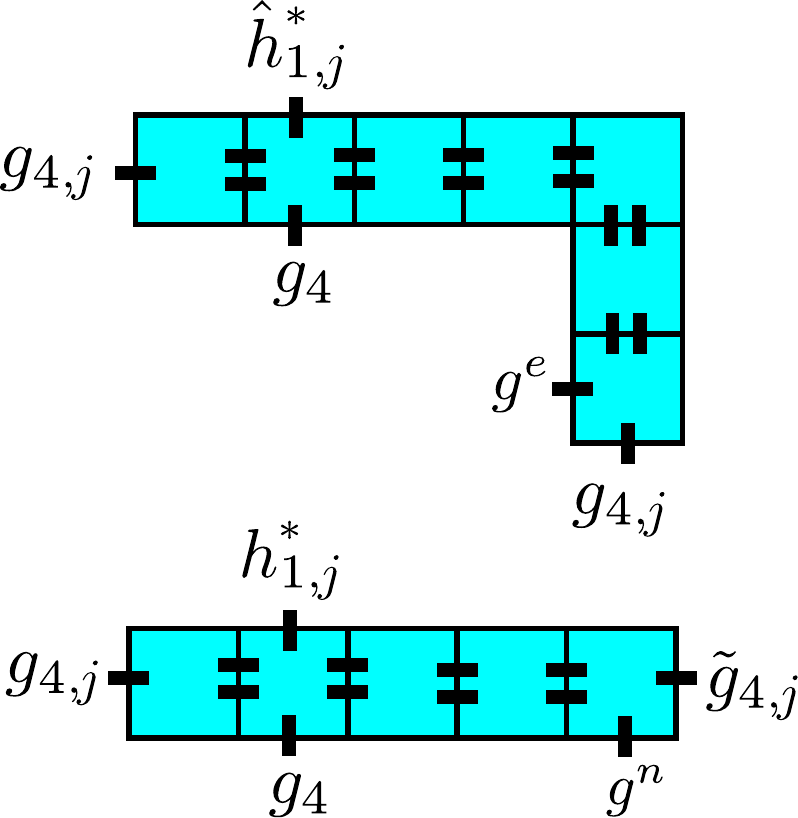}
            \caption{The supertiles that start the growth of \grout\ for $C^s_4$ for $s\geq 2$.}\label{fig:4-init}
          \end{subfigure}    
          \qquad\qquad
\begin{subfigure}{0.25\textwidth}
            \centering
            \includegraphics[width=\textwidth]{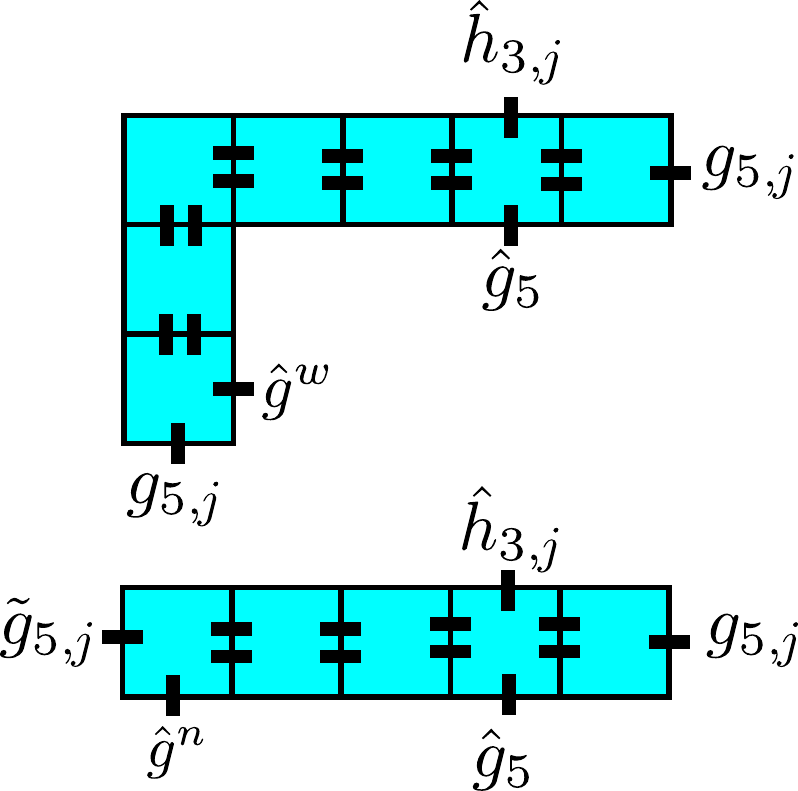}
            \caption{The supertiles that start the growth of \grout\ for $C^s_5$ for $s\geq 2$.}\label{fig:5-init}
          \end{subfigure}
       \caption{\starter\ tile types}\label{fig:inits345}
\end{figure}

\begin{figure}[htp]
    \centering
\begin{subfigure}{0.25\textwidth}
            \centering
            \includegraphics[width=\textwidth]{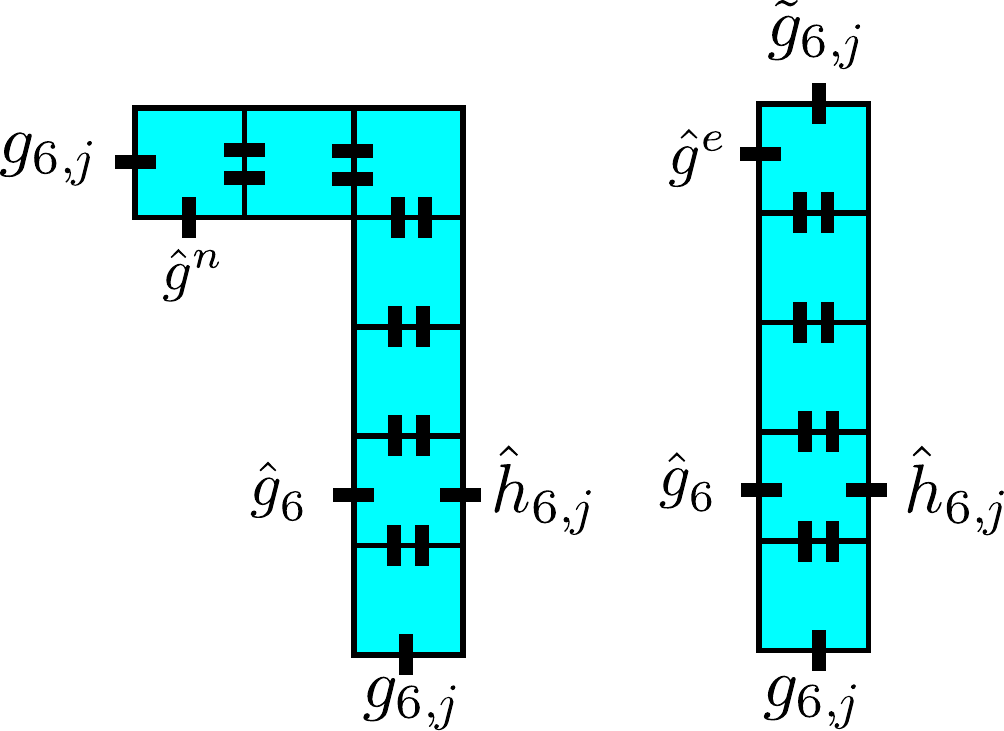}
            \caption{The supertiles that start the growth of \grout\ for $C^s_6$ for $s\geq 2$.}\label{fig:6-init}
          \end{subfigure}
                \qquad\qquad
\begin{subfigure}{0.25\textwidth}
            \centering
            \includegraphics[width=\textwidth]{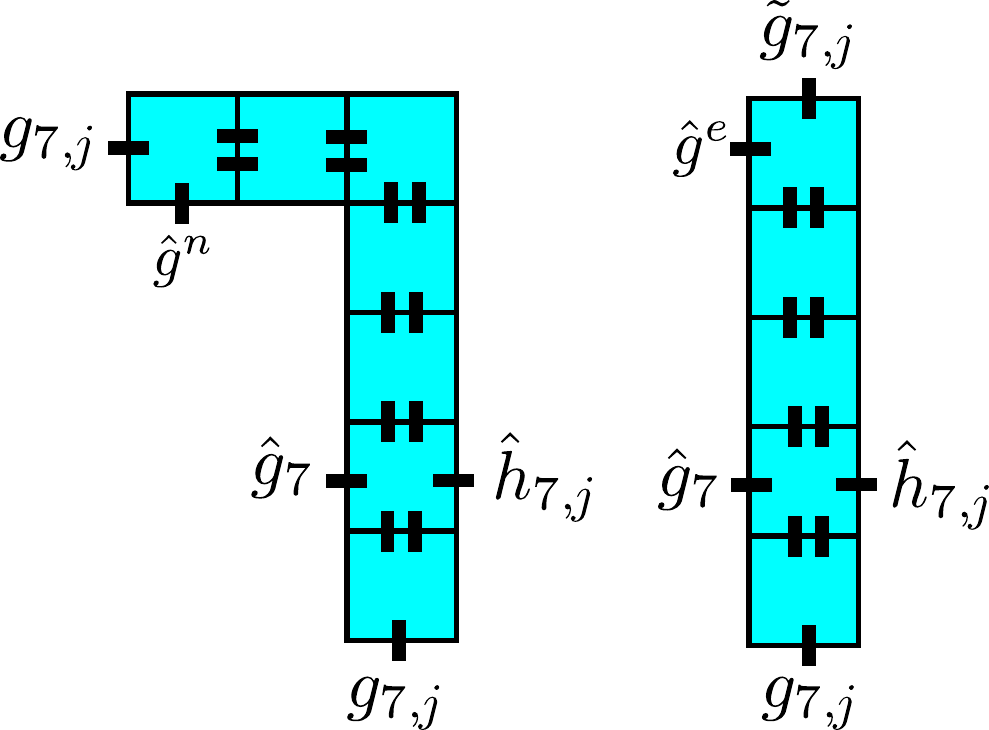}
            \caption{The supertiles that start the growth of \grout\ for $C^s_7$ for $s\geq 2$.}\label{fig:7-init}
          \end{subfigure}    
          \qquad\qquad
\begin{subfigure}{0.25\textwidth}
            \centering
            \includegraphics[width=\textwidth]{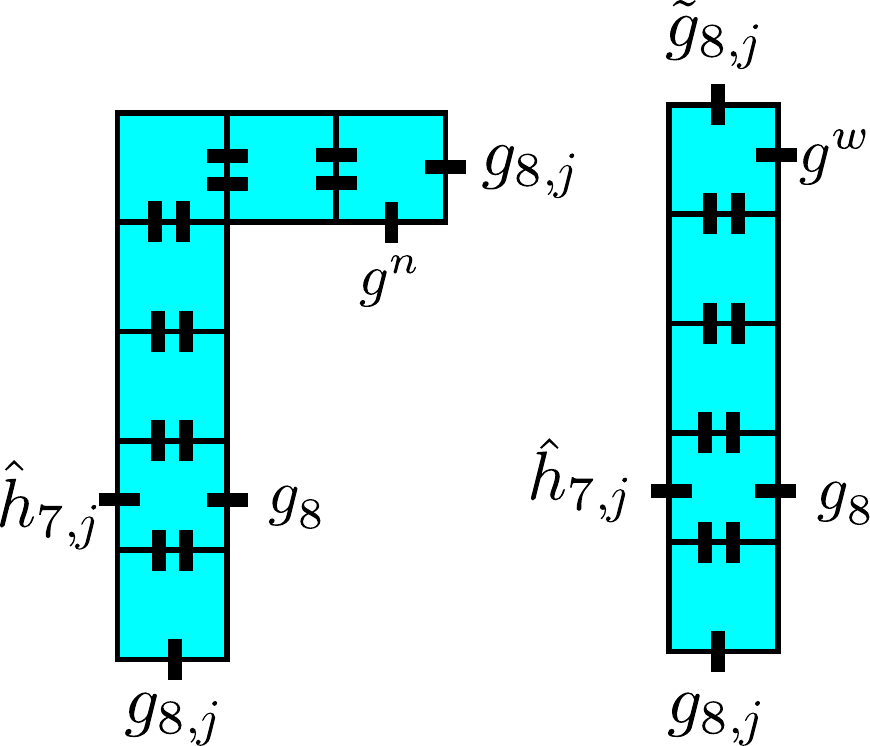}
            \caption{The supertiles that start the growth of \grout\ for $C^s_8$ for $s\geq 2$.}\label{fig:8-init}
          \end{subfigure}
      
    \caption{More \starter\ tile types}\label{fig:inits678}
\end{figure}

\subsection{\crawler\ tile types}

\begin{figure}[htp]
              \centering
              \includegraphics[width=0.6\textwidth]{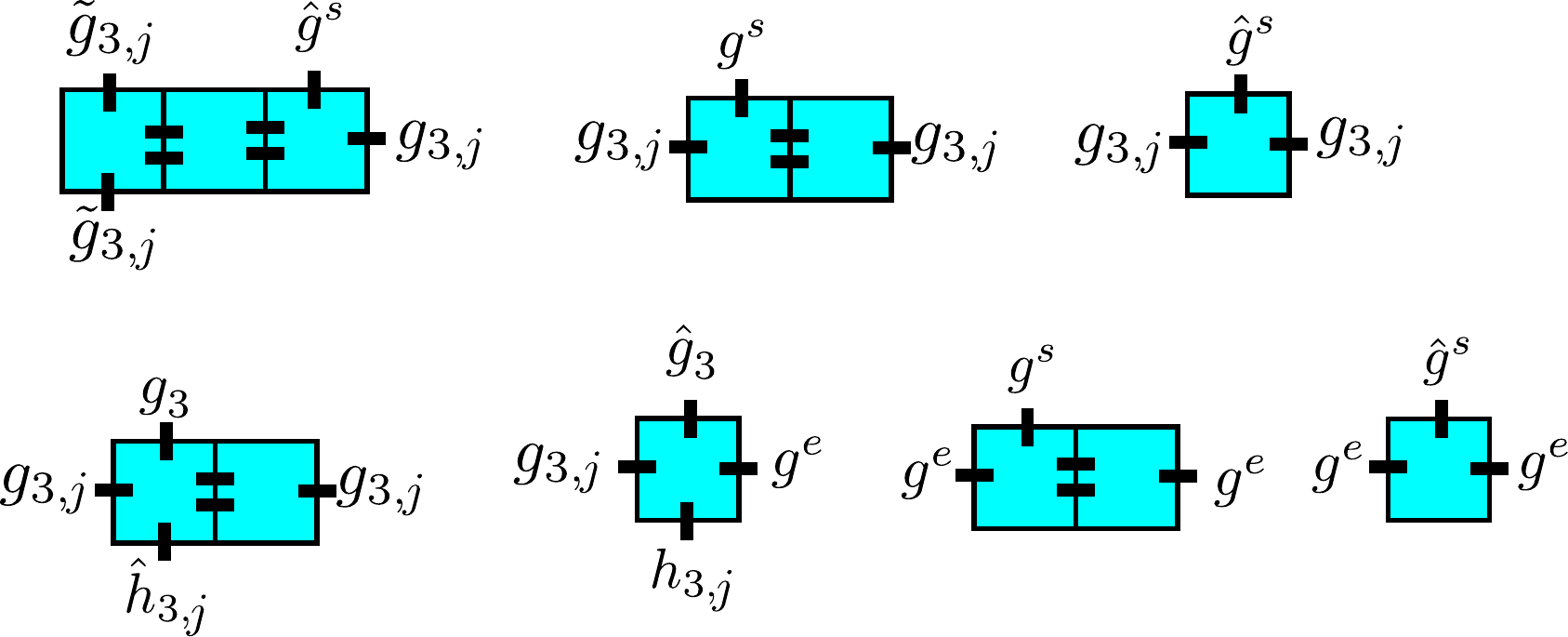}
              \caption{The tiles and supertiles that bind to the south side of $C^s_3$ for $s\geq 2$.}\label{fig:3-south}
            \end{figure}
\begin{figure}[htp]
              \centering
              \includegraphics[width=0.5\textwidth]{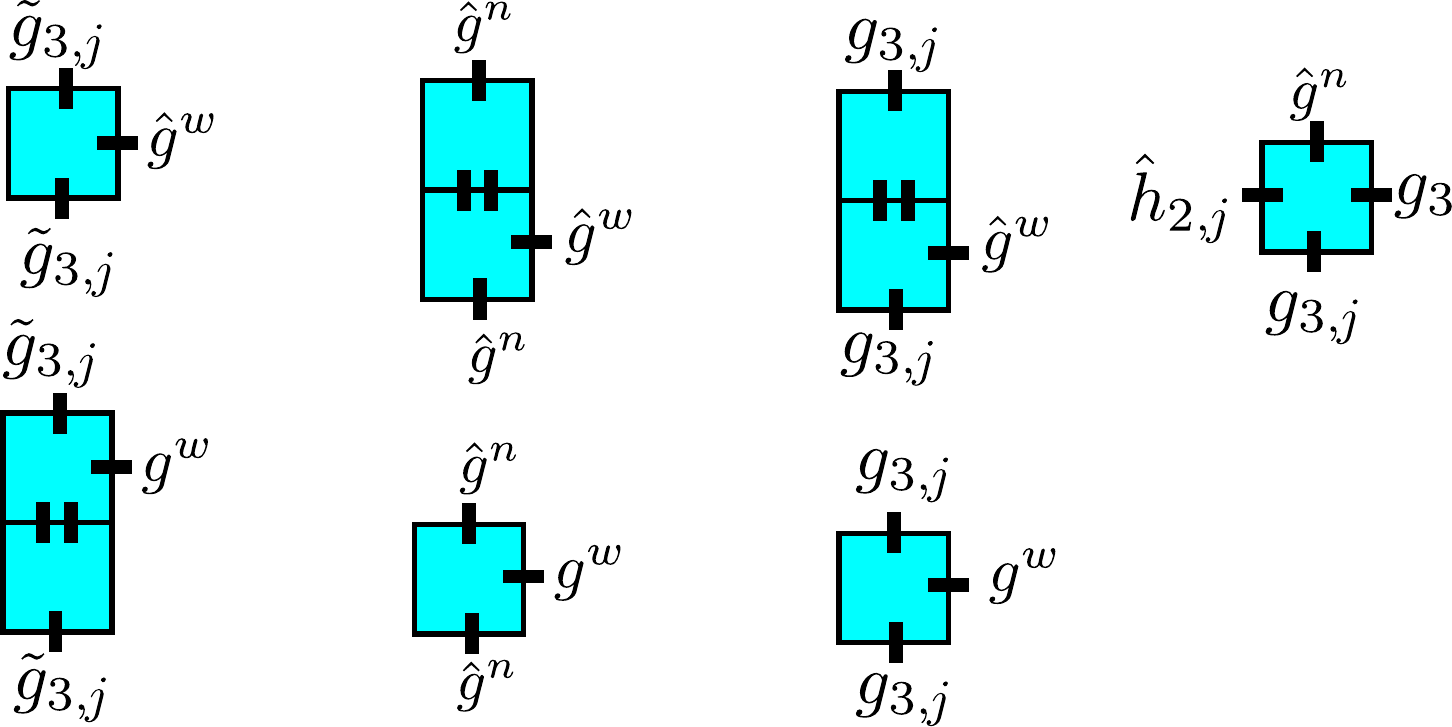}
              \caption{The tiles and supertiles that bind to the west side of $C^s_3$ for $s\geq 2$.}\label{fig:3-west}
            \end{figure}

\begin{figure}[htp]
              \centering
              \includegraphics[width=0.6\textwidth]{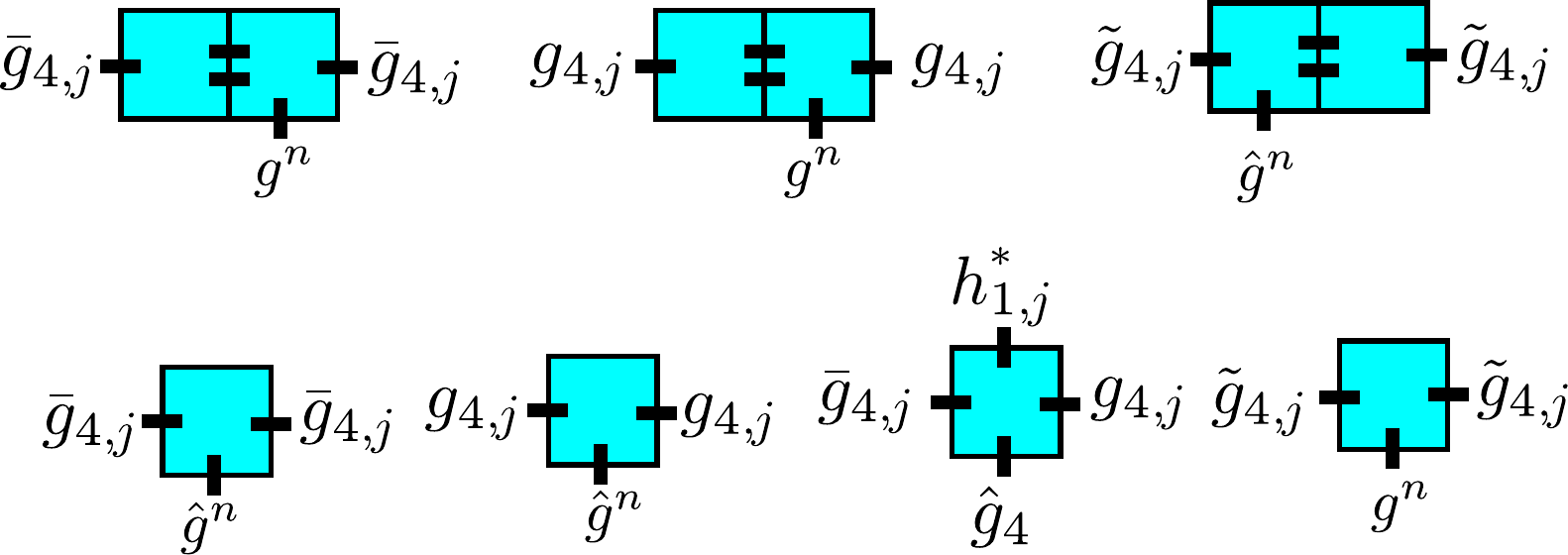}
              \caption{The tiles and supertiles that bind to the north side of $C^s_4$ for $s\geq 2$.}\label{fig:4-north}
            \end{figure}
\begin{figure}[htp]
              \centering
              \includegraphics[width=0.3\textwidth]{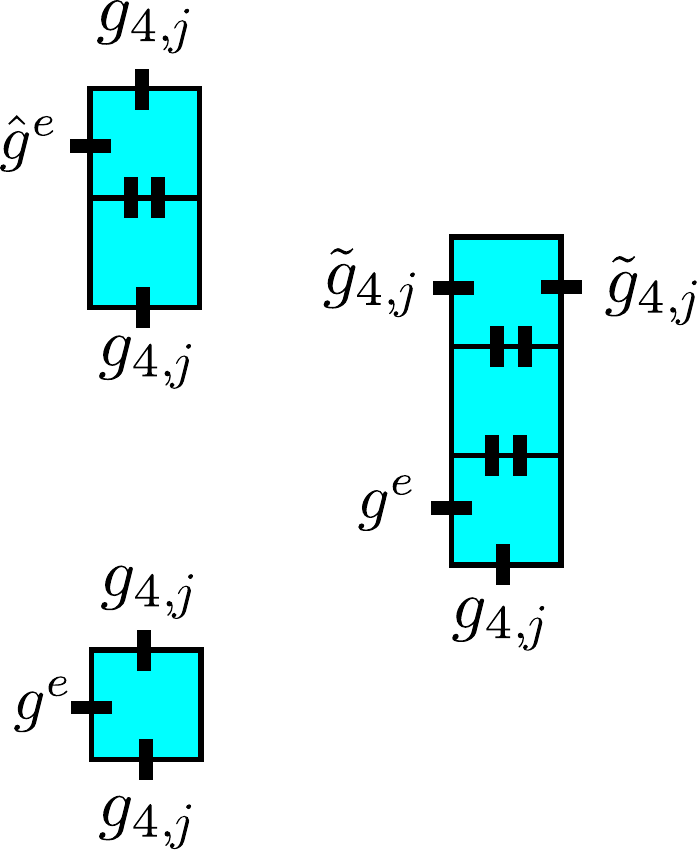}
              \caption{The tiles and supertiles that bind to the east side of $C^s_4$ for $s\geq 2$.}\label{fig:4-east}
            \end{figure}
\begin{figure}[htp]
              \centering
              \includegraphics[width=0.6\textwidth]{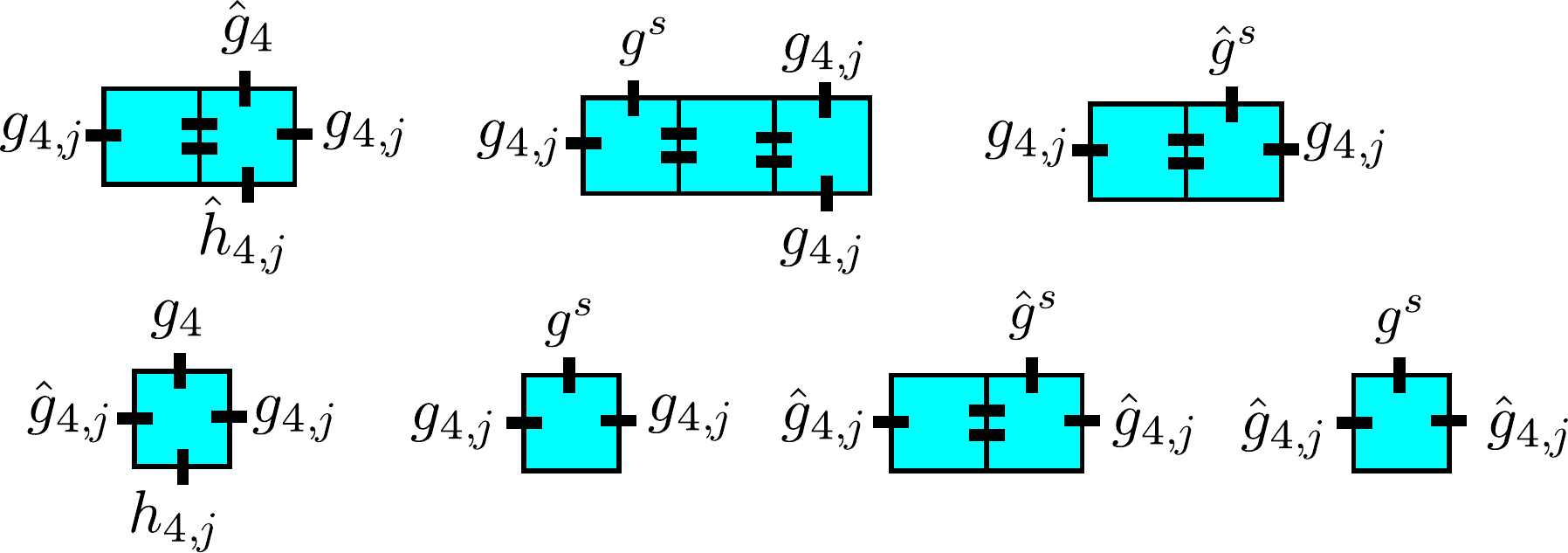}
              \caption{The tiles and supertiles that bind to the south side of $C^s_4$ for $s\geq 2$.}\label{fig:4-south}
            \end{figure}

\begin{figure}[htp]
              \centering
              \includegraphics[width=0.6\textwidth]{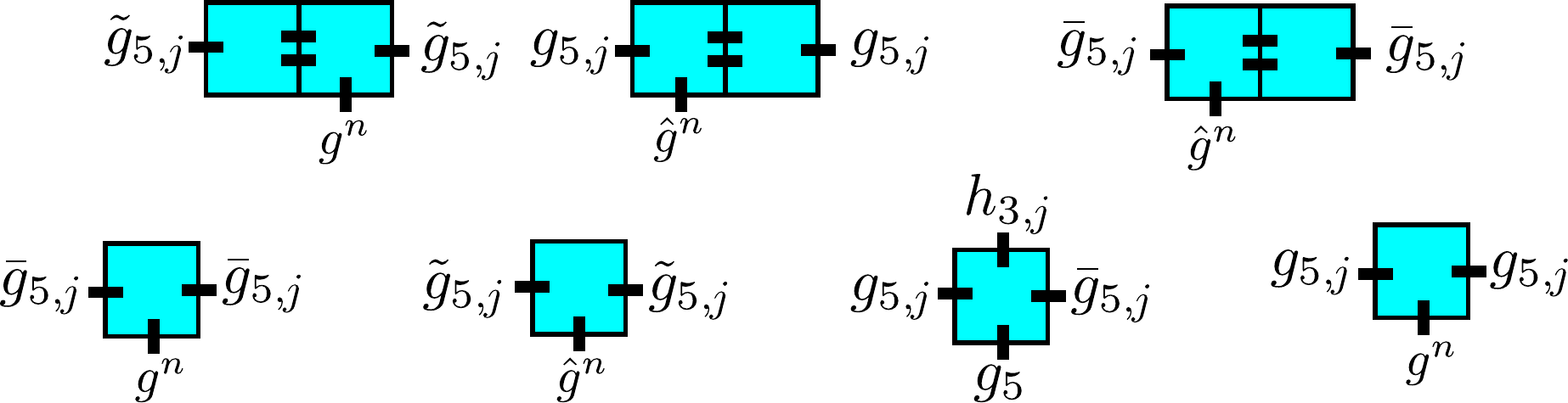}
              \caption{The tiles and supertiles that bind to the north side of $C^s_5$ for $s\geq 2$.}\label{fig:5-north}
            \end{figure}
\begin{figure}[htp]
              \centering
              \includegraphics[width=0.6\textwidth]{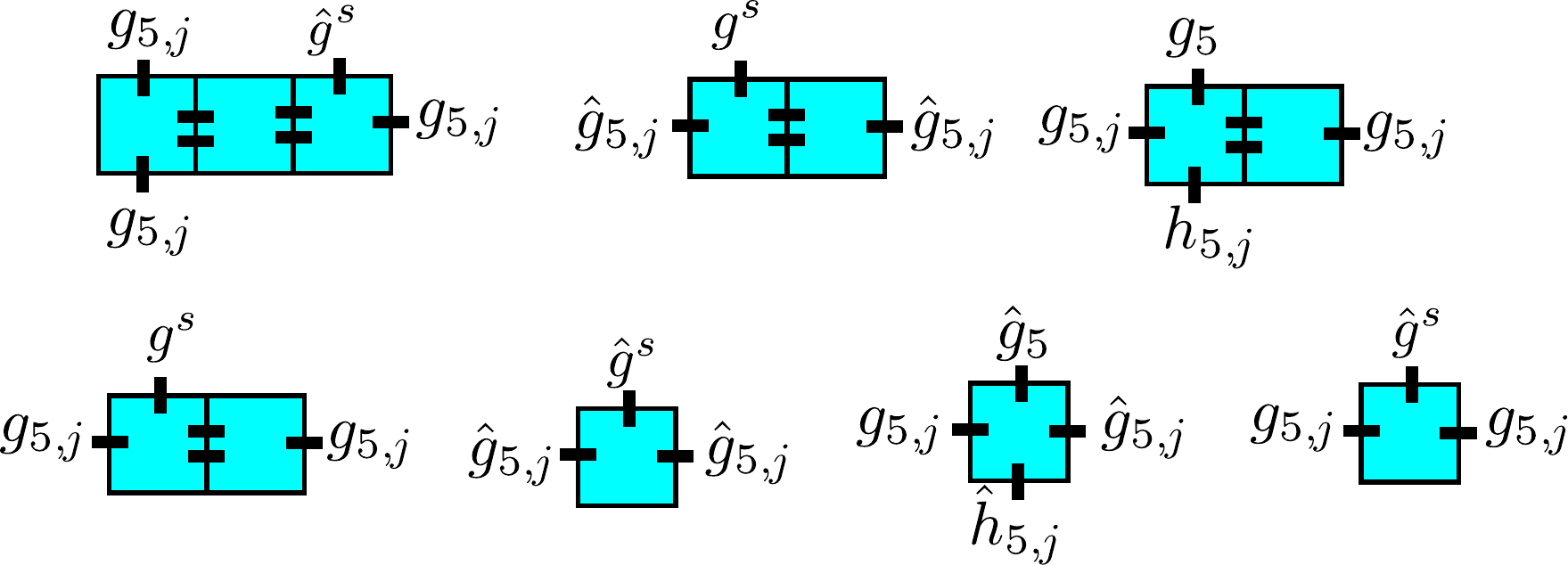}
              \caption{The tiles and supertiles that bind to the south side of $C^s_5$ for $s\geq 2$.}\label{fig:5-south}
            \end{figure}
\begin{figure}[htp]
              \centering
              \includegraphics[width=0.3\textwidth]{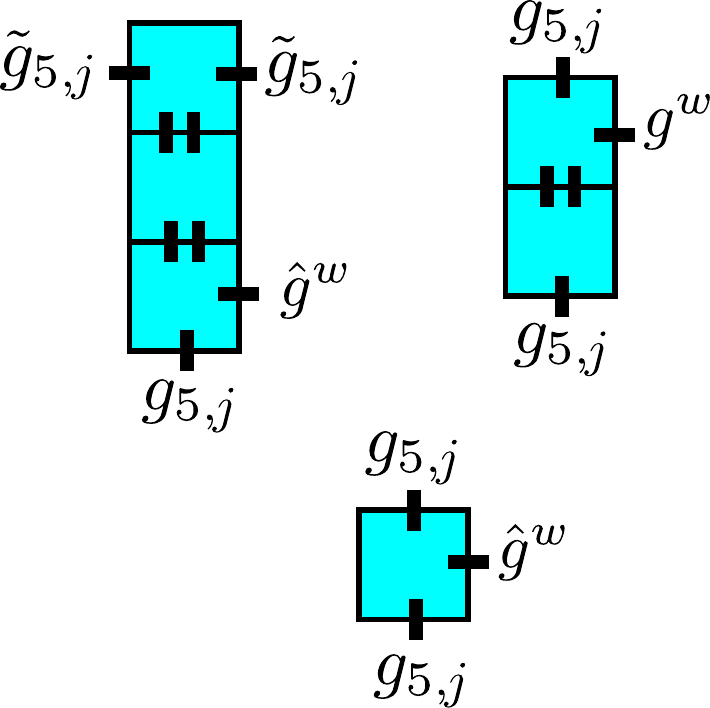}
              \caption{The tiles and supertiles that bind to the west side of $C^s_5$ for $s\geq 2$.}\label{fig:5-west}
            \end{figure}

\begin{figure}[htp]
              \centering
              \includegraphics[width=0.6\textwidth]{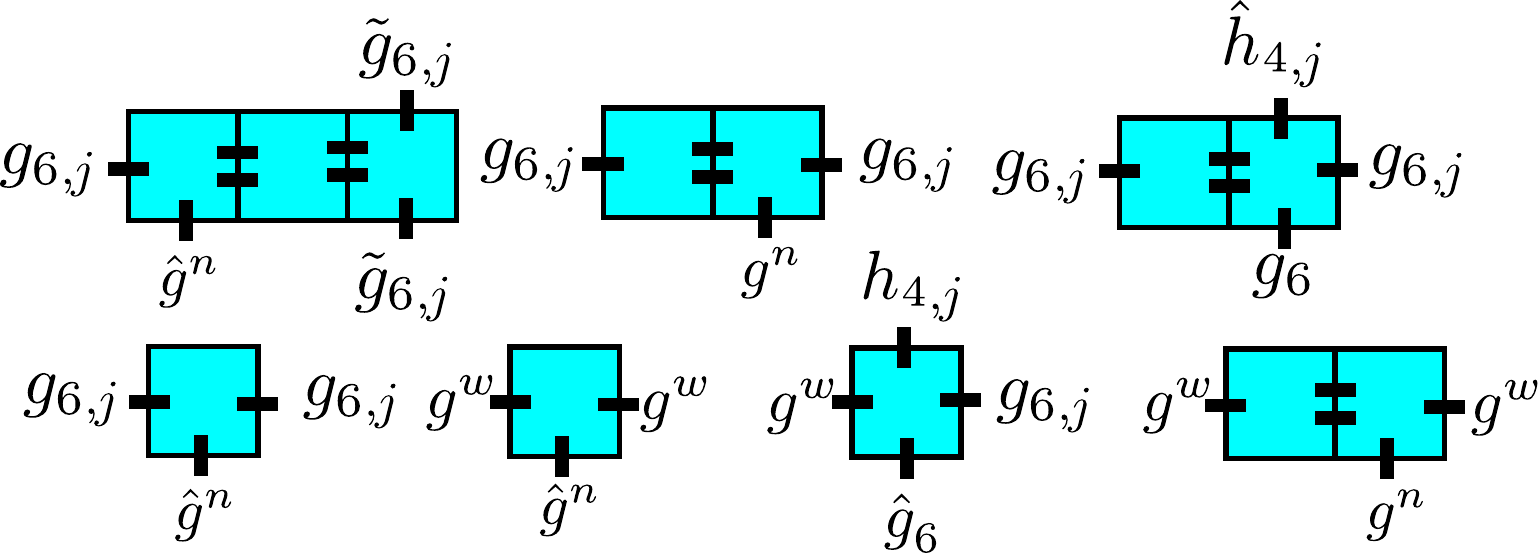}
              \caption{The tiles and supertiles that bind to the north side of $C^s_6$ for $s\geq 2$.}\label{fig:6-north}
            \end{figure}
\begin{figure}[htp]
              \centering
              \includegraphics[width=0.4\textwidth]{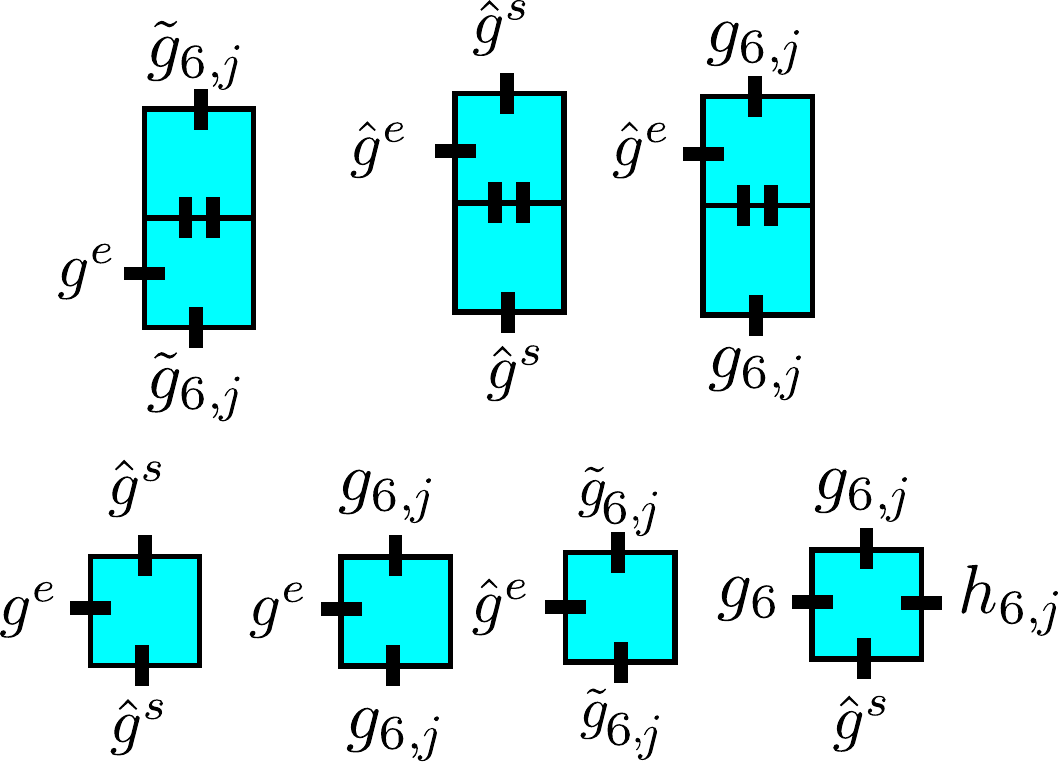}
              \caption{The tiles and supertiles that bind to the east side of $C^s_6$ for $s\geq 2$.}\label{fig:6-east}
            \end{figure}

\begin{figure}[htp]
              \centering
              \includegraphics[width=0.4\textwidth]{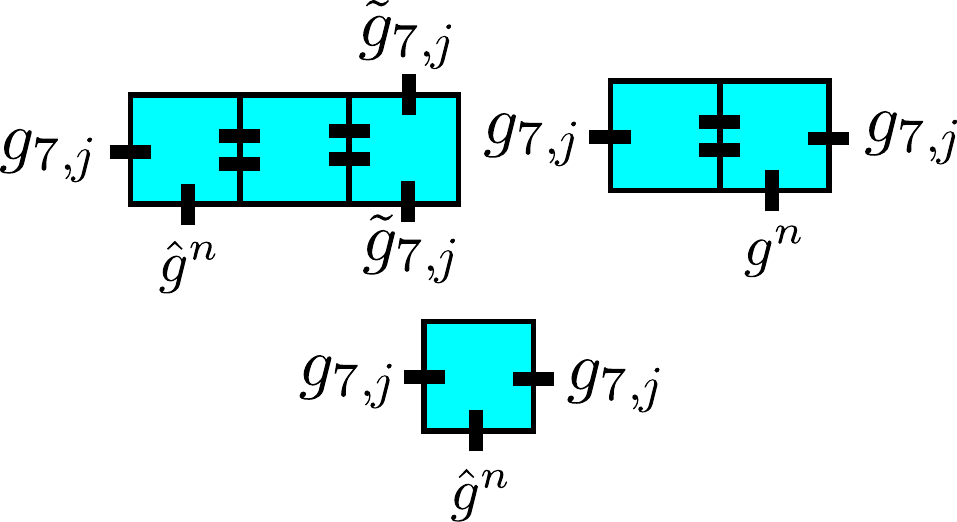}
              \caption{The tiles and supertiles that bind to the north side of $C^s_7$ for $s\geq 2$.}\label{fig:7-north}
            \end{figure}
\begin{figure}[htp]
              \centering
              \includegraphics[width=0.4\textwidth]{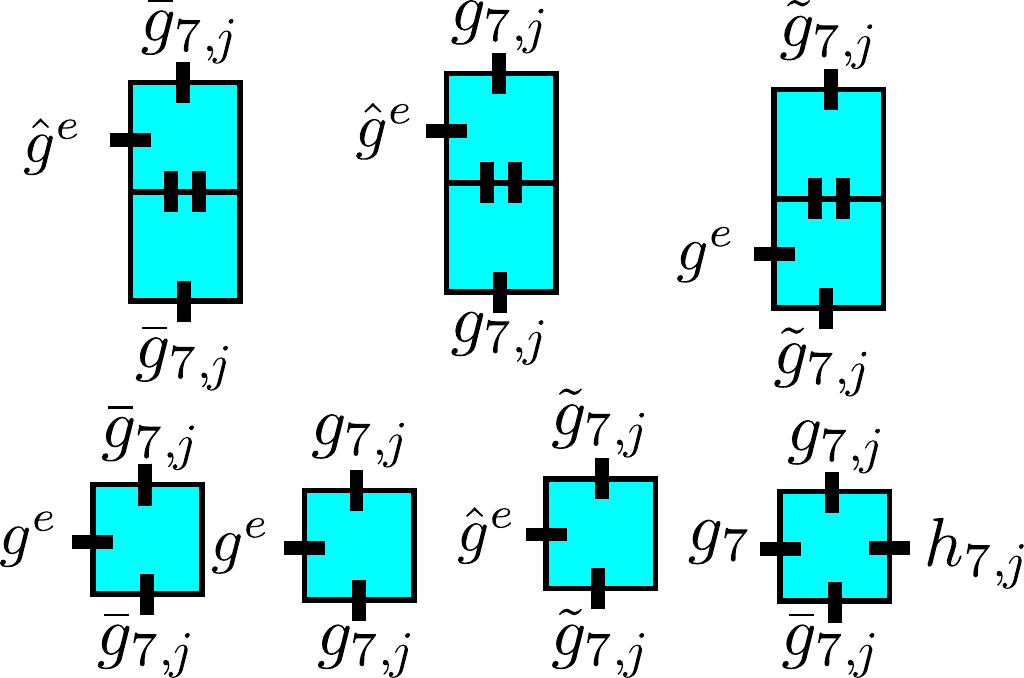}
              \caption{The tiles and supertiles that bind to the east side of $C^s_7$ for $s\geq 2$.}\label{fig:7-east}
            \end{figure}
\begin{figure}[htp]
              \centering
              \includegraphics[width=0.5\textwidth]{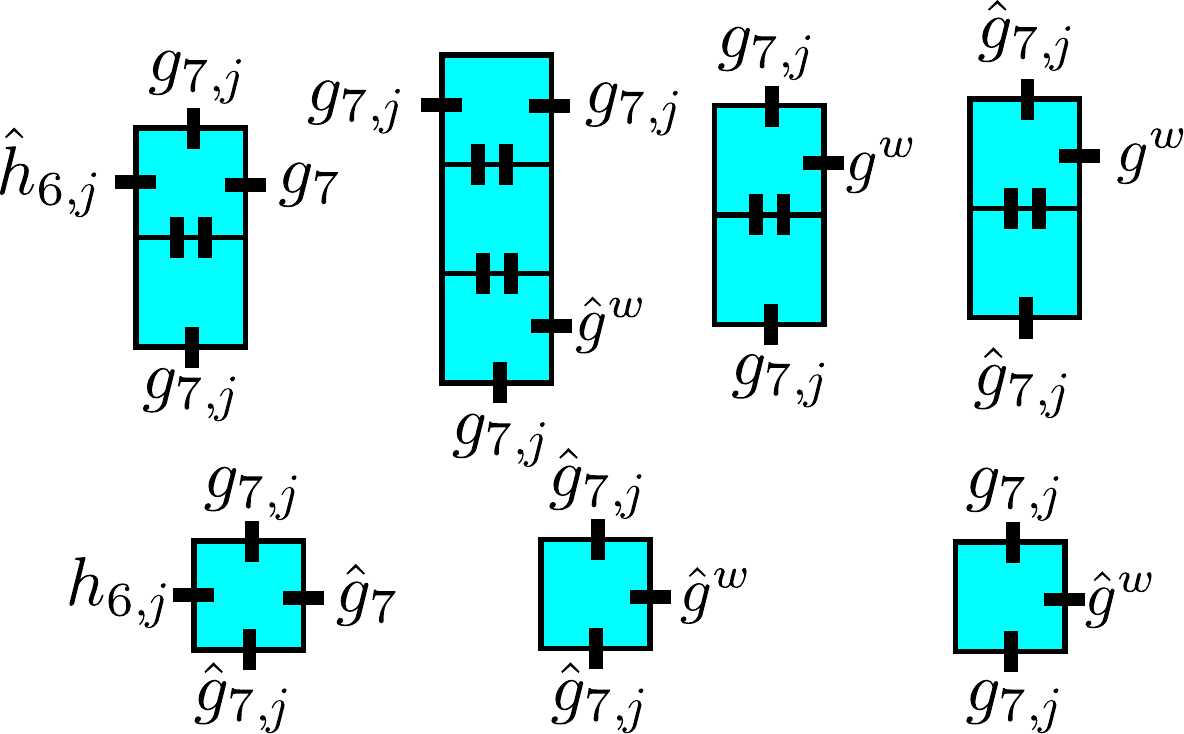}
              \caption{The tiles and supertiles that bind to the west side of $C^s_7$ for $s\geq 2$.}\label{fig:7-west}
            \end{figure}

\begin{figure}[htp]
              \centering
              \includegraphics[width=0.6\textwidth]{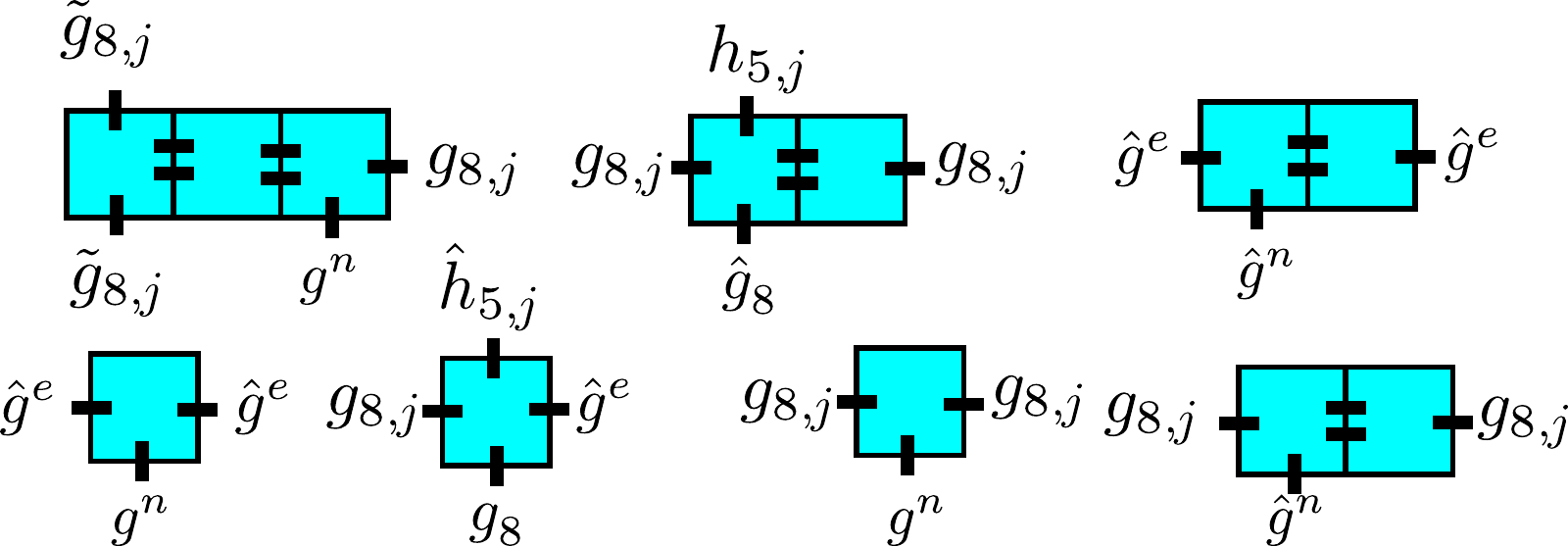}
              \caption{The tiles and supertiles that bind to the north side of $C^s_8$ for $s\geq 2$.}\label{fig:8-north}
            \end{figure}
\begin{figure}[htp]
              \centering
              \includegraphics[width=0.4\textwidth]{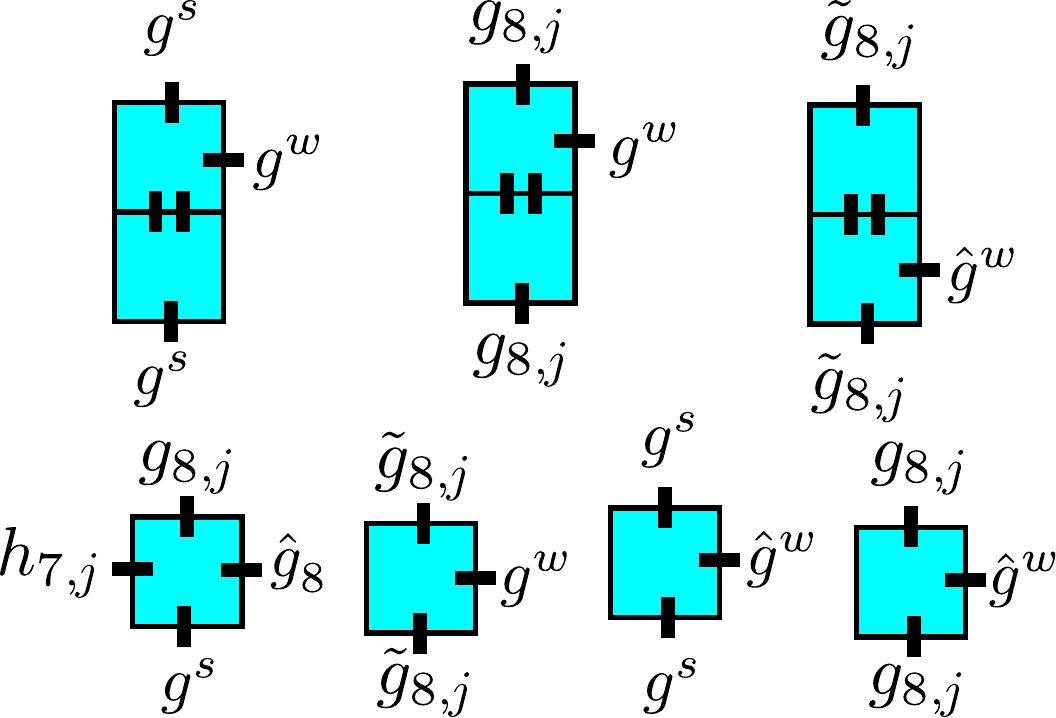}
              \caption{The tiles and supertiles that bind to the west side of $C^s_8$ for $s\geq 2$.}\label{fig:8-west}
            \end{figure}

\end{appendices}
\fi

\end{document}